%% file: combined.tex
\documentclass[journal]{IEEEtran}
\input{preamble}

\begin{document}

\title{Securing Autonomous Vehicle Systems via Twin-Aware Federated Reinforcement Learning}

\author{
		\IEEEauthorblockN{
        Zifan Zhang$\IEEEauthorrefmark{1}$,
		Minghong Fang$\IEEEauthorrefmark{2}$,
        Dianwei Chen$\IEEEauthorrefmark{3}$,
        Zhuqing Liu$\IEEEauthorrefmark{4}$,\\
        Prashant Khanduri$\IEEEauthorrefmark{5}$,
        Xianfeng Yang$\IEEEauthorrefmark{1}$,
        Anupam Das$\IEEEauthorrefmark{1}$,
        Yuchen Liu$\IEEEauthorrefmark{1}$
 }

		\IEEEauthorblockA{
		$\IEEEauthorrefmark{1}$North Carolina State University, USA,
		$\IEEEauthorrefmark{2}$University of Louisville, USA, \\
        $\IEEEauthorrefmark{3}$University of Maryland, USA,
        $\IEEEauthorrefmark{4}$University of North Texas, USA,
        $\IEEEauthorrefmark{5}$Wayne State University, USA
		}

	}

\maketitle
\begin{abstract}
Federated reinforcement learning (FRL) is crucial for enabling collaborative learning across multiple agents without sharing raw data, thereby enhancing privacy and scalability in the decision-making process within dynamic vehicular environments.
However, poisoning attacks pose a significant threat to the security and reliability of FRL-based systems, particularly in safety-critical autonomous driving, where this vulnerability remains largely unexplored.
These attacks can compromise the global control model by subtly injecting malicious system parameters, leading to potential hazards.
To counter these challenges, we present \alg (\underline{Sec}ure \underline{A}ggregation with \underline{p}oisoning-\underline{p}revention and historical reinforcement) as a defensive framework aimed at enhancing the robustness of FRL systems designed for safety-critical driving scenarios.
\alg strategically integrates digital twins for rehearsal-based learning and leverages historical aggregated model parameters along with a selected central gradient to ensure that only benign data is aggregated, effectively mitigating the influence of malicious agents.
Theoretical guarantees are provided for the convergence performance of \alg in the presence of poisoning attacks.
We also validate the effectiveness of \alg using developed digital twins that model realistic highway environments to evaluate the control of autonomous vehicles under adversarial conditions.
\end{abstract}

\input{sec/intro}

\input{sec/motivation}
\input{sec/attack}

\input{sec/defense}

\input{sec/theoretical}
\input{sec/exp}

\input{sec/discussion}

\input{sec/related}
\input{sec/conclusion}

\input{sec/appendix}

\bibliographystyle{IEEEtran}
\bibliography{ref}

\end{document}

%% file: preamble.tex
\IEEEoverridecommandlockouts
\usepackage{cite}
\usepackage{amsmath,amssymb,amsfonts}
\usepackage{graphicx}
\usepackage{textcomp}
\usepackage{xcolor}

\usepackage{times}
\usepackage{epsfig}
\usepackage{amsmath}
\usepackage{amssymb}
\usepackage{algorithm}
\usepackage{algpseudocode}
\usepackage{booktabs}
\usepackage{makecell}
\usepackage{caption}
\usepackage{subcaption}
\usepackage{balance}

\usepackage{indentfirst}
\usepackage{multirow}
\usepackage{rotating}
\usepackage{array}
\usepackage{makecell}
\usepackage{tabularx}

\newcolumntype{C}[1]{>{\centering\arraybackslash}m{#1}}
\newcolumntype{Y}{>{\centering\arraybackslash}X}

\usepackage{bm}
\usepackage{amssymb}
\usepackage{pifont}

\usepackage{float}

\usepackage{amsmath,amsthm,amssymb}
\newtheorem{theorem}{Theorem}

\newtheorem{lemma}{Lemma}

\newtheorem{assumption}{Assumption}
\theoremstyle{remark}
\theoremstyle{claim}

\newtheorem*{remark}{Remark}

\renewcommand{\theequation}{\arabic{equation}}
\theoremstyle{remark}

\usepackage[caption=false, font=normalsize, labelfont=sf, textfont=sf]{subfig}

\newcommand{\myparatight}[1]{\smallskip\noindent{\bf {#1}:}~}

\usepackage{xspace}
\newcommand{\alg}{\textsf{SecApp}\xspace}
\newcommand{\algns}{{\textsf{SecApp}}}

\usepackage{cleveref}
\crefname{axiom}{axiom}{axioms}
\crefname{definition}{definition}{definitions}
\crefname{lemma}{lemma}{lemmata}

\def\BibTeX{{\rm B\kern-.05em{\sc i\kern-.025em b}\kern-.08em
    T\kern-.1667em\lower.7ex\hbox{E}\kern-.125emX}}

%% file: sec/intro.tex
\section{Introduction}

Advancements in computing power have significantly improved decision-making and problem-solving. Among various approaches, reinforcement learning (RL) has proven effective in real-world mobile applications, including robotics~\cite{haarnoja2024learning
}, autonomous driving (AD)~\cite{he2023fear, chen2024deep}, GPT-4~\cite{openai2024gpt4technicalreport}, wireless networks~\cite{10605806, 9329087
}, and healthcare~\cite{yang2023algorithmic}.
In essence, agents on mobile entities or controllers receive observations to perceive the dynamics of their environments and train RL models to perform actions that maximize long-term cumulative rewards. 
However, single-agent RL suffers from low sampling efficiency and limited observability in mobile environments~\cite{levine2020offline}. In safety-critical scenarios like AD, such agents often struggle to find optimal solutions for a wide range of operational conditions, especially in handling accident cases. Overlooking these sparsely distributed corner cases can lead to severe consequences, including threats to human safety. 
%
As a remedy, federated reinforcement learning (FRL) is employed by combining the strengths of federated learning (FL) and RL, enabling multiple distributed agents to collaboratively learn optimal policies from a holistic view while preserving the privacy of their individual data. 
FRL has been widely used in various domains~\cite{zhou2023digital, 9945653, jin2024embracing, lee2020federated}, which enhances efficiency and generalization of decision-making models by leveraging diverse data sources from different perspectives of agents in dynamic environments. 


While FRL offers significant advantages, it presents several risks when applied in safety-critical mobile systems~\cite{qi2021federated}. 
The most critical challenge is the vulnerability to poisoning attacks, where malicious agents deliberately provide incorrect or misleading update signals, compromising the integrity of the global control for decision-making~\cite{fang2020local, cao2020fltrust, zhang2024poisoning, shejwalkar2021manipulating}. 
These attacks can significantly degrade the performance of FRL-based systems and potentially lead to failures in completing critical tasks. 
Although prior research has studied several defensive strategies to mitigate the impact of malicious agents~\cite{cao2020fltrust, fan2021fault, fung2018mitigating, li2024backdoorindicator, sharma2023flair, xia2019faba, mcmahan2017communication, blanchard2017machine}, these defense mechanisms are primarily designed for pure FL training and do not account for the unique attributes of RL agents. 
The primary goal in those FL systems is to aggregate models trained on labeled datasets with the aim of minimizing a predefined loss function.
In contrast, FRL systems focus on learning optimal policies through interactions with networked environments, involving sequential decision-making and delayed rewards.
This introduces additional complexities, such as non-stationarity and the need to balance exploration and exploitation, which are absent in traditional FL scenarios. Besides, RL agents operate in diverse and evolving environments, resulting in more significant variations in local gradients (i.e. the learning updates based on their respective experiences before aggregation) compared to the relatively stable updates in pure FL systems.
Our experiments showed that applying FL security mechanisms to the FRL framework yields unsatisfactory performance in AD decision-making tasks (detailed in Sec. II).
More importantly, there is a lack of comprehensive theoretical analysis regarding the convergence of secure FRL optimization~\cite{xu2020improved}, which makes their application in safety-critical systems a long-standing matter requiring prudent consideration. 
%
These issues underscore the need to develop new, specialized solutions that integrate both environmental \textit{security awareness} and \textit{theoretical guarantees} against potential attacks in FRL-enabled mobile systems.


\myparatight{Main Contribution}
To bridge this gap between secure computing and robust machine learning in safety-critical mobile systems, we propose a novel robustness-aware framework called \alg -- \underline{Sec}ure \underline{A}ggregation with \underline{p}oisoning-\underline{p}revention and historical reinforcement. 
The core objective is to protect AD system from poisoning attacks, where malicious entities attempt to corrupt data-driven control model in each mobile entity by sending misleading or harmful information pieces. 
In essence, \alg employs a multi-step strategy to mitigate the effects of poisoning attacks with varying levels of severity. Initially, central server filters out the gradients from each mobile agent (e.g., deployed on vehicles) that differ from the majority, ensuring only consistent and harmless gradients are included in the aggregation set. 
The server then identifies a central gradient to represent the benign mobile agents and aggregates only those gradients that are close to it. 
This process is repeated across multiple rounds, allowing the system to iteratively refine the global model for decision-making while minimizing impacts of malicious agents.
In addition, we develop a digital twin (DT) of the system environment by incorporating historical sensor inputs and operational parameters into a cohesive data-driven model.
With the DT in the loop, the system can generate diverse environments that reflect a wide range of operating conditions and support rehearsal-based learning. 
This diversity enables robust validation of decision-making algorithms, allowing the system to adapt to various scenarios without the risks of physical deployment, i.e. from normal operations to corner cases.

\myparatight{Theory Achievement} Theoretically, we establish that under certain mild assumptions, such as the smoothness and non-convexity of the loss function, the decision-making model learned by our \alg is proven to converge to a first-order stationary point of the learning objective, which is the strongest guarantee attainable for the highly non-convex RL loss. This outcome underscores the robustness of \algns, ensuring that it preserves the reliable performance of the FRL-enabled mobile systems. This property remains valid even in the presence of Byzantine settings, demonstrating that \alg can effectively handle malicious or faulty updates without compromising performance.

\myparatight{Evaluation Remarks} To illustrate the effectiveness of \algns, we apply it to real-world safety-critical systems, including AD and edge caching scenarios\footnote{Additional experiments on mobile edge caching settings can be found in the supplementary materials (Appendix E).}.
For the AD task, conducting experiments on actual highways is both impractical and unsafe. To address this, we exploit the developed DT, named \textit{HighwayDT}~\cite{NSFOAC_GitHub}, which accurately simulates realistic safety-critical scenarios, including rare and complex cases, by replicating the dynamic conditions of highway AD environments.
%
%
HighwayDT receives real-time data from the physical vehicular network system, continuously updating its control model to accurately reflect the evolving attributes of highway settings. It integrates functionalities such as predictive vehicle control, scenario generation, and decision-making, enabling a comprehensive validation process within a synchronized virtual environment.
For instance, \alg is deployed on autonomous vehicles to learn optimal solutions for tasks such as vehicular longitudinal control and collision avoidance.
To assess its robustness, \alg is evaluated under thirteen poisoning attacks, including strong adaptive attack, and demonstrates superior performance compared to state-of-the-art defense strategies. The results also align closely with the theoretical analysis of \algns’s convergence, highlighting its consistency and reliability in challenging adversarial scenarios.



Our contributions can be summarized as follows:
\begin{list}{\labelitemi}{\leftmargin=1em \itemindent=-0.08em \itemsep=.2em}
    \item We introduce \algns, the first robustness-aware framework specifically designed to safeguard FRL-based safety-critical systems against the impact of poisoning attacks.

    \item We provide the first-of-its-kind comprehensive theoretical analysis that demonstrates the convergence performance of secure FRL in the presence of poisoning attacks. This analysis confirms that \alg can maintain its stability and effectiveness even when facing adversarial conditions.

    \item We construct HighwayDT to precisely map physical safety-critical mobile environments into digital FRL systems, which can provide risk-free and highly heterogeneous validation scenarios for security tests and analysis. 
    \item We conduct extensive experiments in safety-critical scenarios, including AD on highways and edge caching in high-stakes network infrastructures. The evaluation results demonstrate the effectiveness of \alg against existing state-of-the-art targeted and untargeted attacks. 

\end{list}

\noindent \textbf{Notations:} 
In this work, 
notation $[K]$ denotes the set \\ $\{1, 2, \dots, K\}$ for any positive integer $K$. The $\ell_2$-norm is indicated by $\left\| \cdot \right\|$, and the size of a set $\mathcal{S}$ is written as $|\mathcal{S}|$.

%% file: sec/motivation.tex
\section{Motivation and Observation}

Through preliminary experiments, we observe several limitations in existing learning-based control approaches for mobile autonomous vehicle (AV) systems. These observations motivate the need for \alg frameworks that better account for system heterogeneity, interaction dynamics, and security constraints.


\textbf{\textit{Observation 1. Single-agent RL framework struggles to find effective control solutions for AV systems.}} RL has been widely adopted in the research area of AD~\cite{kiran2021deep, wu2022uncertainty}, and used by several enterprises as well for practical implementation, such as Google~\cite{lu2023imitationenoughrobustifyingimitation} and Waymo~\cite{kendall2018learningdriveday}.
In this paper, we use AV system as an example, but the fundamentals of agents and environments can be easily extended to other safety-critical systems, as demonstrated in Appendix E.
Modeled as a Markov Decision Process (MDP)~\cite{sutton2018reinforcement}, the vehicular agent takes actions \(a\) in observed states \(s\), transitions to new states \(s'\), and receives rewards \(r\) based on these transitions. The vehicle's goal is to maximize cumulative rewards by learning an optimal policy \(\pi(s)\), which specifies the best driving decisions in the long run. 
RL framework used in AD involves trial and error, balancing exploration of new actions and exploitation of known rewarding actions, and utilizes techniques such as policy gradients~\cite{xu2020improved,papini2018stochastic,REINFORCE}, 
where the \textit{gradient}—representing the derivative of the expected reward with respect to policy parameters—guides updates to improve driving performance in ever-changing environments.

However, traditional single-agent RL often struggles with low sampling efficiency, particularly in AV systems, where the agent may fail to adequately explore rare but critical corner events, such as accident scenarios. This limited exploration leads to high variance in gradient-based learning, hindering stable and effective policy optimization. 
%
To tackle this problem, a variance reduction (VR) method is often used, where a roadside edge server additionally samples a few trajectories to update the global model to reduce the variance of stochastic gradients.
Meanwhile, extending from a single agent to multiple distributed agents training on AVs improves sampling efficiency, as agents collectively explore a wider state space, capture diverse experiences, and accelerate learning convergence.
Table~\ref{tab:mini-batch} presents our preliminary results that an edge server uses different aggregation rules in multi-agent distributed training (such as FedAvg~\cite{mcmahan2017communication}, Median~\cite{Yin18}, Trimmed-mean (Trim)~\cite{Yin18}, and  Krum~\cite{blanchard2017machine}) to combine agents' local gradients. ``w/o VR'' indicates that the variance reduction method is not used, where the server directly updates the global model using the aggregated gradient (the full procedure is provided in the supplementary material), while ``w/ VR'' applies variance reduction to mitigate gradient variance. 
``Single'' represents training on a single agent without shared observation and data aggregation.
Table~\ref{tab:mini-batch} presents no-collision rates as the performance metric for an AV system\footnote{The detailed evaluation and experimental setup can be found in Section~\ref{sec: exp}.}. 
The results show that single-agent RL struggles to achieve the desired no-collision rates due to limited sampling efficiency, while the multi-agent solution with VR significantly enhances vehicle control performance, even achieving perfect solutions regardless of aggregation method.

\begin{table}[h!]
\centering
\footnotesize
\caption{Comparison of single- and multi-agent training w/o and w/ variance reduction (VR). 
}
\begin{tabular}{p{1.5cm} c c c c c}
\hline
 & FedAvg & Median & Trim & Krum & Single  \\
\hline
w/o VR & 82.62\% & 84.40\% & 85.08\% &  81.41\% & 73.45\%\\
w/ VR & 100.0\% & 100.0\% & 100.0\% & 100.0\% & 89.46\% \\
\hline
\end{tabular}
\label{tab:mini-batch}
\end{table}


\textbf{\textit{Observation 2. Multi-agent FRL framework carries serious risks when deployed in practical mobile systems.}} FRL~\cite{fan2021fault,khodadadian2022federated,jin2022federated} overcomes the limitations in single-agent RL by allowing decentralized agents to work together to capture full observations of complex environments and learn optimal strategies while maintaining the privacy of their individual data. 
%
Consider an FRL system with $K$ agents. In each training round $t$, a generic process involves three steps:
\begin{list}{\labelitemi}{\leftmargin=1em \itemindent=-0.08em \itemsep=.2em}
    \item \textbf{Step I (Global model synchronization).}
    The central server distributes the current global model $\tilde{\mathbf{w}}_{t-1}$ to all participating agents. This corresponds to the operation $\mathbf{w}_{t}^{0} \leftarrow \tilde{\mathbf{w}}_{t-1}$ in Line~\ref{server_send} of Algorithm~\ref{alg:server}.

    \item \textbf{Step II (Local training).}
    Each agent $k \in [K]$ executes Algorithm~\ref{alg:local} to refine its local gradient.
    This involves sampling $\xi_t$ trajectories based on the trajectory distribution $d(\cdot | \mathbf{w}_{t}^{0})$, and computing the local gradient $\mu_{t}^{k}$ based on the distributed global model $\mathbf{w}_{t}^{0}$,
    where $d(\cdot | \mathbf{w}_{t}^{0})$ denotes the trajectory distribution induced by model $\mathbf{w}_{t}^{0}$ and this distribution changes over time.
    The local gradient $\mu_{t}^{k}$ is then uploaded to the central server.
    Note that $g(\cdot)$ in Line~\ref{alg1_line4} of Algorithm~\ref{alg:local} represents the gradient estimator.

    \item \textbf{Step III (Aggregation and global model updating).}
    The central server aggregates the received local gradients $\mu_{t}^{k}$ from all participating agents using a particular aggregation rule to produce the aggregated gradient. 
    For instance, if the server employs the FedAvg aggregation method to combine the local gradients from all agents, the aggregated gradient $\mu_t$ can be calculated as $\mu_t \leftarrow \frac{1}{K} \sum_{k=1}^{K} \mu_t^{k}$, as outlined in Line~\ref{server_avg} of Algorithm~\ref{alg:server}. The server then uses VR methods to update the global model and corresponding policy.
    %
    %
    The updated global model is then distributed to the agents for follow-up interactions and local control.
\end{list}	

\noindent Such FRL framework shows considerable promise in safety-critical systems, yet it carries serious risks when deployed in environments where safety is non-negotiable. A major concern is its susceptibility to model poisoning, where rogue participants intentionally inject erroneous or deceptive updates into the aforementioned \textbf{Step II}. This malicious interference undermines the overall reliability of the collective model in the system, often resulting in sharply diminished control performance and even critical failures. Moreover, while various protective measures have been explored for standard FL, these defenses generally fall short when applied to the more complex and dynamic context of RL agents.
Table~\ref{tab:prelim} presents the preliminary results of different training rules in a multi-agent system under FTI, MinSum, and Adaptive attacks, as studied in~\cite{zhang2024poisoning,shejwalkar2021manipulating}. 
The results indicate that FRL cannot maintain the 100\% no-collision rate in validation environments due to its vulnerability to model poisoning attacks. 
This highlights the need to augment robustness-aware strategies for FRL-based systems to prevent unexpected consequences.

\begin{table}[!htbp]
\renewcommand{\arraystretch}{1.5}
\centering
\caption{Comparison of multi-agent training under poisoning attacks with standard FL-based defensive strategies.}
\footnotesize
\begin{tabularx}{0.45\textwidth}{l|*{4}{>{\centering\arraybackslash}X}}
\toprule
Attack   & FedAvg & Median & Trim & Krum \\
\midrule
No attack   & 100.0\% & 100.0\% & 100.0\% & 100.0\% \\
\hline
FTI attack   & 12.47\% & 14.01\% & 8.36\% & 11.02\% \\
\hline
MinSum attack   & 10.29\% & 13.77\% & 100.0\% & 9.49\% \\
\hline
Adaptive attack & 14.49\% & 11.49\% & 7.11\% & 13.68\% \\
\bottomrule
\end{tabularx}
\label{tab:prelim}
\end{table}


\textbf{\textit{Observation 3. Homogeneous testing environments fail to ensure the stability and robustness required for safety-critical mobile systems.}}  Taming AD systems with safety awareness demands extensive testing and validation across a wide spectrum of operational conditions. Relying solely on homogeneous testing environments can overlook rare events that may be encountered in real-world scenarios, thereby compromising system safety. Heterogeneous testing environments are crucial as they expose the AD system to varying environmental conditions, road configurations, and unexpected interactions with surrounding AVs, ensuring that the system's performance is robust under diverse circumstances. 

DTs are high-fidelity virtual replicas of physical systems that are continuously updated with real-world data to reflect current conditions and behaviors~\cite{liu2024digitalnetworktwinsnextgeneration}. They enable detailed what-if analyses, allowing simulation of various scenarios, such as sudden changes in environmental conditions or vehicle component failures, to predict how the AD system might respond without the risks and costs of real-world experiments.

In the context of AD, FRL benefits significantly from the integration of DTs. By providing a controlled yet diverse set of virtual environments, DTs enrich the training data with a wide range of operational scenarios, which enhances the robustness and generalization of the learned policies. Unlike single-agent offline training, distributed training leverages multiple agents interacting with different instances from the DT in parallel, thereby capturing a broader spectrum of reaction variability. A roadside server is essential in this setup to coordinate knowledge updates and ensure consistency across agents, mitigating the risk of divergence or conflicting strategies that could occur in a fully decentralized system. 
A preliminary result of RL-based vehicles trained without diverse scenarios is shown in Fig.~\ref{fig:fail}. 
It is observed that the RL agent fails to provide proper maneuvers due to sparse rewards, eventually leading to crashes.
For instance, vehicles 2, 3, and 4, the three vehicles in the middle, continue to accelerate even as the leading vehicle begins to decelerate after approximately 100 steps, ultimately causing crashes as indicated by the dashed lines in Fig.~\ref{fig:fail}.
This result demonstrates that a policy trained in a single, homogeneous environment lacks stability and robustness when evaluated across diverse validation scenarios.

\begin{figure}
    \centering
    \includegraphics[width=0.74\linewidth]{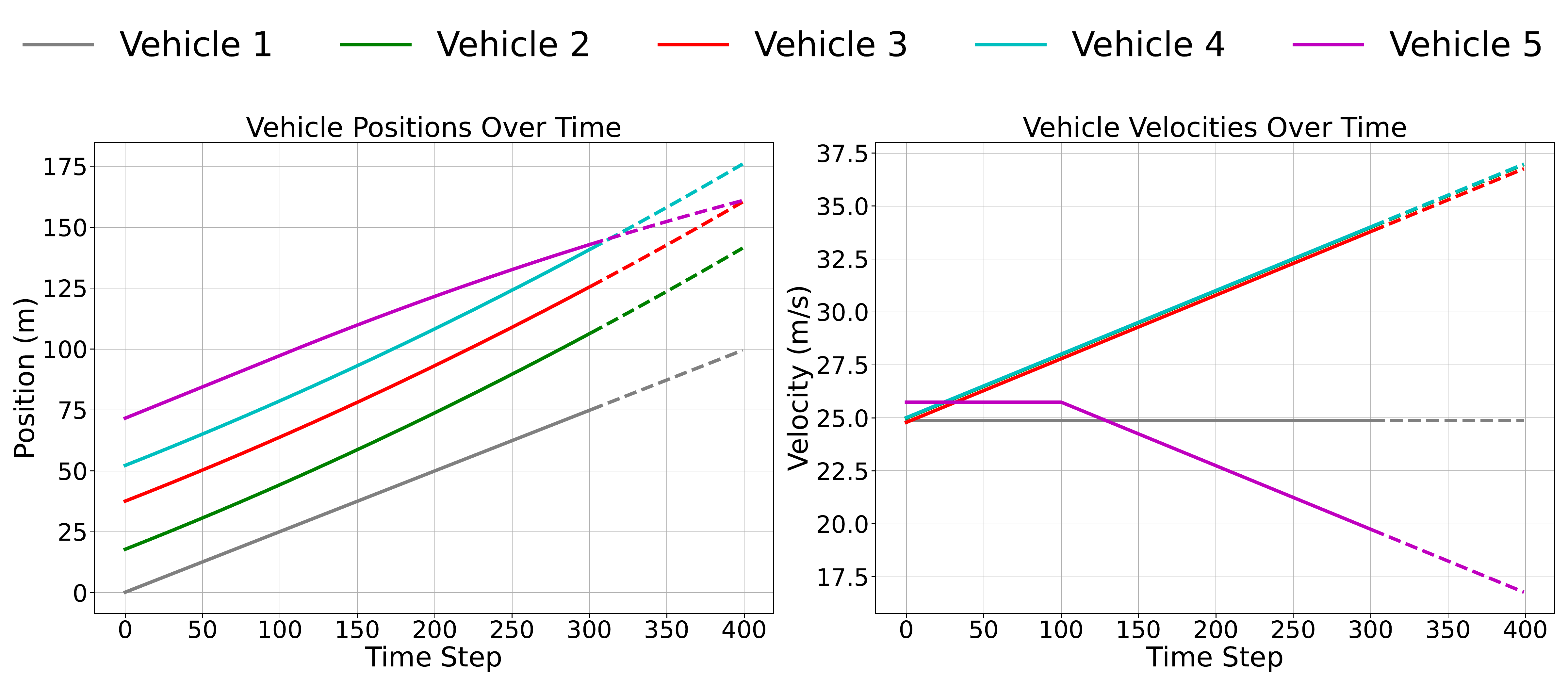}
    \caption{Vehicles fail to respond due to training on a homogeneous environment.}
    \label{fig:fail}
\end{figure}

Based on these observations and preliminary validation, we are motivated to design a robustness-aware framework that integrates defensive policies and DTs to safeguard FRL-based mobile systems, as discussed in following sections. 

%% file: sec/attack.tex
\section{Problem Statement and System Architecture}

\begin{figure}[htbp]
    \centering
    \includegraphics[width=0.42\textwidth]{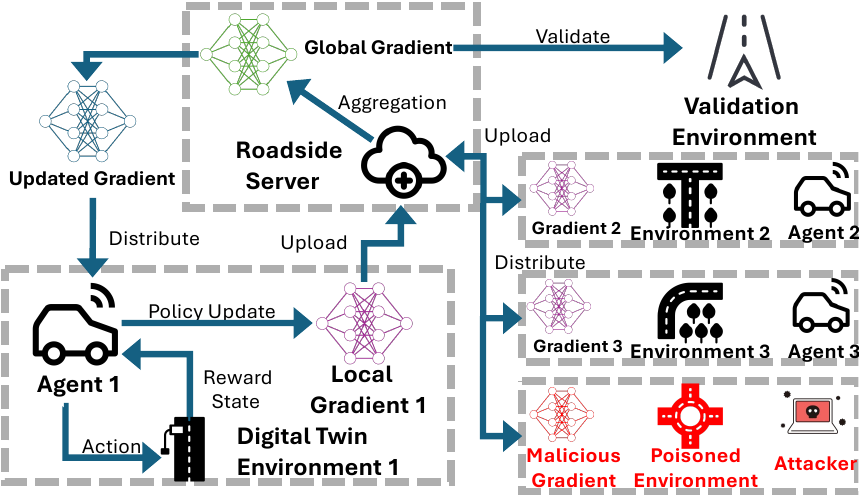}
    \caption{System overview under threats.}
    \label{fig:frl}
\end{figure}

\myparatight{Threat model}
In FRL-based mobile systems, we consider a scenario where attackers compromise a subset of agents (e.g., RL-assisted AVs), rendering them malicious.
We use a sample system with three benign agents and one malicious agent to illustrate the problem in Fig.~\ref{fig:frl}.
Multiple DTs are used to generate diverse AD environment instances for parallel local training.
Specifically, the malicious agent manipulates its local model updates (e.g., gradients) and transmits corrupted information to roadside server. Such an adversarial behavior can compromise the integrity of the globally aggregated control model, leading to higher collision rates and loss of individual vehicle control through model propagation.

%

\myparatight{Attacker’s knowledge}
We assume a worst-case scenario where attackers have comprehensive knowledge of FRL-based safety-critical system.
This includes details about local training trajectories from DTs, aggregation rules, gradients of all agents, and global models. 
Environment instances generated by DTs can also be accessed and manipulated by attackers.
Such extensive knowledge allows the attacker to craft sophisticated attacks, ensuring that our proposed security mechanisms are evaluated under the most challenging conditions.

\myparatight{Defender’s knowledge and goal}
The defensive mechanism is designed to operate without any prior information about the attacker's strategy with assistance from DTs. 
Each benign agent is unaware of the location, identity, or number of malicious agents in the system and their corresponding DTs. Our objective is to develop a robustness-aware FRL framework that ensures safe and reliable operations under the aforementioned challenging conditions. Specifically, the system aims for high control and learning performance, maintaining its effectiveness in non-adversarial mobile environments with assistance from DTs. The policy learned by each benign agent should perform at a level comparable to averaging-based aggregation, which is recognized for its superior performance in non-adversarial settings. Second, the framework must exhibit both theoretical and practical resilience to model poisoning attacks, ensuring robust operations even in the presence of multiple compromised agents.
Lastly, the proposed security component is designed to maintain communication and computation efficiency, introducing minimal or no additional overhead compared to safety-oblivious FRL approaches.



\begin{algorithm}[t]
\caption{LocalTraining ($t, k, \mathbf{w}_{t}^{0}, \xi_t, d(\cdot|\mathbf{w}_{t}^{0})$).}
\label{alg:local}
\begin{algorithmic}[1]
\State {\bfseries Output:} $\mu_t^{k}$
\For{each local training round}
    \State Sample $\xi_t$ trajectories $\{\upsilon^{k}_{t,i}\}_{i=1}^{\xi_t}$ from $d(\cdot|\mathbf{w}_{t}^{0})$
    \State {$\mu_t^{k} = \frac{1}{\xi_t} \sum_{i=1}^{\xi_t}  g(\upsilon^{k}_{t,i} | \mathbf{w}_{t}^{0})$ }
    \label{alg1_line4}
\EndFor
\end{algorithmic}
\end{algorithm}

\begin{algorithm}[t]
\caption{FRL training process.}
\label{alg:server}
\begin{algorithmic}[1]
\State {\bfseries Input:} $\tilde{\mathbf{w}}_{0} \in \mathbb{R}^{e}$,  batch size $\xi_t$, mini-batch size $b_t$, step size $\eta_{t}$
\For{$t=1$ {\bfseries to} $T$}
   \State $\mathbf{w}_{t}^{0} \leftarrow \tilde{\mathbf{w}}_{t-1}$ 
   \label{server_send}
   \For{$k=1$ {\bfseries to} $K$ \text{in parallel}}
       \State $\mu_{t}^{k} = \text{LocalTraining}(t, k, \mathbf{w}_{t}^{0}, \xi_t, d(\cdot|\mathbf{w}_t^0))$
   \EndFor
   \State $\mu_t \leftarrow {\frac{1}{K} \sum_{k=1}^{K}} \mu_t^{k}$ \Comment {FedAvg~\cite{mcmahan2017communication}}
   \label{server_avg}
   \State Sample steps $\mathcal{N}_t \sim Geom(\frac{\xi_t}{\xi_t+b_t})$ 
   \For{$n=0$ {\bfseries to} $\mathcal{N}_t-1$}
       \State Sample $b_t$ trajectories $\{\upsilon_{t,j}^n\}_{j=1}^{b_t}$ from $d(\cdot|\mathbf{w}_t^n)$
       \State {$\zeta_{t}^{n} = \frac{1}{b_t} \sum_{j=1}^{b_t} [g(\upsilon_{t,j}^{n}|\mathbf{w}_{t}^{n}) - $
       \Statex \hspace{5em}$\delta(\upsilon_{t,j}^n|\mathbf{w}_{t}^{n}, \mathbf{w}_{t}^{0}) g(\upsilon_{t,j}^{n}|\mathbf{w}_{t}^{0})] + \mu_t$}
       \label{server_semi}
       \State $\mathbf{w}^{n+1}_t = \mathbf{w}_{t}^{n} + \eta_{t}\zeta_{t}^{n}$
       \label{server_semi_update}
   \EndFor
   \State $\tilde{\mathbf{w}}_{t} \leftarrow \mathbf{w}^{\mathcal{N}_t}_{t}$
\EndFor
\State {\bfseries Output: $\tilde{\mathbf{w}}_{a}$} uniformly randomly picked from $\{\tilde{\mathbf{w}}_t\}_{t=1}^T$
\end{algorithmic}
\end{algorithm}

%% file: sec/defense.tex
\section{\alg Solution}
Technically, our proposed \alg is designed to safeguard the global decision-making model $\tilde{\mathbf{w}}_T$ under the threats through distributed learning across multiple local RL agents over $T$ training rounds. 
Poisoning attacks occur when some agents (e.g., AVs) provide noisy or malicious gradients to corrupt the global model that gathers holistic observations for system control, and the framework must be able to defend against such ``bad-news-travels-fast'' behavior. 
The algorithm achieves this by first constructing a majority-consistent set of computed gradients, then selecting a central gradient from that set, and eventually ensuring that only gradients consistent with the central gradient are aggregated by average. 
This multi-step approach minimizes the influence of malicious agents, ensuring the aggregated gradient is representative of honest gradients and protects the learning process from being corrupted by any adversarial behavior.

We construct data-driven DTs by first replicating the physical system's dynamics in a virtual environment using historical data and known system parameters. 
To generate heterogeneous environments, we introduce diverse variations in system parameters into the DTs, such as noise levels, vehicle maneuvers, and starting points. This controlled perturbation process is automated via parameterized simulation scripts, allowing for the rapid production of a wide range of scenarios that mirror both normal and rare cases. 
These diverse environment instances enable extensive training and testing of driving behaviors under varying conditions, including dangerous events that cannot be safely replicated in the physical world, thus broadening the agent’s experience and expertise to handle increasingly challenging scenarios.

On the other hand, malicious agents in AD systems carry out poisoning attacks by altering either the direction or magnitude of their gradients before transmitting them to the roadside server.
To mitigate these attacks, the server first employs a majority-based filtering mechanism to retain only gradients aligned with the majority consensus. In particular, our \alg constructs a set, denoted as $\mathcal{S}$, by filtering out outliers and including only local gradients that are sufficiently close to the majority of others. Specifically, a gradient $\mu_{t}^{k}$ is included in the set $\mathcal{S}$ if following condition is satisfied:
\begin{equation}
\label{set_select}
\left|\left\{k^{\prime}\in[K]:\left\|\mu_{t}^{k^{\prime}}-\mu_{t}^{k}\right\| \leq \psi\right\}\right|>\frac{K}{2},
\end{equation}
where $\psi$ is a threshold that defines the maximum allowable norm difference between two gradients for them to be considered ``close''. The rule includes a gradient \( \mu_{t}^k \) in the set \( \mathcal{S} \) only if the number of gradients close to it (within the distance \( \psi \)) constitutes a majority, i.e., greater than \( \frac{K}{2} \).
This ensures the exclusion of potentially malicious gradients that deviate significantly from the consensus of most gradients.

The server further enhances robustness by selecting a center gradient $\mu_{t}^{\mathcal{S}}$ that is closest to the center of $\mathcal{S}$. 
This technique is designed to minimize the impact of any remaining outliers or malicious gradients that are far from other gradients in $\mathcal{S}$. This method ensures that the selected gradient is relatively robust to any adversarial behavior that may still be present, but it needs to be refined to achieve a global model that is completely resilient to poisoning attacks.
The server then defines the set $\mathcal{W}_{t}$, which contains the indices of the gradients $\mu_{t}^{k}$ that are close to the previously selected center gradient $\mu_{t}^{\mathcal{S}}$ with:
\begin{equation}
\label{fix_scaling}
    \left\|\mu_{t}^{k}-\mu_{t-1}\right\| \leq \lambda \left\|  \mu_{t}^{\mathcal{S}} - \mu_{t-1}\right\|,
\end{equation}
where $\lambda$ is a scaling factor.
Specifically, an gradient is included in $\mathcal{W}_{t}$ if its distance from the previous aggregated gradient $\mu_{t-1}$ is within a scaled distance, determined by the factor $\lambda$, of the distance between $\mu_{t}^{\mathcal{S}}$ and $\mu_{t-1}$. 
Note that \(\lambda\) can be interpreted as a measure of the allowable rate of change or ``smoothness'' in the updates, ensuring that agents' gradients do not deviate too far from the central gradient $\mu_{t}^{\mathcal{S}}$, which helps in identifying benign gradients.
This step ensures that only gradients that are consistent with the robustly selected central gradient $\mu_{t}^{\mathcal{S}}$ are included in the final aggregation. Gradients that are fairly far away from $\mu_{t}^{\mathcal{S}}$ are excluded, as they are more likely to be influenced by malicious agents. The aggregated gradient $\mu_t$ is computed by the average of all gradients in the benign set $\mathcal{W}_{t}$.
This iterative process, which repeats for $T$ rounds, ensures that the global model $\tilde{\mathbf{w}}_T$ incorporates reliable gradients from honest participants, while progressively filtering out any erroneous or malicious gradients. 
If the set \(\mathcal{S}\) is empty, we can increase the value of \(\psi\) to ensure that it becomes non-empty.

Subsequently, the server employs a stochastic variance reduced gradient (SVRG)-like~\cite{johnson2013accelerating,lei2017non} to update the global model for $\mathcal{N}_t$ iterations, where $\mathcal{N}_t$ is an integer drawn from a Geometric distribution with a parameter of $\frac{\xi_t}{\xi_t+b_t}$. During each iteration, the server first samples a mini-batch of $b_t$ trajectories from the trajectory distribution $d(\cdot | \mathbf{w}_{t}^{n})$ induced by the current global model $\mathbf{w}_{t}^{n}$. It then computes a semi-stochastic gradient $\zeta_{t}^{n}$ (Line~\ref{server_semi} in Algorithm~\ref{alg:server}), which is subsequently used to update the current global model (Line~\ref{server_semi_update} of Algorithm~\ref{alg:server}). Notably, 
to reduce the learning variance, the control term  $\delta(\upsilon_{t,j}^n | \mathbf{w}_t^n, \mathbf{w}_t^0)$ is calculated as $\delta(\upsilon_{t,j}^n | \mathbf{w}_t^n, \mathbf{w}_t^0) = \frac{d(\upsilon_{t,j}^n | \mathbf{w}_t^0)}{d(\upsilon_{t,j}^n| \mathbf{w}_t^n)}$.


\myparatight{Remark}
According to Eq.~(\ref{set_select}), our \alg requires computing pairwise distances between gradients, which can lead to significant computational overhead, particularly when the number of agents \( K \) is large. This overhead grows quadratically, with a complexity of \( O(K^2) \), making it impractical in many scenarios. 
To mitigate this issue, approximate nearest neighbor (ANN) search~\cite{indyk1998approximate,arya1998optimal} is employed. This approach leverages efficient data structures such as KD-Trees~\cite{bentley1990k}, Ball Trees~\cite{cayton2008fast}, or locality-sensitive hashing (LSH)~\cite{datar2004locality} to identify gradients within a distance \( \psi \) of a given gradient without explicitly calculating all pairwise distances. 
While ANN search operates with sub-linear time complexity for neighbor identification, it yields approximate rather than exact results. However, these approximations are typically adequate for consensus-based methods like the proposed \alg.

\begin{algorithm}[!htbp]
\caption{\alg}
\label{alg:our}
\begin{algorithmic}[1]
\State {\bfseries Input:} $\tilde{\mathbf{w}}_{0} \in \mathbb{R}^{e}$,  batch size $\xi_t$, mini-batch size $b_t$, step size $\eta_{t}$, threshold $\psi$, scaling factor $\lambda$
\For{$t=1$ {\bfseries to} $T$}
   \State $\mathbf{w}_{t}^{0} \leftarrow \tilde{\mathbf{w}}_{t-1}$ 
   \label{eq:synchronize}
   \For{$k=1$ {\bfseries to} $K$ \text{in parallel}}
       \State $\mu_{t}^{k} = \text{LocalTraining}(\mathbf{w}_{t}^{0}, \xi_t, d(\cdot|\mathbf{w}^0_t))$
   \EndFor
   \label{local}
\State {$\mathcal{S} = \{\mu_{t}^{k}\} \text { where } k \in[K] \text { s.t. } $
\Statex \hspace{4em}$\left|\left\{k^{\prime}\in[K]:\left\|\mu_{t}^{k^{\prime}}-\mu_{t}^{k}\right\| \leq \psi\right\}\right|>\frac{K}{2}$}
\label{server_mom_filter}
\State {$\mu_{t}^{\mathcal{S}} \leftarrow \operatorname{argmin}_{\mu_t^{\tilde{k}}}\|\mu_t^{\tilde{k}} - \text{mean}(\mathcal{S})\| \text{ where } \tilde{k} \in \mathcal{S} $}
\label{server_mom}
\State $ \mathcal{W}_{t} = \left\{k \in[K]:\left\|\mu_{t}^{k}-\mu_{t-1}\right\| \leq \lambda \left\|  \mu_{t}^{\mathcal{S}} - \mu_{t-1}\right\| \right\}$ 
\label{server_model_filter}
\State $\mu_{t} = \frac{1}{\left|\mathcal{W}_{t}\right|} \sum_{k \in \mathcal{W}_{t}} \mu_{t}^{k}$ 
\label{server_aggregation}
\State Sample steps $\mathcal{N}_t \sim Geom(\frac{\xi_t}{\xi_t+b_t})$ 
\label{sample}
\For{$n=0$ {\bfseries to} $\mathcal{N}_t-1$}
   \State Sample $b_t$ trajectories $\{\upsilon_{t,j}^n\}_{j=1}^{b_t}$ from $d(\cdot|\mathbf{w}_t^n)$
   \State {$\zeta_{t}^{n} = \frac{1}{b_t} \sum_{j=1}^{b_t} [g(\upsilon_{t,j}^{n}|\bm{\mathbf{w}}_{t}^{n}) - $
   \Statex \hspace{5em} $\delta(\upsilon_{t,j}^n|\mathbf{w}_{t}^{n}, \mathbf{w}_{t}^{0}) g(\upsilon_{t,j}^{n}|\mathbf{w}_{t}^{0})] + \mu_t$}
   \State $\mathbf{w}^{n+1}_t = \mathbf{w}_{t}^{n} + \eta_{t}\zeta_{t}^{n}$
\EndFor
\State $\tilde{\mathbf{w}}_{t} \leftarrow \mathbf{w}^{\mathcal{N}_t}_{t}$
\label{global_update}
\EndFor
\State {\bfseries Output: $\tilde{\mathbf{w}}_{a}$} uniformly randomly picked from $\{\tilde{\mathbf{w}}_t\}_{t=1}^T$
\end{algorithmic}
\end{algorithm}

%% file: sec/theoretical.tex
\section{Theoretical Analysis of Robustness Guarantees}
\label{sec:theoretical}

\noindent To ensure \algns's applicability for safety-critical scenarios, we present the first-of-its-kind theoretical foundations underlying the convergence and sample complexity of FRL-based robust aggregation in \algns. 
%
We use $\pi_{\boldsymbol{\mathbf{w}}}$ to denote the policy parameterized by $\mathbf{w}$.
Let $H$ be the trajectory horizon, i.e., the length of the trajectory.
Our objective is to maximize the cumulative discounted reward for a trajectory $\upsilon$.
The action space and state space are denoted by $\mathcal{A}$ and $\Omega$, respectively. 
The cumulative discounted reward for $\upsilon$ is represented by $\mathcal{R}(\upsilon) = \sum_{h=0}^{H-1} \gamma_{h} \mathcal{R}\left(s_{h}, a_{h}\right)$, where $\mathcal{R}\left(s_{h}, a_{h}\right): \Omega \times \mathcal{A} \mapsto [0, R]$ is the reward the agent receives after being in state $s_h$ and taking action $a_h$, and $\gamma_h$ is the discount factor.  
In RL, the loss function 
$\mathcal{L}(\boldsymbol{\mathbf{w}}) = \mathbb{E}_{\upsilon \sim d(\cdot | \boldsymbol{\mathbf{w}})}[\mathcal{R}(\upsilon)|\Xi]$ 
is used to evaluate the performance of model $\mathbf{w}$, given the MDP $\Xi$.
It is important to highlight that the loss function \(\mathcal{L}(\boldsymbol{\mathbf{w}})\) exhibits significant non-convexity.
%

\begin{assumption}[Bounded variance of the gradient estimator]
    \label{assumption-bounded-variance}
    It is assumed that there exists a constant $\sigma$ such that for any trajectory $\upsilon \sim d(\upsilon | {\mathbf{w}})$, the gradient estimator $g(\upsilon|\boldsymbol{\mathbf{w}})$ satisfies $\|g(\upsilon|\boldsymbol{\mathbf{w}}) - \nabla \mathcal{L}(\boldsymbol{\mathbf{w}})\| \leq \sigma$ for every policy $\pi_{\boldsymbol{\mathbf{w}}}$.
\end{assumption}
\begin{assumption}[Variance of importance weights]
    \label{assumption-uai-weight}
    It is assumed that there exists a finite constant $Q$ such that for any pair of policies, the variance of the importance weights $\delta(\upsilon|\boldsymbol{\mathbf{w}}_1, \boldsymbol{\mathbf{w}}_2)$ is bounded by $Q$, for all $\boldsymbol{\mathbf{w}}_1, \boldsymbol{\mathbf{w}}_2 \in \mathbb{R}^e$ and $\upsilon \sim d(\cdot | \boldsymbol{\mathbf{w}}_1)$, where $e$ is model dimension.
\end{assumption}
\begin{assumption}[Policy gradient properties]
    \label{assumption-policy-derivatives}
    Let $\pi_{\boldsymbol{\mathbf{w}}}(a|s)$ denote the policy of an agent in state $s$. There exist positive constants $G$ and $M$ such that, for any action $a \in \mathcal{A}$ and state $s \in \Omega$, log-gradient and Hessian of policy satisfy the bounds:
    \begin{align*}
        |\nabla_{\boldsymbol{\mathbf{w}}}\log \pi_{\boldsymbol{\mathbf{w}}}(a | s) | \leq G, \quad \|\nabla_{\boldsymbol{\mathbf{w}}}^2\log \pi_{\boldsymbol{\mathbf{w}}}(a | s) \| \leq M.
    \end{align*}
\end{assumption}

\begin{remark}
   Assumption \ref{assumption-bounded-variance} is commonly employed in various existing works~\cite{lei2017non,alistarh2018byzantine,Bulusu_TSPIN_2021,allen2016variance}, which bounds the variance of the gradient estimator with a constant $\sigma$. This ensures that the gradient estimates remain controlled across different agents, which is critical given the stochastic nature of policy gradients in federated settings. Additionally, we note that Assumption~\ref{assumption-bounded-variance}
    always hold true for complex, real-world problems involving continuous and high-dimensional controls, particularly when the MDP satisfies a Lipschitz continuity condition, as detailed in~\cite{pirotta2015policy}. 
   Assumption~\ref{assumption-uai-weight} is standard in the literature of RL system (see, e.g., \cite{papini2018stochastic,xu2020improved}), which introduces a finite bound $Q$ on the variance of importance weights. This constraint helps prevent large fluctuations in importance sampling corrections, thereby stabilizing the learning process. 
  Lastly, Assumption~\ref{assumption-policy-derivatives} sets specific bounds, denoted by $G$ and $M$, on the log-gradient and Hessian of the policy. These bounds ensure that policy gradients remain smooth and controlled, where gradients are aggregated from multiple agents. This assumption is widely used in the existing literature~\cite{allen2016variance, reddi2016stochastic, papini2018stochastic, xu2020improved}.

\end{remark}

\begin{theorem}[Convergence of \algns]
\label{main-theorem}
Given the conditions specified in Assumptions~\ref{assumption-bounded-variance}, \ref{assumption-uai-weight}, and \ref{assumption-policy-derivatives}, assume that the objective function $\mathcal{L}({\mathbf{w}})$ is $L_{\ell}$-smooth, where $L_{\ell} = \frac{HM(R+HG^2)}{1-\gamma}$, and that the fraction of malicious agents satisfies $\alpha < \frac{1}{2}$, so that benign gradients form the majority required by the filtering rule in Eq.~\eqref{set_select}.
If the step size $\eta_t$ satisfies $\eta_t \leq \frac{1}{2 \tau \xi_t^{\frac{2}{3}}}$, with $b_t = 1$ and $\xi_t \geq \frac{4\Phi}{L_{\ell}^2}$, where $\Phi = L_g + Z_g^2 Z_w$ and $\tau = (L_{\ell}\Phi)^{\frac{1}{3}}$, the parameters are given by $L_g = \frac{HM(R + |Z_b|)}{1-\gamma}$, $Z_g = \frac{HG(R + |Z_b|)}{1-\gamma}$, $Z_b$ denotes the baseline reward, and $Z_w = H(Q + 1)(2HG^2 + M)$ with $Q$ defined in Assumption~\ref{assumption-uai-weight}.
Let \( V = 2 \log\left(\frac{2K}{\theta}\right) \) and \( \psi = 2\sigma \sqrt{\frac{V}{\xi_t}} \), where the confidence parameter \( \theta \in (0, 1) \) is chosen as $\theta = \frac{V}{4\xi_t}$ such that $e^{\frac{\theta \xi_t}{2(1-2\theta)}} \leq \frac{2K}{\theta} \leq e^{\frac{\xi_t}{2}}$.
The scaling factor is chosen by \\ $\lambda \leq \frac{\psi - \| \mu^{\mathcal{S}}_t - \mu_{t-1} \|}{\|\mu^{\mathcal{S}}_t - \mu_{t-1} \|}$.
At the end of $T$ training rounds, the model $\tilde{{\mathbf{w}}}_a$, selected uniformly at random from $\{\tilde{{\mathbf{w}}}_t\}_{t=1}^{T}$, satisfies the convergence bound:
\begin{align*}
    \mathbb{E}[\|\nabla \mathcal{L}(\tilde{{\mathbf{w}}}_{a})\|^{2}] 
    \leq \frac{8 \tau \left[\mathcal{L}(\tilde{{\mathbf{w}}}^*) - \mathcal{L}(\tilde{{\mathbf{w}}}_{0})\right]}{T\xi_t^{\frac{1}{3}}} 
     + 16\psi^2 + \frac{512 \sigma^2 V}{\xi_t},
\end{align*}
where \(\tilde{\mathbf{w}}^*\) represents the global maximizer of \(\mathcal{L}\).
\end{theorem}

\begin{proof}

Due to the page limit, we only provide a proof sketch. The full proof can be found in Appendix A.
Starting the proof of Theorem~\ref{main-theorem} from the update equation $\mathbf{w}^{n+1}_t = \mathbf{w}_t^n + \eta_t \zeta_t^n$, we have:
\begin{align*}
&\mathbb{E}_{\upsilon_t^n}\|\mathbf{w}^{n+1}_t - \mathbf{w}_t^0\|^2 
    \notag\\&= \mathbb{E}_{\upsilon_t^n}\|\mathbf{w}_{t}^n - \mathbf{w}_t^0 + \eta_t \zeta_t^n\|^2 \\
    &\leq \eta_t^2[(2L_g + 2Z_g^2Z_w)\|\mathbf{w}_t^n-\mathbf{w}_t^0\|^2 \\
    &\quad + 2\|\nabla \mathcal{L}(\mathbf{w}_t^n)\|^2 + 2\|\mu_{t}-\nabla \mathcal{L}\left(\mathbf{w}_{t}^{0}\right)\|^2] \\
    &\quad + 2\eta_t \langle \mu_{t}-\nabla \mathcal{L}\left(\mathbf{w}_{t}^{0}\right), \mathbf{w}_t^n-\mathbf{w}_t^0 \rangle \\
    &\quad + 2\eta_t \langle\nabla \mathcal{L}(\mathbf{w}_t^n), \mathbf{w}_t^n-\mathbf{w}_t^0 \rangle + \|\mathbf{w}_t^n - \mathbf{w}_t^0\|^2 
    \stepcounter{equation}\tag{\theequation}\label{lemma-neg-2eta-E-1-sk} 
    %
\end{align*}
%
Next, let $\mathbb{E}_t$ denote the expectation with respect to all trajectories $\{\upsilon_t^1, \upsilon_t^2, \dots\}$, given $\mathcal{N}_t$. Since the trajectories are independent of $\mathcal{N}_t$, $\mathbb{E}_t$ can be viewed as the expectation over $\{\upsilon_t^1, \upsilon_t^2, \dots\}$. Therefore, we have:
\begin{align*}
    &\mathbb{E}_{t}\|\mathbf{w}^{n+1}_t - \mathbf{w}_t^0\|^2 \notag\\
    &\leq [(2Z_g^2Z_w+2L_g)\eta_t^2+1]\mathbb{E}_{t}\|\mathbf{w}_t^n-\mathbf{w}_t^0\|^2 \\
    &\quad + 2\eta_t\mathbb{E}_{t}\langle\nabla \mathcal{L}(\mathbf{w}_t^n), \mathbf{w}_t^n-\mathbf{w}_t^0 \rangle \\
    &\quad + 2\eta_t\mathbb{E}_{t}\langle \mu_{t}-\nabla \mathcal{L}\left(\mathbf{w}_{t}^{0}\right), \mathbf{w}_t^n-\mathbf{w}_t^0 \rangle \\
    &\quad + 2\eta_t^2\mathbb{E}_{t}\|\nabla \mathcal{L}(\mathbf{w}_t^n)\|^2 + 2\eta_t^2\|\mu_{t}-\nabla \mathcal{L}\left(\mathbf{w}_{t}^{0}\right)\|^2.
    \stepcounter{equation}\tag{\theequation}\label{lemma-neg-2eta-E-2-sk} 
\end{align*}
Taking the expectation over $\mathcal{N}_t$ with $n = \mathcal{N}_t$, following Fubini's theorem
and replacing replace $\mathbf{w}^{\mathcal{N}_t}_t$ with $\tilde{\mathbf{w}}_t$, and $\mathbf{w}_t^0$ with $\tilde{\mathbf{w}}_{t-1}$, taking the expectation over the equation, we have
\begin{align}\label{lemma-neg-2eta-E-sk}
    &-2\eta_t\mathbb{E} \langle \mu_{t}-\nabla \mathcal{L}\left(\mathbf{w}_{t}^{0}\right), \tilde{\mathbf{w}}_t - \tilde{\mathbf{w}}_{t-1} \rangle \notag\\
    &\leq \left[\eta_t^2(2L_g + 2Z_g^2Z_w) - \frac{1}{\xi_t}\right]\mathbb{E}\|\tilde{\mathbf{w}}_t - \tilde{\mathbf{w}}_{t-1}\|^2 \notag\\
    &\quad + 2\eta_t \mathbb{E} \langle \nabla \mathcal{L}(\tilde{\mathbf{w}}_t), \tilde{\mathbf{w}}_t - \tilde{\mathbf{w}}_{t-1} \rangle + 2\eta^2_t\mathbb{E}\|\nabla \mathcal{L}(\tilde{\mathbf{w}}_t)\|^2  \notag\\
    &\quad + 2\eta_t^2\mathbb{E}\|\mu_{t}-\nabla \mathcal{L}\left(\mathbf{w}_{t}^{0}\right)\|^2.
\end{align}
Rearrange the terms,
it follows $\mathbb{E}\left[D_N - D_{N+1}\right] = \left(1 - \frac{1}{\Gamma}\right)(\mathbb{E}[D_N] - D_0)$ with Fubini's theorem.
Note that $\tilde{\boldsymbol{\mathbf{w}}}_t = \boldsymbol{\mathbf{w}}^{\mathcal{N}_t}_t$ and $\tilde{\boldsymbol{\mathbf{w}}}_{t-1} = \boldsymbol{\mathbf{w}}_t^0$. If we take expectation over all randomness and denote it by $\mathbb{E}$, we get:
\begin{align*}
    &\eta_{t}(1-L_{\ell} \eta_{t}) \mathbb{E}\|\nabla \mathcal{L}(\tilde{\boldsymbol{\mathbf{w}}}_{t})\|^{2} \\
    &\leq \frac{1}{\xi_t} \mathbb{E}\left[\mathcal{L}(\tilde{\boldsymbol{\mathbf{w}}}_{t})-\mathcal{L}(\tilde{\boldsymbol{\mathbf{w}}}_{t-1})\right] \\
    &\quad + \frac{1}{2 \eta_{t} \xi_t} \left[-\frac{1}{\xi_t}+\eta_t^2(2L_g+2Z_g^2Z_w)\right] \mathbb{E}\|\tilde{\boldsymbol{\mathbf{w}}}_{t}-\tilde{\boldsymbol{\mathbf{w}}}_{t-1}\|^{2} \\
    &\quad + \frac{1}{\xi_t} \mathbb{E}\left\langle\nabla \mathcal{L}(\tilde{\boldsymbol{\mathbf{w}}}_{t}), \tilde{\boldsymbol{\mathbf{w}}}_{t}-\tilde{\boldsymbol{\mathbf{w}}}_{t-1}\right\rangle 
    \notag\\&\quad + \frac{\eta_{t}}{\xi_t} \mathbb{E}\|\nabla \mathcal{L}(\tilde{\boldsymbol{\mathbf{w}}}_{t})\|^{2} 
    + \frac{\eta_{t}}{\xi_t} \mathbb{E}\left\|\mu_{t}-\nabla \mathcal{L}\left(\mathbf{w}_{t}^{0}\right)\right\|^{2} \notag\\
    & \quad +L_{\ell}\eta_t^2(L_g+Z_g^2Z_w) \mathbb{E}\|\tilde{\boldsymbol{\mathbf{w}}}_{t}-\tilde{\boldsymbol{\mathbf{w}}}_{t-1}\|^{2} 
    \notag\\ & \quad + \eta_{t}(1+L_{\ell} \eta_{t}) \mathbb{E}\left\|\mu_{t}-\nabla \mathcal{L}\left(\mathbf{w}_{t}^{0}\right)\right\|^{2}, 
    \stepcounter{equation}\tag{\theequation}\label{main-lemma-neg-2eta-E-sk}
\end{align*}
%
%
We apply Young's inequality, that 
For any real numbers $x$ and $y$, and for any $\rho > 0$, the following holds:
$xy \leq \frac{x^2}{2\rho} + \frac{\rho}{2} y^2$
, 
on $ \mathbb{E}\left\langle\nabla \mathcal{L}(\tilde{\boldsymbol{\mathbf{w}}}_{t}), \tilde{\boldsymbol{\mathbf{w}}}_{t}-\tilde{\boldsymbol{\mathbf{w}}}_{t-1}\right\rangle$ 
using $x=\tilde{\boldsymbol{\mathbf{w}}}_t - \tilde{\boldsymbol{\mathbf{w}}}_{t-1}$, $y=\nabla \mathcal{L}(\tilde{\boldsymbol{\mathbf{w}}}_t)$, and $\rho = \frac{1-2\eta_t^2(L_g + Z_g^2Z_w)\xi_t - 2L_{\ell}\eta_t^3(L_g + Z_g^2Z_w) \xi_t^2}{\eta_t \xi_t}$ 
to get:
\begin{align*}
    &\frac{1}{\xi_t} \mathbb{E}\left\langle y, x \right\rangle 
    \leq \frac{1}{2 \rho}
    \mathbb{E}\|x\|^2 
    + \frac{\rho}{2}
    \mathbb{E}\|y\|^2. 
    \stepcounter{equation}\tag{\theequation}\label{lemma-proof:label-name-so-hard-2-sk}
\end{align*}
Combining \eqref{main-lemma-neg-2eta-E-sk}  and \eqref{lemma-proof:label-name-so-hard-2-sk} and rearranging, we have:
\begin{align*}
    &\eta_{t} \Bigg( 
 -  \frac{1}{2[1 - 2\eta_t^2(L_g + Z_g^2Z_w)\xi_t - 2L_{\ell}\eta_t^3(L_g + Z_g^2Z_w) \xi_t^2]} \notag\\ &+ 1 - L_{\ell} \eta_{t}  - \frac{1}{\xi_t} \Bigg) \mathbb{E}\|\nabla \mathcal{L}(\tilde{\boldsymbol{\mathbf{w}}}_{t})\|^{2} \\
    \leq& \frac{1}{\xi_t}\mathbb{E}\left[\mathcal{L}(\tilde{\boldsymbol{\mathbf{w}}}_{t}) - \mathcal{L}(\tilde{\boldsymbol{\mathbf{w}}}_{t-1})\right] 
     \notag\\&+ \eta_t(1 +L_{\ell}\eta_t + \frac{1}{\xi_t}) \left[2\psi^2 + 64 \sigma^2 \frac{ V}{\xi_t}\right],
    \stepcounter{equation}\tag{\theequation}\label{main-final-1-sk}
\end{align*}
%
We want to choose $\eta_t$ such that $1-2\eta_t^2(L_g + Z_g^2Z_w)\xi_t - 2L_{\ell}\eta_t^3(L_g + Z_g^2Z_w) \xi_t^2 > 0$. Denoting $\Phi = L_g + Z_g^2Z_w$, we have the following equation:
\begin{align*}
    1>1-2\eta_t^2\Phi \xi_t - 2L_{\ell}\eta_t^3\Phi \xi_t^2 &> 0 \\
    \frac{1}{2(1-2\eta_t^2\Phi \xi_t - 2L_{\ell}\eta_t^3\Phi \xi_t^2)} &> \frac{1}{2}
\end{align*}
Thus, we can select $\eta_t$ such that:
\begin{align*}
    1- L_{\ell} \eta_{t}  - \frac{1}{\xi_t} -  \frac{1}{2[1-2\eta_t^2\Phi \xi_t - 2L_{\ell}\eta_t^3\Phi \xi_t^2]}  & \geq \frac{1}{4} \stepcounter{equation}\tag{\theequation}\label{eq:eta_1-sk}
    \\
    L_{\ell} \eta_{t}  + \frac{1}{\xi_t} +  \frac{1}{2[1-2\eta_t^2\Phi \xi_t - 2L_{\ell}\eta_t^3\Phi \xi_t^2]} &\leq \frac{3}{4}
\end{align*}
Next, we aim to satisfy the condition: 
\begin{align*}
    \text{(i)} \frac{1}{2} <  \frac{1}{2[1-2\eta_t^2\Phi \xi_t - 2L_{\ell}\eta_t^3\Phi \xi_t^2]}  & \leq \frac{5}{8} \\ 
    \text{(ii)} L_{\ell}\eta_t \leq \frac{1}{16} \qquad\qquad\qquad\quad \text{(iii)} \frac{1}{\xi_t} & \leq \frac{1}{16}
\end{align*}
To ensure the conditions are satisfied, we choose $\eta_t = \frac{1}{2\tau \xi_t^{2/3}}$ when $\xi_t \geq 16$ which satisfies \eqref{eq:eta_1-sk}.
We can derive the following from \eqref{main-final-1-sk} and \eqref{eq:eta_1-sk}.
Replacing $\eta_t = \frac{1}{2\tau \xi_t^{2/3}}$ and telescoping over $t = 1,2,...,T$ with a constant batch size $\xi_t$, we have for $\tilde{\boldsymbol{\mathbf{w}}}_{a}$ uniformly sampled from $\{\tilde{\boldsymbol{\mathbf{w}}}_t\}_{t=1}^{T}$:
\begin{align*}
    \mathbb{E}\|\nabla \mathcal{L}(\tilde{\boldsymbol{\mathbf{w}}}_{a})\|^{2} 
    &\leq \frac{8 \tau \left[\mathcal{L}(\tilde{\boldsymbol{\mathbf{w}}}^*) - \mathcal{L}(\tilde{\boldsymbol{\mathbf{w}}}_{0})\right]}{T\xi_t^{\frac{1}{3}}} \\
    &\quad + 16\psi^2 + 512 \sigma^2 \frac{ V}{\xi_t},
\end{align*}
which completes the proof.
\end{proof}

\begin{remark}

Theorem~\ref{main-theorem} demonstrates that, with appropriately chosen parameters and fewer than half of the agents being malicious, the learned model converges to a first-order stationary point of \(\mathcal{L}\) even in Byzantine environments. In RL-based systems, where the loss function is highly non-convex, convergence to a stationary point represents the best possible guarantee. We also note that the analysis is carried out for the mini-batch size $b_t = 1$ following~\cite{lei2017non}; the same argument extends to a general mini-batch size $b_t$, with the corresponding variance terms scaled by $\frac{1}{b_t}$.

\end{remark}

%% file: sec/exp.tex
\section{Performance Evaluation}
\label{sec: exp}

\begin{figure*}
    \centering
    \includegraphics[width=0.82\linewidth]{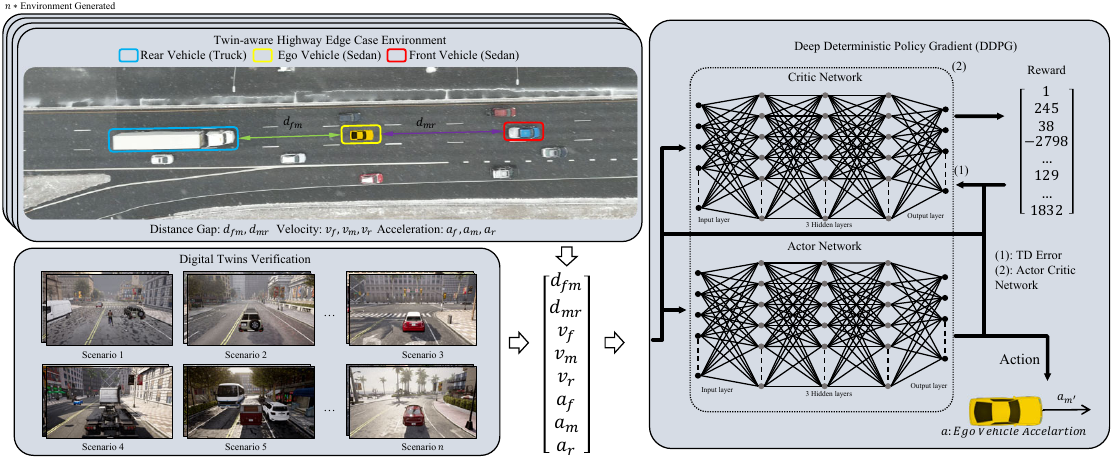}
    \caption{Framework of vehicular RL training and validation in HighwayDT.}
    \label{DRL_framework}
    \vspace{-0.12in}
\end{figure*}

\begin{figure}[b]
    \centering
\includegraphics[width=0.4\textwidth]{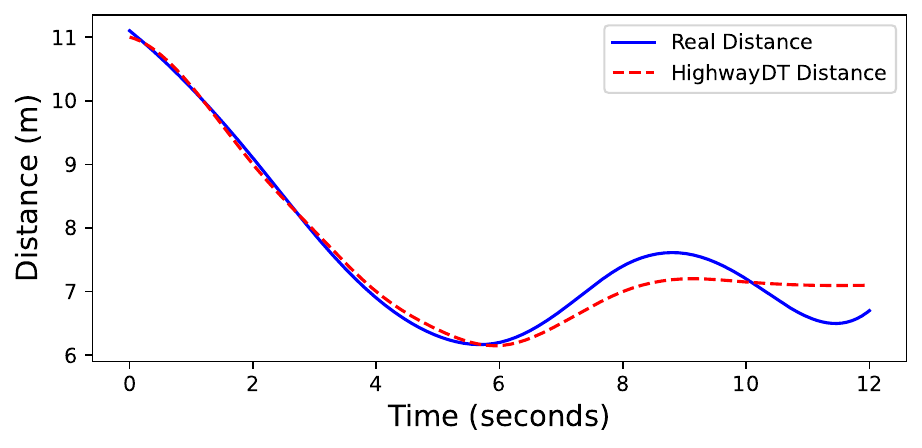}
    \caption{Fidelity test samples of inter-vehicle following distances in our HighwayDT and the real-world tests.}
    \label{fig:gap}
\end{figure}





\begin{figure*}[!ht]
    \centering
    \begin{subfigure}[b]{0.34\textwidth}
        \centering
        \includegraphics[width=\linewidth]{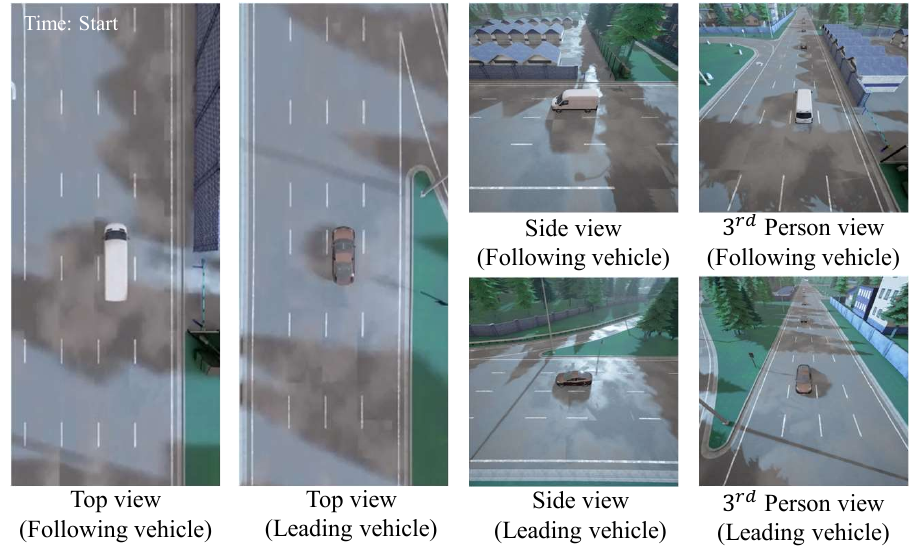}
        \caption{Vehicle positions at the start of the successful scenario.}
        \label{fig:sub1_1}
    \end{subfigure}
    \begin{subfigure}[b]{0.34\textwidth}
        \centering
        \includegraphics[width=\linewidth]{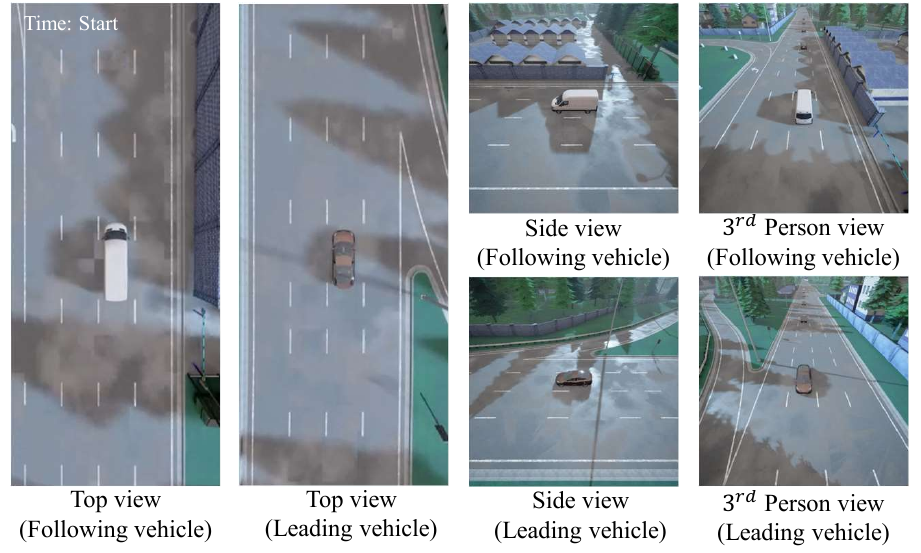}
        \caption{Vehicle positions at the start of the failed scenario.}
        \label{fig:sub1_2}
    \end{subfigure}
    
    \begin{subfigure}[b]{0.34\textwidth}
        \centering
        \includegraphics[width=\linewidth]{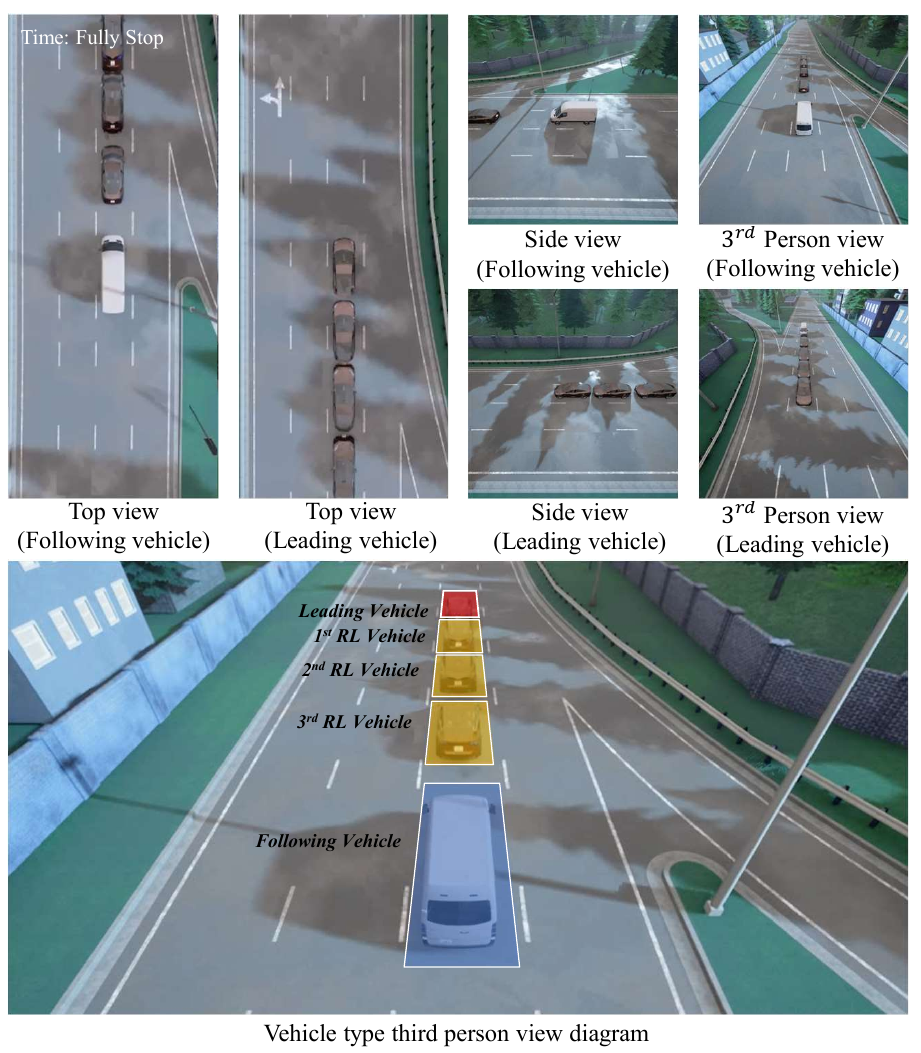}
        \caption{Vehicle positions when fully stopped.}
        \label{fig:sub2_1}
    \end{subfigure}
    \begin{subfigure}[b]{0.34\textwidth}
        \centering
        \includegraphics[width=\linewidth]{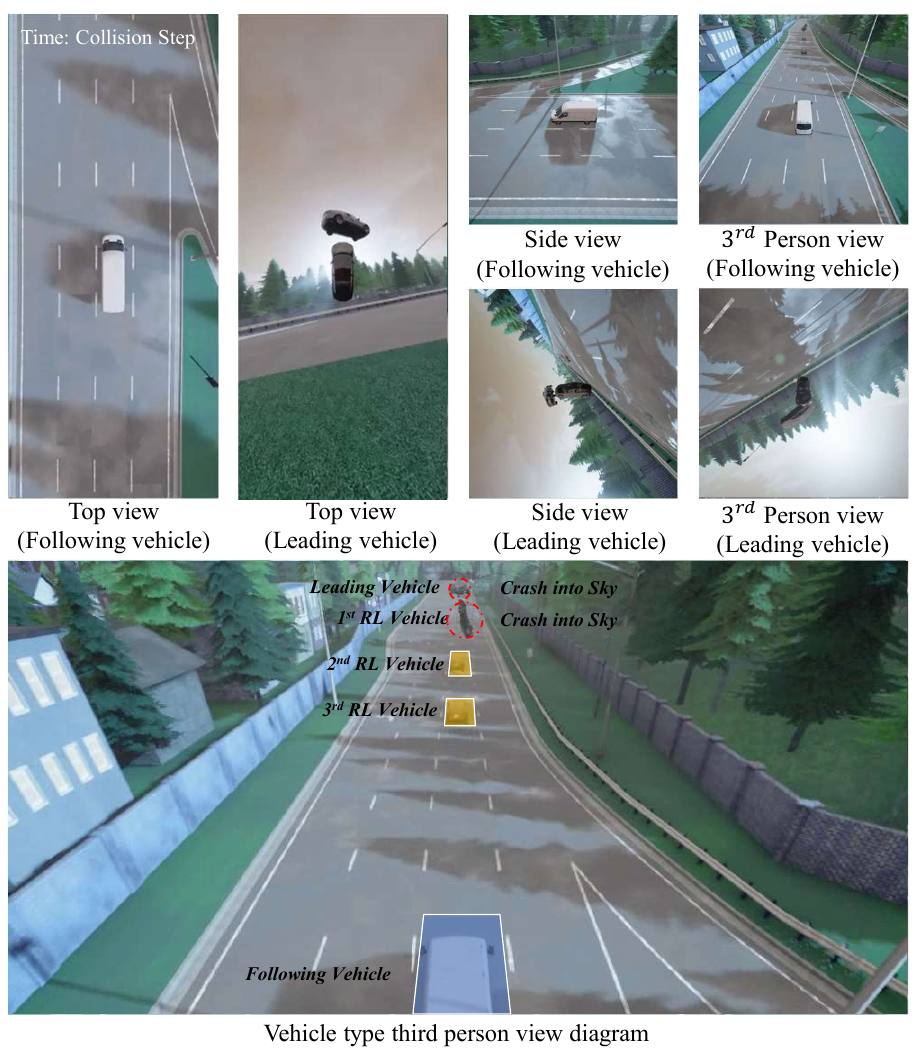}
        \caption{Vehicle positions when a collision happens.}
        \label{fig:sub2_2}
    \end{subfigure}
    \caption{Visualization of two scenarios generated by HighwayDT. Fig.~\ref{fig:sub1_1} and Fig.~\ref{fig:sub2_1} depict instances of successful actions taken by the RL-assisted vehicle, while Fig.~\ref{fig:sub1_2} and Fig.~\ref{fig:sub2_2} illustrate instances of failed actions that eventually lead to crashes. }
    \vspace{-0.1in}
    \label{fig:visual}
\end{figure*}


\subsection{Experimental Setup}

\subsubsection{Scenario Configurations}


\noindent \myparatight{FRL-based AV environments and DT setup} 
In our experiments, we gather real-world vehicular data at the I-695 highway segment in Baltimore, Maryland, using drones. Given the dangers and restrictions associated with testing on actual highways, we first develop the \textit{HighwayDT} based on the collected data, tailored for evaluating FRL-based AV systems in controlled, realistic settings.
%
This DT is a virtual replica of a physical autonomous driving system in highway scenarios, modeling its dynamics and behavior in real time to enable continuous monitoring, simulation, and optimization. 
HighwayDT builds upon the static environment by incorporating additional functionalities to enhance its capabilities. 
Figure~\ref{fig:gap} illustrates the minimal sim-to-reality gap achieved by our DT system. To emulate realistic highway scenarios, we incorporate the Intelligent Driver Model (IDM) for vehicle-following model calibration, real-world vehicle trajectories, and additional physical constraints into a SUMO-CARLA co-simulator.
%
Specifically, it constructs highly \textit{heterogeneous} autonomous driving scenarios, focusing on critical aspects such as longitudinal control and collision avoidance, as depicted in Figs.~\ref{DRL_framework} and \ref{fig:visual}.
The virtual development environment supports a closed-loop process, encompassing safety-critical scenario generation, high-fidelity simulations, and predictive analytics.
It is also capable of generating numerous parallel scenarios that allow RL agents to learn and make reliable decisions under complex and ever-changing driving conditions. 
%
More importantly, the predictive analytics capabilities allow for the simulation of future driving scenarios based on a real-time data stream, helping identify potential risks and test the RL agents' responses to unforeseen events, as shown in the bottom-left part of Fig.~\ref{DRL_framework}.
It is worth noting that the agents trained within such emulators eventually perform better than the current integration of Adaptive Cruise Control (ACC) and Automatic Emergency Braking (AEB) systems because of enriched features as evidenced by existing works~\cite{galvani2019history, burton1997evaluation}. The rigorous validation processes ensure that only algorithms such as \alg meeting high safety and performance standards are considered for deployment in open road tests. 
This includes testing performance under varied vehicle behaviors and rare crash scenarios.

Fig.~\ref{fig:visual} visualizes two scenarios generated by HighwayDT, where Fig.~\ref{fig:sub1_1} and Fig.~\ref{fig:sub2_1} demonstrate a successful vehicle control scenario, and Fig.~\ref{fig:sub1_2} and Fig.~\ref{fig:sub2_2} demonstrate the failed cases. The first six subfigures of Fig.~\ref{fig:sub1_1} or Fig.~\ref{fig:sub2_1} display three perspectives (Top, Side, and $3^{rd}$ Person view) for both the leading and following vehicles. In the successful avoidance scenario, Fig.~\ref{fig:sub1_1} shows the initial step of the RL agent, while Fig.~\ref{fig:sub2_1} shows the vehicle fully stopped with no collision. The red vehicle represents the leading vehicle, the yellow ones are controlled by an RL agent, and the blue is the follower. In the unsuccessful scenario, Fig.~\ref{fig:sub1_2} shows the initial step of an attacked RL agent, and Fig.~\ref{fig:sub2_2} depicts a collision between the leading and the $1^{st}$ vehicle. The red-dashed vehicle highlights the collision, while yellow and blue vehicles represent the RL-assisted vehicles and the following vehicles, respectively.

\myparatight{FRL training framework}
As shown in Fig.~\ref{DRL_framework}, all RL agents are responsible for controlling the middle ego vehicles in their respective DTs.
%
%
Note that in our experiments, the FRL-based system controls only the three middle ego vehicles, which are depicted as the three yellow vehicles in Fig.~\ref{fig:sub2_1}.
The entire process operates within a continuous state space represented by an eight-dimensional vector (distance gaps between the leading and RL vehicles ($d_{fm}$) and between the RL vehicles and the following vehicles ($d_{mr}$), along with the velocities ($v_f, v_m, v_r$) and accelerations ($a_f, a_m, a_r$) of the leading, RL, and following vehicles. 
The state space is continuous and entirely sensor-based, derived from actual measurements such as position, velocity, and acceleration rather than discrete time components, which allows the FRL algorithm to align seamlessly with the environmental modeling.
The action space which ranges from $3 m/s^2$ to $-12 m/s^2$ continuously is defined to control the acceleration of the middle ego RL vehicle agents through a 5-layer neural network (detailed in the supplementary material) that aims to maximize the expected return. Each RL vehicle receives a reward for each successful action (representing no collision) and receives a large penalty if a collision occurs.
We exploit the policy gradient~\cite{REINFORCE} to perform local RL training for each vehicle within the DTs.

%

\myparatight{Parameter settings} 
We assume that there are a total of 10 agents that are assigned to different scenarios. 
By default, 20\% of these agents are considered malicious, following the standard settings as in~\cite{Yin18,fang2020local,fung2018mitigating,xia2019faba,sharma2023flair,cao2020fltrust,rieger2022deepsight,fan2021fault,nguyen2022flame}.
In the AD scenario, the evaluation advances in discrete time steps, updating vehicle positions and velocities at each interval, starting from an initial time step of 0.01s. A safety threshold of 3 meters -- equivalent to the length of the RL vehicles -- is maintained to prevent collisions by ensuring a minimum distance between vehicles. 
The following vehicle's braking behavior is modeled with a normally distributed deceleration of around $-6 m/s^2$.
Parameters related to environment configurations, agent training, and those in our method are provided in the supplementary material.


\subsubsection{Poisoning Attacks}
%
%
By default, we employ the following eight untargeted attacks to evaluate the effectiveness of our proposed \alg in practical FRL scenarios: Trim attack~\cite{fang2020local}, Random attack~\cite{cao2022mpaf}, History attack~\cite{cao2022mpaf}, MPAF attack~\cite{cao2022mpaf}, FTI attack~\cite{zhang2024poisoning}, MinMax attack~\cite{shejwalkar2021manipulating}, MinSum attack~\cite{shejwalkar2021manipulating}, and Adaptive attack~\cite{shejwalkar2021manipulating}.
It is worth noting that the Adaptive attack represents the worst-case scenario, where the attacker has full knowledge of the FRL system. Detailed descriptions of these eight attacks can be found in the supplementary materials.
Additionally, we evaluate three backdoor attacks specifically designed for RL-based systems in Appendix D, with attack details provided in Appendix F.

\subsubsection{Baseline Comparison.}
We compare \alg with eleven security rules, including one non-robust aggregation rule (FedAvg~\cite{mcmahan2017communication}) and ten robust aggregation rules (Median~\cite{Yin18}, Trimmed-mean (Trim)~\cite{Yin18}, Krum~\cite{blanchard2017machine}, FoolsGold~\cite{fung2018mitigating},  FABA~\cite{xia2019faba}, FLTrust~\cite{cao2020fltrust}, FLAIR~\cite{sharma2023flair}, 
FedPG-BR (FedPG)~\cite{fan2021fault}, 
FLAME~\cite{nguyen2022flame},
and DeepSight~\cite{rieger2022deepsight}). 
Note that among the ten robust aggregation rules, only FedPG is specifically designed for the FRL system to date, while the other nine methods were originally developed for FL systems but can be adapted to the FRL setting.
Comprehensive descriptions of these aggregation rules are provided in the supplementary materials (Appendix C).
These baseline schemes are implemented by replacing Lines \ref{server_mom_filter} to \ref{server_aggregation} in Alg.~\ref{alg:our}. All other modules stay the same for a fair comparison and stable model training process.

\subsection{Evaluation Metrics}
Similar to the practical AD scenarios, both the leading and following vehicles decelerate from an initial speed to a full stop. 
The distance covered during this deceleration for the leading vehicle can be calculated by the formula $\frac{V^2}{2a_{\text{leading}}}$, where $V$ is the starting velocity and $a_{\text{leading}}$ is the deceleration rate of the leading vehicle. Likewise, the following vehicle decelerates from its initial speed to a stop over a distance given by $\frac{V^2}{2a_{\text{following}}}$, where $a_{\text{following}}$ represents the deceleration rate of the following vehicle.
To ensure that the middle vehicle can safely maintain its position between the leading and following vehicles without causing a collision, the following condition must be satisfied:
\begin{equation}
    d_{\text{collision}} \geq \frac{V_{\text{r\_init}}^2 - V_{\text{r\_final}}^2}{2a_{\text{following}}} - \frac{V_{\text{f\_init}}^2 - V_{\text{f\_final}}^2}{2a_{\text{leading}}},
\end{equation}
where $d_\text{collision}$ is the length of the middle vehicle; 
$V_{\text{r\_init}}$ and $V_{\text{f\_init}}$ are respectively the initial velocities of following and leading vehicles; $V_{\text{r\_final}}$ and $V_{\text{f\_final}}$ are respectively the final velocities of following and leading vehicles; and $a_{\text{following}}$ and $a_{\text{leading}}$ are respectively the deceleration rates of following and leading vehicles. Any cases that violate this condition will lead to collisions. The no-collision rate $\beta$ can be defined as:
\begin{equation}
    \beta = \frac{n_\text{no-collision}}{n_\text{collision} + n_\text{no-collision}},
\end{equation}
where $n_\text{no-collision}$ and $n_\text{collision}$ mean the number of successful control cases (i.e., no-collision) and collisions, respectively. We measure the no-collision rate across over 50K scenarios as our primary performance metric. However, some scenarios are inherently \textit{collision-prone} regardless of the middle vehicle's actions. 
Scenarios that violate natural laws are excluded from our evaluation.
All parameters are directly obtained from HighwayDT rather than being defined by humans.
The larger the no-collision rate $\beta$, the better the performance of the implemented scheme.
Note that in typical real-world AD scenarios, at least a no-collision rate of \textbf{99.95\%} during the validation process is required for an RL-based approach to be considered successful and ready to implement~\cite{Group_2019,NHTSA}.

\subsection{Experimental Results}

\begin{table*}[!htbp]
\renewcommand{\arraystretch}{1.5}
\centering
\caption{No-collision rate of different methods under different attacks.}
\footnotesize 
\begin{tabularx}{\textwidth}{l|*{12}{>{\centering\arraybackslash}X}}
\toprule
Attack & FedAvg & Median & Trim & Krum & FoolsGold & FABA & FLTrust & FLAIR & FedPG & FLAME & Deepsight & \textbf{\alg} \\
\midrule
No attack       & 100.0\% & 100.0\% & 100.0\% & 100.0\% & 100.0\% & 100.0\% & 100.0\% & 100.0\% & 100.0\% & 100.0\% & 100.0\% & \textbf{100.0\%} \\
\hline
Trim attack     & 14.97\% & 7.05\% & 7.05\% & 100.0\% & 8.63\% & 11.79\% & 100.0\% & 9.12\% & 100.0\% & 100.0\% & 100.0\% & \textbf{100.0\%} \\
\hline
Random attack   & 14.25\% & 100.0\% & 100.0\% & 100.0\% & 100.0\% & 14.54\% & 100.0\% & 11.20\% & 100.0\% & 100.0\% & 100.0\% & \textbf{100.0\%} \\
\hline
History attack  & 8.43\% & 100.0\% & 14.88\% & 100.0\% & 100.0\% & 100.0\% & 100.0\% & 100.0\% & 100.0\% & 10.24\% & 100.0\% & \textbf{100.0\%} \\
\hline
MPAF attack     & 10.91\% & 100.0\% & 100.0\% & 13.26\% & 100.0\% & 14.05\% & 8.55\% & 100.0\% & 100.0\% & 100.0\% & 100.0\% & \textbf{100.0\%} \\
\hline
FTI attack      & 12.47\% & 14.01\% & 8.36\% & 11.02\% & 100.0\% & 9.51\% & 12.89\% & 100.0\% & 10.14\% & 100.0\% & 14.22\% & \textbf{100.0\%} \\
\hline
MinMax attack   & 12.73\% & 100.0\% & 100.0\% & 1.69\% & 8.09\% & 9.87\% & 10.24\% & 14.68\% & 13.76\% & 11.95\% & 14.34\% & \textbf{100.0\%} \\
\hline
MinSum attack   & 10.29\% & 13.77\% & 100.0\% & 9.49\% & 11.58\% & 100.0\% & 8.35\% & 12.14\% & 14.19\% & 100.0\% & 100.0\% & \textbf{100.0\%} \\
\hline
Adaptive attack & 14.49\% & 11.49\% & 7.11\% & 13.68\% & 9.84\% & 100.0\% & 13.94\% & 11.19\% & 8.97\% & 8.92\% & 12.33\% & \textbf{100.0\%} \\
\bottomrule
\end{tabularx}
\label{tab:no-collision}
\vspace{-0.1in}
\end{table*}

\myparatight{\alg is effective}
The results presented in Table~\ref{tab:no-collision} demonstrate the superior performance of our proposed \alg over other baseline schemes in mitigating the impact of various poisoning attacks within the vehicular driving environment. 
``No attack'' indicates that all RL agents are benign, with no malicious agents present in the system.
Across all attack scenarios, \alg consistently achieves a 100\% no-collision rate, outperforming all other baseline defense mechanisms, which struggle to maintain similar performance under the same attacks.
For instance, under the Random attack, \alg attains a perfect 100\% no-collision rate, while other methods such as FABA and FLAIR only achieve 14.54\% and 11.20\%, respectively. The success of \alg in this scenario can be attributed to its ability to detect and filter out anomalous gradients based on historical aggregated gradients, which are randomly introduced by the attackers, ensuring that only benign gradients are aggregated. 
One step further, Table~\ref{tab:fpr} shows the false positive rate (FPR) and false negative rate (FNR) of our \alg under various attacks, where FPR denotes the fraction of benign agents incorrectly classified as malicious, and FNR is the fraction of malicious agents incorrectly classified as benign.
As shown in Table~\ref{tab:fpr}, we observe that both the FPR and FNR of our \alg fall below 5\%, indicating an effective filtering performance.
Detailed analysis indicates that \alg excels particularly in scenarios where other defenses exhibit vulnerabilities. For instance, during the History attack, defenses such as FLAME show significant drops in performance with no-collision rates of only 10.24\%. The great performance of \alg can be explained by its ability to recognize and disregard historical patterns that are exploited by attackers to introduce subtle but harmful changes into the aggregated gradient.
Note that our \alg may misclassify few malicious agents as benign, as shown by a low FNR, but these misclassified malicious agents minimally affect the overall system.

%

\begin{table}[htbp]
\centering
\caption{FPR and FNR (\%) of \alg under various attacks.}
\begin{tabularx}{\columnwidth}{l|>{\centering\arraybackslash}X|>{\centering\arraybackslash}X}
\toprule
Attack & FPR & FNR \\
\midrule
No attack       & 2.5 & -- \\
\hline
Random attack   & 1.2 & 2.0 \\
\hline
History attack  & 0.7 & 0.3 \\
\hline
Trim attack     & 3.2 & 1.5 \\
\hline
FTI attack      & 3.8 & 0.8 \\
\hline
MPAF attack     & 2.1 & 1.3 \\
\hline
MinMax attack   & 1.9 & 1.0 \\
\hline
MinSum attack   & 2.3 & 1.1 \\
\hline
Adaptive attack & 0.4 & 0.1 \\
\bottomrule
\end{tabularx}
\label{tab:fpr}
\end{table}



\myparatight{Impact of fraction of malicious agents}
Fig.~\ref{fig:fraction} explores the impact of fraction of malicious agents on the no-collision rate. As the proportion of malicious agents rises, most conventional methods such as FedAvg and Trim exhibit a steep decline in their ability to maintain a high no-collision rate.
%
%
This declining trend is consistent across other attack scenarios, as shown in Figs.~\ref{fig:fraction}b-h, where methods like FoolsGold, FLTrust, and Krum also display significant vulnerabilities to larger fraction of malicious agents, resulting in frequent collisions. Conversely, \alg consistently sustains a 100\% no-collision rate even when the fraction of malicious agents is 40\%. 
This outstanding performance is due to \algns's ability to detect the impact of malicious gradients and augmented experiences from DTs, thereby preserving the integrity of the decision-making. 
With over 50\% malicious agents, \alg experiences a significant drop in the no-collision rate. However, such a high proportion of malicious agents is unlikely in real-world scenarios.
%

\begin{figure*}[!htbp]
    \centering
    \includegraphics[width=0.62\textwidth]{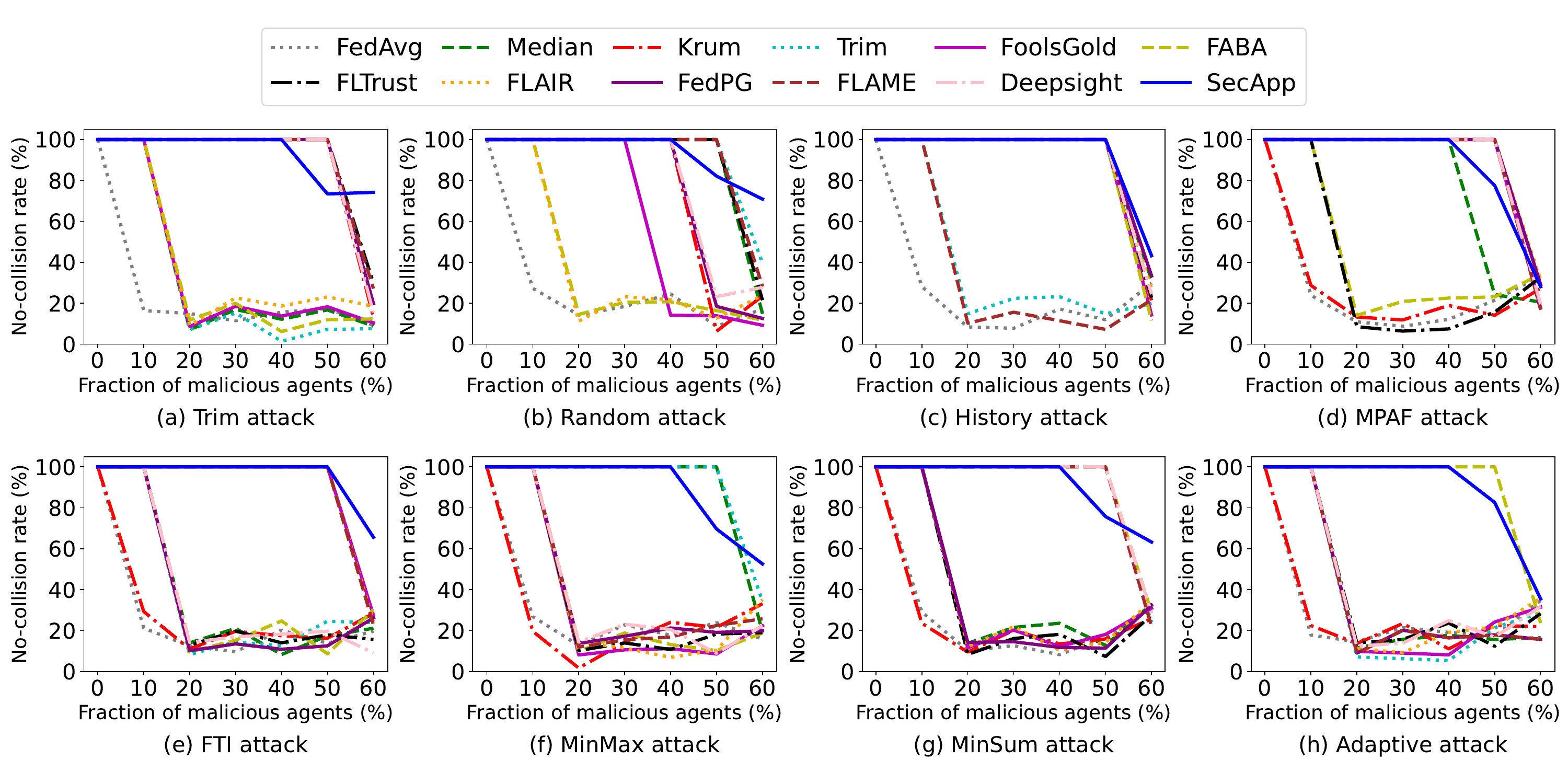}
    \caption{Impact of fraction of malicious agents.}
    \label{fig:fraction}
    \vspace{-0.1in}
\end{figure*}

\myparatight{Impact of scaling factor $\lambda$}
The selection of the scaling factor $\lambda$ from Eq.~(2) in \alg is critical for optimal defensive performance, as it controls the allowable rate of change in model updates, 
preventing excessive deviations and aiding in the identification of benign gradients. From Table~\ref{tab:scaling}, it is observed that when $\lambda$ is set too high, such as 1,000, the system tends to include too many system parameters in the aggregation process. This over-inclusion can allow malicious gradients to be included, leading to a significant drop in performance. For example, under the MinSum attack, the no-collision rate drops to 18.49\% when $\lambda = 1,000$, compared to 100.0\% when $\lambda$ is set to 10. This decline indicates that a large $\lambda$ compromises the system's ability to filter out harmful system parameters.
However, when $\lambda$ is set fairly low, such as 0.1, the system becomes overly conservative, i.e. unconsciously filtering out a substantial number of benign gradients. This leads to inadequate gradient aggregation and poor decision-making processes. 
As a result, an appropriately chosen $\lambda$ is critical to run \algns, aiming to balance the inclusion of enough benign gradients to make accurate decisions while effectively filtering out potential threats.

\begin{table}[htbp]
\renewcommand{\arraystretch}{1.2}
\centering
\caption{Performance on different scaling parameter $\lambda$.}
\footnotesize 
\begin{tabularx}{\columnwidth}{l|*{6}{>{\centering\arraybackslash}X}}
\toprule
Attack & 0.1 & 1 & 10 & 100 & 1000  \\
\midrule
No attack       & 100.0\% & 100.0\% & 100.0\% & 100.0\% & 100.0\%  \\
\hline
Trim attack     & 11.23\% & 17.89\% & 100.0\% & 100.0\% & 33.26\%  \\
\hline
Random attack   & 19.54\% & 22.13\% & 100.0\% & 100.0\% & 100.0\%  \\
\hline
History attack  & 12.47\% & 100.0\% & 100.0\% & 100.0\% & 100.0\%  \\
\hline
MPAF attack     & 23.45\% & 12.34\% & 100.0\% & 100.0\% & 19.42\%  \\
\hline
FTI attack      & 15.67\% & 100.0\% & 100.0\% & 100.0\% & 100.0\%  \\
\hline
MinMax attack   & 13.89\% & 100.0\% & 100.0\% & 100.0\% & 100.0\%  \\
\hline
MinSum attack   & 18.76\% & 15.43\% & 100.0\% & 22.54\% & 18.49\%  \\
\hline
Adaptive attack & 20.32\% & 100.0\% & 100.0\% & 34.72\% & 23.29\%  \\
\bottomrule
\end{tabularx}
\label{tab:scaling}
\end{table}

\myparatight{Performance against more complicated attacks}
The ``a little is enough'' (LIE) attack~\cite{baruch2019little}  and Krum attack~\cite{fang2020local} are more adaptive and complicated poisoning attack strategies designed to exploit vulnerabilities in FL systems by subtly manipulating gradients to degrade overall performance without being detected by conventional anomaly detection mechanisms. The LIE attack injects minor, seemingly harmless changes into gradients that accumulate over time, significantly distorting the final aggregated model. This attack is particularly challenging to be detected because the individual gradients closely resemble normal ones, making it difficult for basic statistical methods to identify them as malicious. Conversely, the Krum attack targets Byzantine-robust aggregation methods by crafting malicious gradients that are sufficiently similar to the majority to be included in the aggregation process, yet are strategically designed to degrade the model’s performance. 
%
Table~\ref{tab:adaptive} illustrates the resilience of \alg against LIE and Krum attacks. 
These attacks are specifically engineered to circumvent traditional defense mechanisms by subtly altering gradients in ways that evade detection by basic anomaly detection methods. 
For example, under the LIE attack, \alg retains a perfect 100.0\% no-collision rate, whereas methods like FoolsGold and FLAIR, which employ less sophisticated detection algorithms, see their no-collision rates drop to 20.02\% and 18.71\%, respectively. This sharp contrast underscores \algns's effectiveness in identifying and neutralizing even well-camouflaged malicious gradients.


\begin{table*}[t]
\renewcommand{\arraystretch}{1.2}
\centering
\caption{Performance on advanced attacks.}
\footnotesize
\begin{tabularx}{0.95\textwidth}{l|*{12}{>{\centering\arraybackslash}X}}
\toprule
Attack & FedAvg & Median & Trim & Krum & FoolsGold & FABA & FLTrust & FLAIR & FedPG & FLAME & Deepsight & \textbf{\alg} \\
\midrule
No attack & 100.0\% & 100.0\% & 100.0\% & 100.0\% & 100.0\% & 100.0\% & 100.0\% & 100.0\% & 100.0\% & 100.0\% & 100.0\% & \textbf{100.0\%} \\
\hline
LIE attack & 11.48\% & 10.17\% & 23.17\% & 27.29\% & 20.02\% & 100.0\% & 100.0\% & 18.71\% & 100.0\% & 29.24\% & 100.0\% & \textbf{100.0\%} \\
\hline
Krum attack & 28.57\% & 23.97\% & 12.68\% & 22.86\% & 100.0\% & 10.27\% & 100.0\% & 10.99\% & 17.06\% & 100.0\% & 100.0\% & \textbf{100.0\%} \\
\bottomrule
\end{tabularx}
\label{tab:adaptive}
\end{table*}

\vspace{-0.02in}

\myparatight{More experiments}
Additional experiments and ablation study can be found in Appendix D, including how \alg reacts to backdoor attacks,
an investigation of the impact of the total number of vehicle agents, the impact of threshold $\psi$, and the effectiveness of adaptive scaling factor $\lambda$.
The application of our \alg in an edge caching scenario is presented in Appendix E.
We also discuss the limitation of \alg in Appendix G.

%% file: sec/discussion.tex
\section{Computational Efficiency}
\label{limitations}
Fig.~\ref{fig:cost} illustrates the computational costs associated with various defensive schemes, where the computational cost refers to the time required to aggregate gradients over 2,000 rounds.
Note that the training process is conducted on a server equipped with NVIDIA RTX 4090 and INTEL 13900K.
As illustrated in Fig.~\ref{fig:cost}, \alg incurs a lower cost of 10 seconds, which is significantly lower than more complex techniques such as FABA and FLTrust, which are 46 and 48 seconds, respectively.
Overall, \alg incurs minimal overhead compared to generic approaches like FedAvg. This demonstrates its ability to balance efficiency and robustness, making it well-suited for secure, large-scale FRL systems.

\begin{figure}
    \centering
    \includegraphics[width=0.4\textwidth]{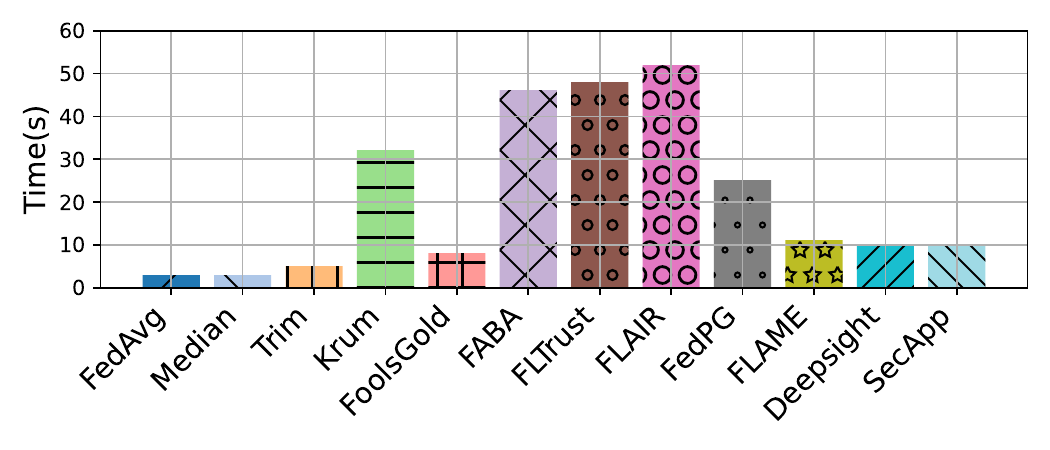}
    \caption{Computational cost of different schemes.}
    \label{fig:cost}
\end{figure}

%% file: sec/related.tex
\section{Related Works}

\paragraph*{\textbf{Model Poisoning Attacks}} Model poisoning attacks significantly compromise the robustness and reliability of FL systems, as documented in several studies~\cite{cao2022mpaf,shejwalkar2022back,shejwalkar2021manipulating,2018How,baruch2019little,xie2019dba,fang2020local,zhang2024poisoning}. For instance, the research in~\cite{fang2020local} reveals how attackers can leverage Byzantine fault tolerance mechanisms in FL to manipulate model updates. 
Further, the study in~\cite{cao2022mpaf} demonstrates the escalation of this threat by introducing fake clients that exacerbate the damage through scaled attacks. Additionally, the work in~\cite{shejwalkar2021manipulating} focuses on optimizing these poisoning strategies to maximize their disruptive impact on FL systems.
However, these sophisticated attack techniques remain under-explored in safety-critical FRL mobile systems.



\paragraph*{\textbf{Byzantine-robust Aggregation Rules}} Secure aggregation rules~\cite{Yin18,blanchard2017machine,fung2018mitigating,xia2019faba,cao2020fltrust,sharma2023flair,fan2021fault,nguyen2022flame,rieger2022deepsight} in FL are essential for maintaining the integrity and robustness of the learning process, especially in the presence of poisoning attacks. 
However, the standard approach FedAvg~\cite{mcmahan2017communication} assumes honest participants, making it vulnerable to Byzantine attacks. The authors in~\cite{blanchard2017machine} addressed this by proposing a method named Krum that enhances the resilience of FL systems against malicious participants. These methods, while effective in standard FL scenarios, struggle in FRL-based systems where dynamic and continuous learning processes can be exploited by sophisticated adversaries.
%
In such systems, the impact of erroneous gradients can be more pronounced, potentially leading to unsafe behaviors in real-time applications. To address this gap, our work focuses on developing an innovative defense strategy that effectively mitigates security challenges in FRL systems, with both theoretical guarantee and experimental validation in safety-critical scenarios.

\paragraph*{\textbf{FL versus FRL}} Although FL and FRL both employ distributed training frameworks, their differences are pronounced in the context of mobile computing and systems~\cite{qi2021federated,zhuo2020federateddeepreinforcementlearning,jin2022federated,10.1145/3636534.3690665}.
In FL, each client processes locally collected data—often static or semi-static—to learn predictive models with or without labels. These clients typically operate under the assumption that their training data are independent and identically distributed (IID), and aggregation schemes assume a convex loss function~\cite{Yin18,cao2020fltrust}.
In contrast, FRL is inherently more dynamic, as each mobile agent interacts with its environment across multiple time steps, adapting its policy through trial and error while receiving rewards. This means agents are continuously adjusting to varying computational conditions, device mobility, and environmental constraints, leading to highly heterogeneous trajectories. The resulting non-IID and non-convex nature of FRL training makes it significantly more challenging to establish theoretical convergence guarantees, especially in dynamic environments where network topology, latency, and resource availability fluctuate.

%% file: sec/conclusion.tex

\section{Conclusion}
In this work, we introduced \algns, a robustness-aware mechanism designed to mitigate the effects of poisoning attacks, particularly applied in safety-critical scenarios such as AD. By employing a multi-step strategy, \alg effectively identifies and excludes malicious information from distributed agents, thereby enhancing the integrity of the decision-making system. Our theoretical analysis, supported by extensive evaluations using DTs, demonstrates that \alg is robust against various attacks 
in safety-critical systems.


%% file: sec/appendix.tex
\appendices

\begin{algorithm}[t]
\caption{FRL training process without SVRG procedure.}
\label{alg:server_without_scsg}
\begin{algorithmic}[1]
\State {\bfseries Input:} $\tilde{\mathbf{w}}_{0} \in \mathbb{R}^{e}$,  batch size $\xi_t$, step size $\eta_{t}$
\For{$t=1$ {\bfseries to} $T$}
   \State $\mathbf{w}_{t}^{0} \leftarrow \tilde{\mathbf{w}}_{t-1}$ 
   \For{$k=1$ {\bfseries to} $K$ \text{in parallel}}
       \State $\mu_{t}^{k} = \text{LocalTraining}(t, k, \mathbf{w}_{t}^{0}, \xi_t, d(\cdot|\mathbf{w}_t^0))$
   \EndFor
   \State $\mu_t \leftarrow \frac{1}{K} \sum_{k=1}^{K} \mu_t^{k}$ \Comment {FedAvg~\cite{mcmahan2017communication}}
    \State $\tilde{\mathbf{w}}_{t} = \mathbf{w}_{t}^{0}  + \eta_{t} \mu_t$
\EndFor
\State {\bfseries Output: $\tilde{\mathbf{w}}_{T}$} 
\end{algorithmic}
\end{algorithm}

\section{Proof of Theorem \ref{main-theorem}}\label{appendix:proof-of-theorem}
Theorem \ref{main-theorem} is the main result to illustrate the convergence performance of our proposed \algns. 
In the following proof, we adopt a similar setting from \cite{lei2017non} to set the mini-batch size $b_t = 1$. In this section, we first present the proofs of several key lemmas, followed by the proof of Theorem \ref{main-theorem}.
\begin{lemma}[Function Smoothness]
    \label{proposition-uai-grad}
    Under Assumption \ref{assumption-policy-derivatives}, the function $\mathcal{L}(\boldsymbol{\mathbf{w}})$ is
    smooth with a smoothness constant $L_{\ell}$. Let $g(\upsilon | \boldsymbol{\mathbf{w}})$ denote the gradient estimators. Then, for any $\boldsymbol{\mathbf{w}}, \boldsymbol{\mathbf{w}}_1, \boldsymbol{\mathbf{w}}_2 \in \mathbb{R}^e$, the following inequalities hold:
    \begin{align*}
        \|g(\upsilon|\boldsymbol{\mathbf{w}})\| &\leq Z_g, \\
        \left\|g\left(\upsilon \mid \boldsymbol{\mathbf{w}}_{1}\right) - g\left(\upsilon \mid \boldsymbol{\mathbf{w}}_{2}\right)\right\| &\leq L_{g}\left\|\boldsymbol{\mathbf{w}}_{1} - \boldsymbol{\mathbf{w}}_{2}\right\|,
    \end{align*}
    where $L_{\ell}$, $L_g$, $Z_g$, and $Z_b$ are defined in Section~\ref{sec:theoretical}.
\end{lemma}

Lemma \ref{proposition-uai-grad} is key to obtaining a fast convergence rate, and the detailed proof is provided in \cite{xu2020improved}.

\begin{lemma}[Bounded variance of the Algorithm \ref{alg:our}'s update with importance sampling]
\label{lemma-E-v}
The variance of the Algorithm \ref{alg:our}'s update is bounded as follows:
\begin{align}
    \mathbb{E}_{\upsilon_t^n}[\|\zeta_t^n\|^2] &\leq 2\Phi \left\|\mathbf{w}_{t}^{n}-\mathbf{w}_{t}^{0}\right\|^{2} 
    + 2\left\|\nabla \mathcal{L}(\mathbf{w}_{t}^{n})\right\|^{2} \notag\\
    &\quad + 2\|\frac{1}{\left|\mathcal{W}_{t}\right|} \sum_{k \in \mathcal{W}_{t}} \mu_{t}^{k}-\nabla \mathcal{L}\left(\mathbf{w}_{t}^{0}\right)\|^{2},\notag
\end{align}
where $\Phi = L_g + Z_g^2 Z_w$ and $\gamma$ are defined in Section~\ref{sec:theoretical} and $n \in \mathcal{N}_t$.
\end{lemma}
\begin{proof}
According to Algorithm \ref{alg:our}, we have:
\begin{align*}
\zeta_{t}^{n} = \frac{1}{b_t} \sum_{j=1}^{b_t} [g(\upsilon_{t,j}^{n}|\bm{\mathbf{w}}^{n}_{t}) - \delta(\upsilon_{t,j}^n|\mathbf{w}_{t}^{n}, \mathbf{w}_{t}^{0}) g(\upsilon_{t,j}^{n}|\mathbf{w}_{t}^{0})] + \mu_t.
\end{align*}
We start the proof by setting $b_t = 1$ and $\mu_t= \frac{1}{\left|\mathcal{W}_{t}\right|} \sum_{k \in \mathcal{W}_{t}} \mu_{t}^{k}$, following the suggestion from \cite{lei2017non}, to obtain more precise and theoretically sound outcomes.  Thus, we have
\begin{align*}
\zeta_{t}^{n} = g(\upsilon_{t}^{n} \mid \mathbf{w}_{t}^{n})-\delta(\upsilon_{t}^{n} \mid \mathbf{w}_{t}^{n}, \mathbf{w}_{t}^{0}) g(\upsilon_{t}^{n} \mid \mathbf{w}_{t}^{0}) \\
+\frac{1}{\left|\mathcal{W}_{t}\right|} \sum_{k \in \mathcal{W}_{t}} \mu_{t}^{k}.\notag
\end{align*}
Next, the unbiasedness of the importance sampling can be shown as follows:
\begin{align*}
g(\upsilon_n | \mathbf{w}_0) =& \mathbb{E}_{\upsilon \sim d(\cdot | \mathbf{w}_0)} \left[ \nabla_{\mathbf{w}_0} d(\mathbf{w}_0) r(\upsilon) \right]\notag\\=& \int \frac{d(\cdot | \mathbf{w}_0)}{d(\cdot | \mathbf{w}_n)} d(\cdot | \mathbf{w}_n) \nabla_{\mathbf{w}_0} d(\mathbf{w}_0) r(\upsilon) d\upsilon\notag\\=& \mathbb{E}_{\upsilon \sim d(\cdot | \mathbf{w}_0)} \left[ \frac{d(\cdot | \mathbf{w}_n)}{d(\cdot | \mathbf{w}_0)} g(\upsilon | \mathbf{w}_0) \right].
\end{align*}
Then we have,
\[
\delta(\upsilon | \mathbf{w}_n, \mathbf{w}_0) g(\upsilon_n | \mathbf{w}_0) = g(\upsilon_0 | \mathbf{w}_0).
\]
%
%
%
%
Further, we have
$
\mathbb{E}_{\upsilon_t^n} [\zeta_{t}^{n}] 
=\nabla \mathcal{L}(\mathbf{w}_{t}^{n})-\nabla \mathcal{L}\left(\mathbf{w}_{t}^{0}\right)+\mu_{t} 
$
Given that $\nabla \mathcal{L}(\mathbf{w}_{t}^{n})-\nabla \mathcal{L}(\mathbf{w}_{t}^{0}) = \\
\mathbb{E}_{\upsilon_t^n}[g(\upsilon_{t}^{n} \mid \mathbf{w}_{t}^{n})-\delta(\upsilon_{t}^{n} \mid \mathbf{w}_{t}^{n}, \mathbf{w}_{t}^{0}) g(\upsilon_{t}^{n} \mid \mathbf{w}_{t}^{0})]$ and $\mathbb{E}\|X\|^2 = \mathbb{E}\|X-\mathbb{E}X\|^2 + \|\mathbb{E}X\|^2$, we further derive that:
\begin{align*}
&\mathbb{E}_{\upsilon_t^n} [\|\zeta_{t}^{n}\|^2] =\mathbb{E}_{\upsilon_t^n}\|\zeta_{t}^{n}-\mathbb{E}_{\upsilon_t^n} [\zeta_{t}^{n}]\|^{2} +\|\mathbb{E}_{\upsilon_t^n} [\zeta_{t}^{n}]\|^{2}\\
&=\mathbb{E}_{\upsilon_t^n}\|g(\upsilon_{t}^{n} \mid \mathbf{w}_{t}^{n})-\delta(\upsilon_{t}^{n} \mid \mathbf{w}_{t}^{n}, \mathbf{w}_{t}^{0}) g(\upsilon_{t}^{n} \mid \mathbf{w}_{t}^{0}) \\
&\quad +\mu_{t}-(\nabla \mathcal{L}(\mathbf{w}_{t}^{n})-\nabla \mathcal{L}(\mathbf{w}_{t}^{0})+\mu_{t})\|^{2}+\left\|\mathbb{E}_{\upsilon_t^n} [\zeta_{t}^{n}]\right\|^{2}\\
&=\mathbb{E}_{\upsilon_t^n}\left\|g(\upsilon_{t}^{n} \mid \mathbf{w}_{t}^{n})-\delta(\upsilon_{t}^{n} \mid \mathbf{w}_{t}^{n}, \mathbf{w}_{t}^{0}) g(\upsilon_{t}^{n} \mid \mathbf{w}_{t}^{0})\right. \\
&\quad -(\nabla \mathcal{L}(\mathbf{w}_{t}^{n})-\nabla \mathcal{L}(\mathbf{w}_{t}^{0}))\|^{2} \\
&\quad +\left\|\nabla \mathcal{L}(\mathbf{w}_{t}^{n})-\nabla \mathcal{L}\left(\mathbf{w}_{t}^{0}\right)+\mu_{t}\right\|^{2}.
\end{align*}
We derive the following equations using inequalities of \\
$\mathbb{E}\|X\|^2 \geq \mathbb{E}\|X-\mathbb{E}X\|^2$ and $\|X_1\|^2 + \|X_2\|^2 \geq \frac{\|X_1 + X_2\|^2 }{2}$.
\begin{align}
&\mathbb{E}_{\upsilon_t^n}\left\|g(\upsilon_{t}^{n} \mid \mathbf{w}_{t}^{n})-\delta(\upsilon_{t}^{n} \mid \mathbf{w}_{t}^{n}, \mathbf{w}_{t}^{0}) g(\upsilon_{t}^{n} \mid \mathbf{w}_{t}^{0})\right. \notag\\
&\quad -(\nabla \mathcal{L}(\mathbf{w}_{t}^{n})-\nabla \mathcal{L}(\mathbf{w}_{t}^{0}))\|^{2}+\left\|\nabla \mathcal{L}(\mathbf{w}_{t}^{n})-\nabla \mathcal{L}\left(\mathbf{w}_{t}^{0}\right)+\mu_{t}\right\|^{2}\notag\\
&\leq \mathbb{E}_{\upsilon_t^n}\left\|g(\upsilon_{t}^{n} \mid \mathbf{w}_{t}^{n})-\delta(\upsilon_{t}^{n} \mid \mathbf{w}_{t}^{n}, \mathbf{w}_{t}^{0}) g(\upsilon_{t}^{n} \mid \mathbf{w}_{t}^{0})\right\|^{2} \notag\\
&\quad + 2\left\|\nabla \mathcal{L}(\mathbf{w}_{t}^{n})\right\|^{2}+2\left\|\mu_{t}-\nabla \mathcal{L}\left(\mathbf{w}_{t}^{0}\right)\right\|^{2}.\label{v-bound-1}
\end{align}
Additionally, we have
\begin{align}
& \quad \mathbb{E}_{\upsilon_t^n}\left\|g(\upsilon_{t}^{n} \mid \mathbf{w}_{t}^{n})  -\delta(\upsilon_{t}^{n} \mid \mathbf{w}_{t}^{n}, \mathbf{w}_{t}^{0}) g(\upsilon_{t}^{n} \mid \mathbf{w}_{t}^{0})\right\|^{2} \notag\\
&=\mathbb{E}_{\upsilon_t^n} \|g(\upsilon_t^n \mid \mathbf{w}_{t}^n) + g(\upsilon_{t}^{n} \mid \mathbf{w}_{t}^{0}) \notag\\
&\quad - g(\upsilon_{t}^{n} \mid \mathbf{w}_{t}^{0}) -\delta(\upsilon_{t}^{n} \mid \mathbf{w}_{t}^{n}, \mathbf{w}_{t}^{0}) g(\upsilon_t^n \mid \mathbf{w}_t^0) \|^{2} \notag\\
&=\mathbb{E}_{\upsilon_t^n} \|g(\upsilon_t^n \mid \mathbf{w}_{t}^n)  - g(\upsilon_{t}^{n} \mid \mathbf{w}_{t}^{0}) \notag\\
&\quad + (1-\delta(\upsilon_{t}^{n} \mid \mathbf{w}_{t}^{n}, \mathbf{w}_{t}^{0})) g(\upsilon_t^n \mid \mathbf{w}_t^0) \|^{2} \notag\\
&\leq2 \mathbb{E}_{\upsilon_t^n}\left\|g(\upsilon_t^n \mid \mathbf{w}_{t}^n)-g(\upsilon_t^n \mid \mathbf{w}_t^0)\right\|^{2} \notag\\
&\quad + 2 \mathbb{E}_{\upsilon_t^n}\left\|(1-\delta(\upsilon_{t}^{n} \mid \mathbf{w}_{t}^{n}, \mathbf{w}_{t}^{0})) g(\upsilon_t^n \mid \mathbf{w}_t^0)\right\|^{2}.
\label{v-bound-2}
\end{align}
Assuming assumptions \ref{assumption-uai-weight} and \ref{assumption-policy-derivatives} hold true and $\delta(\upsilon \mid \mathbf{w}_1, \mathbf{w}_2) = \frac{d(\upsilon \mid \mathbf{w}_1)}{d(\upsilon \mid \mathbf{w}_2)}, Z_w = H(Q + 1)(2HG^2 + M)$. Then, we have 
\begin{align}\label{aaaa}
  &  \text{Var}(\delta(\upsilon \mid \mathbf{w}_1, \mathbf{w}_2))  \leq  Z_w \|\mathbf{w}_1 - \mathbf{w}_2\|^2,\notag\\
&   \text{Var}_{\mathbf{w}_t^n, \mathbf{w}_t^0}(\delta(\upsilon_t^n \mid \mathbf{w}_t^n, \mathbf{w}_t^0))
    \leq  Z_w \|\mathbf{w}_t^n - \mathbf{w}_t^0\|^2.
\end{align}
The detailed proof of \eqref{aaaa} is provided in \cite{xu2020improved}.
By combining \eqref{v-bound-1} and \eqref{v-bound-2}, we obtain:
\begin{align*}
&\mathbb{E}_{\upsilon_t^n} [\|\zeta_{t}^{n}\|^2] \leq 2 \mathbb{E}_{\upsilon_t^n}\left\|g(\upsilon_t^n \mid \mathbf{w}_{t}^n)-g(\upsilon_t^n \mid \mathbf{w}_t^0)\right\|^{2} \\
&\quad +2 \mathbb{E}_{\upsilon_t^n}\left\|(1-\delta(\upsilon_t^n \mid \mathbf{w}_t^n,\mathbf{w}_t^0)) g(\upsilon_t^n \mid \mathbf{w}_t^0)\right\|^{2} \\
&\quad +2\left\|\nabla \mathcal{L}(\mathbf{w}_{t}^{n})\right\|^{2} + 2\left\|\mu_{t}-\nabla \mathcal{L}\left(\mathbf{w}_{t}^{0}\right)\right\|^{2} \\
&\leq 2L_g\left\|\mathbf{w}_{t}^{n}-\mathbf{w}_{t}^{0}\right\|^{2} + 2 Z_{g}^{2} \mathbb{E}_{\upsilon_t^n}\|(1-\delta(\upsilon_t^n \mid \mathbf{w}_t^n,\mathbf{w}_t^0))\|^{2} \\
&\quad + 2\left\|\nabla \mathcal{L}(\mathbf{w}_{t}^{n})\right\|^{2} + 2\left\|\mu_{t}-\nabla \mathcal{L}\left(\mathbf{w}_{t}^{0}\right)\right\|^{2} \stepcounter{equation}
\\
&\leq 2L_g\left\|\mathbf{w}_{t}^{n}-\mathbf{w}_{t}^{0}\right\|^{2} + 2Z_g^{2} Z_w\left\|\mathbf{w}_{t}^{n}-\mathbf{w}_{t}^{0}\right\|^{2} \\
&\quad + 2\left\|\nabla \mathcal{L}(\mathbf{w}_{t}^{n})\right\|^{2} + 2\|\frac{1}{\left|\mathcal{W}_{t}\right|} \sum_{k \in \mathcal{W}_{t}} \mu_{t}^{k}-\nabla \mathcal{L}\left(\mathbf{w}_{t}^{0}\right)\|^{2} \stepcounter{equation}
\\
&=(2 L_{g}+2Z_g^{2} Z_w)\left\|\mathbf{w}_{t}^{n}-\mathbf{w}_{t}^{0}\right\|^{2} \\
&\quad + 2\left\|\nabla \mathcal{L}(\mathbf{w}_{t}^{n})\right\|^{2} + 2\|\frac{1}{\left|\mathcal{W}_{t}\right|} \sum_{k \in \mathcal{W}_{t}} \mu_{t}^{k}-\nabla \mathcal{L}\left(\mathbf{w}_{t}^{0}\right)\|^{2}. \stepcounter{equation}
\end{align*}
We conclude the proof by noting that $2 L_{g}+2Z_g^{2} Z_w = 2\Phi$.
\end{proof}
%
%
%
%
%
%
%
%
%
%
%
%
%
%
%
%
%
%
%
\begin{lemma}[Gradient alignment analysis]
\label{lemma-eta-E}
We have the following equation hold,
\begin{align*}
    &\eta_{t} \mathbb{E}\left\langle \mu_{t}-\nabla \mathcal{L}\left(\mathbf{w}_{t}^{0}\right), \mathbb{E} \nabla \mathcal{L}(\tilde{\mathbf{w}}_{t})\right\rangle \\
    &= \frac{1}{\xi_t} \mathbb{E}\left\langle \mu_{t}-\nabla \mathcal{L}\left(\mathbf{w}_{t}^{0}\right), \tilde{\mathbf{w}}_{t}-\tilde{\mathbf{w}}_{t-1}\right\rangle \\
    & \quad\quad- \eta_{t} \mathbb{E}\left\|\mu_{t}-\nabla \mathcal{L}\left(\mathbf{w}_{t}^{0}\right)\right\|^{2}.
\end{align*}    
\end{lemma}
\begin{proof}
Let $H_t^n = \langle \mu_{t}-\nabla \mathcal{L}\left(\mathbf{w}_{t}^{0}\right), \mathbf{w}_t^n - \mathbf{w}_t^0 \rangle$. We then have:
$
    H^{n+1}_t - H_t^n 
    = \eta_t \langle \mu_{t}-\nabla \mathcal{L}\left(\mathbf{w}_{t}^{0}\right), \zeta_t^n \rangle
$
Taking the expectation with respect to $\upsilon_t^n$, it follows that:
\begin{align*}
 &\mathbb{E}_{\upsilon_t^n}\left[H^{n+1}_{t}-H^{n}_{t}\right] = \eta_{t}\left\langle \mu_{t}-\nabla \mathcal{L}\left(\mathbf{w}_{t}^{0}\right),  \mathbb{E}_{\upsilon_t^n}[\zeta_{t}^{n}]\right\rangle \\
 &= \eta_{t}\left\langle \mu_{t}-\nabla \mathcal{L}\left(\mathbf{w}_{t}^{0}\right), \nabla \mathcal{L}(\mathbf{w}_{t}^{n})\right\rangle  + \eta_{t}\left\|\mu_{t}-\nabla \mathcal{L}\left(\mathbf{w}_{t}^{0}\right)\right\|^{2}.
\end{align*}
When taking the expectation with respect to the random variable, we replace the random gradient in the inner product with its average.
Define $\mathbb{E}_t$ as the expectation over all trajectories $\{\upsilon_t^1, \upsilon_t^2, \dots\}$, given $\mathcal{N}_t$. Since the trajectories are independent of $\mathcal{N}_t$, $\mathbb{E}_t$ effectively represents the expectation over $\{\upsilon_t^1, \upsilon_t^2, \dots\}$. Thus, we obtain:
\begin{align*}
    \mathbb{E}_{t}[H^{n+1}_{t}-H^{n}_{t}] = &\eta_{t}\left\langle \mu_{t}-\nabla \mathcal{L}\left(\mathbf{w}_{t}^{0}\right), \mathbb{E}_{t}\nabla \mathcal{L}(\mathbf{w}_{t}^{n})\right\rangle \\
    &+ \eta_{t}\left\|\mu_{t}-\nabla \mathcal{L}\left(\mathbf{w}_{t}^{0}\right)\right\|^{2}.
\end{align*}
Setting $n = N_t$ and taking the expectation with respect to $\mathcal{N}_t$, we have:
\begin{align*}
    \mathbb{E}^{\mathcal{N}_t}\mathbb{E}_t(H^{N_{t}+1}_{t}-H^{N_{t}}_{t}) = &\eta_{t} \left\langle \mu_{t}-\nabla \mathcal{L}\left(\mathbf{w}_{t}^{0}\right), \mathbb{E}^{\mathcal{N}_t}\mathbb{E}_t\nabla \mathcal{L}(\mathbf{w}^{N_{t}}_{t}) \right\rangle \\
    &+ \eta_{t}\left\|\mu_{t}-\nabla \mathcal{L}\left(\mathbf{w}_{t}^{0}\right)\right\|^{2}.
\end{align*}
By applying Fubini’s theorem and \\
$\mathbb{E}\left[D_N - D_{N+1}\right] = \left(1 - \frac{1}{\Gamma}\right)(\mathbb{E}[D_N] - D_0)$, where $\{D_n\}_{n \geq 0}$ is a sequence that $\mathbb{E}\|D_N\| < \infty$~\cite{lei2017non}, we proceed as follows:
\begin{align*}
    &\mathbb{E}^{\mathcal{N}_t}\mathbb{E}_t(H^{N_{t}+1}_{t}-H^{N_{t}}_{t}) 
    \notag\\= &-\mathbb{E}_{t}\mathbb{E}^{\mathcal{N}_t}(H^{N_{t}}_{t}-H^{N_{t}+1}_{t}) \\
    = &\left(1-\frac{1}{\frac{\xi_t}{\xi_t+1}}\right)\left(H_t^0 - \mathbb{E}^{\mathcal{N}_t}\mathbb{E}_t H^{\mathcal{N}_t}_t\right) \\
    =& \frac{1}{\xi_t} \mathbb{E}^{\mathcal{N}_t}\mathbb{E}_t\left\langle \mu_{t}-\nabla \mathcal{L}\left(\mathbf{w}_{t}^{0}\right), \tilde{\mathbf{w}}_{t}-\tilde{\mathbf{w}}_{t-1}\right\rangle \\
    = &\eta_{t}\left\langle \mu_{t}-\nabla \mathcal{L}\left(\mathbf{w}_{t}^{0}\right), \mathbb{E}^{\mathcal{N}_t}\mathbb{E}_t \nabla \mathcal{L}(\mathbf{w}^{N_{t}}_{t})\right\rangle \\
    &+ \eta_{t}\left\|\mu_{t}-\nabla \mathcal{L}\left(\mathbf{w}_{t}^{0}\right)\right\|^{2}.
\end{align*}
Taking the expectation over the entire history concludes the lemma.
\end{proof}
\begin{lemma}[Bound on the Error Term]
\label{lemma-error-bound}
Consider Algorithm \ref{alg:our}. The following bound holds for $\mathbb{E}\|\mu_{t}-\nabla \mathcal{L}\left(\mathbf{w}_{t}^{0}\right)\|^2$:
\begin{align*}
     \mathbb{E} [\| \mu_{t}-\nabla \mathcal{L}\left(\mathbf{w}_{t}^{0}\right) \|^2 ] \leq 2\psi^2 + 64 \sigma^2 \frac{ V}{\xi_t}.
\end{align*}
Here, $\xi_t$ and $\theta$ are used to determine the filtering threshold $\psi=2 \sigma \sqrt{\frac{V}{\xi_t}}$, where $V=2\log\left(\frac{2K}{\theta}\right)$ and $\theta \in (0,1)$.
\end{lemma}
\begin{proof}
We aim to bound the error term $\mu_{t}-\nabla \mathcal{L}\left(\mathbf{w}_{t}^{0}\right)$. We begin with:
\begin{align*}
  &\mathbb{E} [\| \mu_{t}-\nabla \mathcal{L}\left(\mathbf{w}_{t}^{0}\right) \|^2 ] \\
  & = \mathbb{E} \| \mu_t - \nabla f(x_t) \|^2 \\
  & \overset{(a)}{=} \mathbb{E} \Big\| \frac{1}{|\mathcal{W}_{t}|} \sum_{k \in \mathcal{W}_{t}}  \mu_t^{k} - \nabla f(x_t) \Big\|^2 \\
  & = \mathbb{E} \Big\| \frac{1}{|\mathcal{W}_{t}|} \sum_{k \in \mathcal{W}_{t}}  \mu_t^{k} - \mu_{t}^{\mathcal{S}} + \mu_{t}^{\mathcal{S}} - \nabla f(x_t) \Big\|^2 \\
  & \overset{(b)}{\leq} 2 \underbrace{\mathbb{E} \Big\| \frac{1}{|\mathcal{W}_{t}|} \sum_{k \in \mathcal{W}_{t}}  \mu_t^{k} - \mu_{t}^{\mathcal{S}} \Big\|^2}_{\text{E1}} + 2 \underbrace{\| \mu_{t}^{\mathcal{S}} - \nabla f(x_t)  \|^2}_{\text{E2}}.
\end{align*}
Here, step $(a)$ follows from the definition of $\mu_t$ in Algorithm \ref{alg:our}, and step $(b)$ follows from the inequality $\| \sum_{i = 1}^n a_i \|^2 \leq n \sum_{i = 1}^n \| a_i \|^2$. Next, we analyze Terms E1 and E2 separately.
For Term E1, we have:
\begin{align*}
    \text{E1} & = \mathbb{E}  \Big\| \frac{1}{|\mathcal{W}_{t}|} \sum_{k \in \mathcal{W}_{t}}  \mu_t^{k} - \mu_{t}^{\mathcal{S}}  \Big\|^2  \\
    & \overset{(c)}{\leq} \frac{1}{|\mathcal{W}_{t}|} \sum_{k \in \mathcal{W}_{t}} \mathbb{E} \| \mu_t^{k} - \mu_{t}^{\mathcal{S}}  \|^2 \\
    & \overset{(d)}{\leq} \psi^2.
\end{align*}
Here, step $(c)$ again follows from the inequality $\| \sum_{i = 1}^n a_i \|^2 \leq n \sum_{i = 1}^n \| a_i \|^2$, and step $(d)$ follows from our proposed Byzantine filtering rule and the choice of the filtering constant $\lambda \leq \frac{\psi - \| \mu^{\mathcal{S}}_t - \mu_{t-1} \|}{\|\mu^{\mathcal{S}}_t - \mu_{t-1} \|}$.
Given inequalities
$\left\|\mu_{t}^{k}-\mu_{t-1}\right\| \leq \lambda \left\|  \mu_{t}^{\mathcal{S}} - \mu_{t-1}\right\| \\ \le \frac{\psi - \| \mu^{\mathcal{S}}_t - \mu_{t-1} \|}{\|\mu^{\mathcal{S}}_t - \mu_{t-1} \|} \times \left\|  \mu_{t}^{\mathcal{S}} - \mu_{t-1}\right\| = \psi - \| \mu^{\mathcal{S}}_t - \mu_{t-1} \|$,
and using the triangle inequality, we obtain:
\begin{align*}
 \|\mu_{t}^{k}-\mu_{t-1}\| & \leq \psi - \| \mu^{\mathcal{S}}_t - \mu_{t-1} \| \\
 \|\mu_{t}^{k}-\mu_{t-1}\| + \| \mu^{\mathcal{S}}_t - \mu_{t-1} \| & \leq \psi \\
\|\mu^{\mathcal{S}}_t - \mu_t^{k} \| & \leq \psi.
\end{align*}
We assume that $\psi > \| \mu^{\mathcal{S}}_t - \mu_{t-1} \|$. 
Now, consider Term E2. Note that with high probability, $\| \mu_{t}^{\mathcal{S}} - \nabla f(x_t)  \|$ is bounded by $2 \psi$, and it is bounded by $8 \sigma$ almost surely (see \cite[Lemma 8]{Bulusu_TSPIN_2021}). To bound E2, we define the following events:
\textbf{Event A:} We define $\| \mu_{t}^{\mathcal{S}} - \nabla f(x_t)  \| \leq 2 \psi$ as Event A, and its complement, Event A!, is defined as $\| \mu_{t}^{\mathcal{S}} - \nabla f(x_t)  \| > 2 \psi$. We have $\mathbb{P}[\text{Event A}] \geq 1 - \theta$ and $\mathbb{P}[\text{Event A!}] \leq \theta$.
Using these events, we can express E2 as:
\begin{align*}
    \text{E2} & = \mathbb{E} \big[ \| \mu_{t}^{\mathcal{S}} - \nabla f(x_t)  \|^2 \big]\\
    & = \mathbb{P}[\text{Event A}]  \mathbb{E} \big[ \| \mu_{t}^{\mathcal{S}} - \nabla f(x_t)  \|^2 \big| \text{Event A} \big] \\
    &\quad + \mathbb{P}[\text{Event A!}]  \mathbb{E} \big[ \| \mu_{t}^{\mathcal{S}} - \nabla f(x_t)  \|^2 \big| \text{Event A!} \big] \\
    & \leq \mathbb{E} \big[ \| \mu_{t}^{\mathcal{S}} - \nabla f(x_t)  \|^2 \big| \text{Event A} \big] \\
    &\quad + \theta \, \mathbb{E} \big[ \| \mu_{t}^{\mathcal{S}} - \nabla f(x_t)  \|^2 \big| \text{Event A!} \big] \\
    & \leq 4 \psi^2 + 64 \theta \sigma^2.
\end{align*}
Using the definition $\psi = 2 \sigma \sqrt{\frac{V}{\xi_t}}$ and choosing $\theta = \frac{V}{4\xi_t}$, we get:
$$\text{E2} = \mathbb{E} \big[ \| \mu_{t}^{\mathcal{S}} - \nabla f(x_t)  \|^2 \big] \leq 32 \sigma^2 \frac{ V}{\xi_t}.$$
Combining E1 and E2 through step $(b)$, we conclude:
\begin{align*}
     \mathbb{E} [\| \mu_{t}-\nabla \mathcal{L}\left(\mathbf{w}_{t}^{0}\right) \|^2 ] \leq 2\psi^2 + 64 \sigma^2 \frac{ V}{\xi_t}.
\end{align*}

\end{proof}

\begin{proof}

Starting the proof of Theorem~\ref{main-theorem} from the update equation $\mathbf{w}^{n+1}_t = \mathbf{w}_t^n + \eta_t \zeta_t^n$, we have:
\begin{align*}
&\mathbb{E}_{\upsilon_t^n}\|\mathbf{w}^{n+1}_t - \mathbf{w}_t^0\|^2 
    \notag\\&= \mathbb{E}_{\upsilon_t^n}\|\mathbf{w}_{t}^n - \mathbf{w}_t^0 + \eta_t \zeta_t^n\|^2 \\
    &= 2\eta_t \langle\mathbb{E}_{\upsilon_t^n}[\zeta_t^n], \mathbf{w}_t^n - \mathbf{w}_t^0\rangle+ \|\mathbf{w}_t^n - \mathbf{w}_t^0\|^2 + \eta_t^2\mathbb{E}_{\upsilon_t^n}\|\zeta_t^n\|^2 \\
    &\leq \eta_t^2[(2L_g + 2Z_g^2Z_w)\|\mathbf{w}_t^n-\mathbf{w}_t^0\|^2 \\
    &\quad + 2\|\nabla \mathcal{L}(\mathbf{w}_t^n)\|^2 + 2\|\mu_{t}-\nabla \mathcal{L}\left(\mathbf{w}_{t}^{0}\right)\|^2] \\
    &\quad + 2\eta_t \langle \mu_{t}-\nabla \mathcal{L}\left(\mathbf{w}_{t}^{0}\right), \mathbf{w}_t^n-\mathbf{w}_t^0 \rangle \\
    &\quad + 2\eta_t \langle\nabla \mathcal{L}(\mathbf{w}_t^n), \mathbf{w}_t^n-\mathbf{w}_t^0 \rangle + \|\mathbf{w}_t^n - \mathbf{w}_t^0\|^2 
    \stepcounter{equation}\tag{\theequation}\label{lemma-neg-2eta-E-1} \\
    &= [1+\eta_t^2(2L_g + 2Z_g^2Z_w)]\|\mathbf{w}_t^n-\mathbf{w}_t^0\|^2 \\
    &\quad + 2\eta_t \langle\nabla \mathcal{L}(\mathbf{w}_t^n), \mathbf{w}_t^n-\mathbf{w}_t^0 \rangle \\
    &\quad + 2\eta_t \langle \mu_{t}-\nabla \mathcal{L}\left(\mathbf{w}_{t}^{0}\right), \mathbf{w}_t^n-\mathbf{w}_t^0 \rangle \\
    &\quad + 2\eta_t^2\|\nabla \mathcal{L}(\mathbf{w}_t^n)\|^2 + 2\eta_t^2\|\mu_{t}-\nabla \mathcal{L}\left(\mathbf{w}_{t}^{0}\right)\|^2 
\end{align*}
where inequality \eqref{lemma-neg-2eta-E-1} follows from the bound on $\mathbb{E}_{\upsilon_t^n}\|\zeta_t^n\|^2$ derived earlier.
Next, let $\mathbb{E}_t$ denote the expectation with respect to all trajectories $\{\upsilon_t^1, \upsilon_t^2, \dots\}$, given $\mathcal{N}_t$. Since the trajectories are independent of $\mathcal{N}_t$, $\mathbb{E}_t$ can be viewed as the expectation over $\{\upsilon_t^1, \upsilon_t^2, \dots\}$. Therefore, we have:
\begin{align*}
    &\mathbb{E}_{t}\|\mathbf{w}^{n+1}_t - \mathbf{w}_t^0\|^2 \notag\\
    &\leq [(2Z_g^2Z_w+2L_g)\eta_t^2+1]\mathbb{E}_{t}\|\mathbf{w}_t^n-\mathbf{w}_t^0\|^2 \\
    &\quad + 2\eta_t\mathbb{E}_{t}\langle\nabla \mathcal{L}(\mathbf{w}_t^n), \mathbf{w}_t^n-\mathbf{w}_t^0 \rangle \\
    &\quad + 2\eta_t\mathbb{E}_{t}\langle \mu_{t}-\nabla \mathcal{L}\left(\mathbf{w}_{t}^{0}\right), \mathbf{w}_t^n-\mathbf{w}_t^0 \rangle \\
    &\quad + 2\eta_t^2\mathbb{E}_{t}\|\nabla \mathcal{L}(\mathbf{w}_t^n)\|^2 + 2\eta_t^2\|\mu_{t}-\nabla \mathcal{L}\left(\mathbf{w}_{t}^{0}\right)\|^2.
\end{align*}

Taking the expectation over $\mathcal{N}_t$ with $n = \mathcal{N}_t$ and following Fubini's theorem, we obtain:
\begin{align*}
    &-2 \eta_{t} \mathbb{E}_{N_{t}} \mathbb{E}_{t} \left \langle  \mu_{t}-\nabla \mathcal{L}\left(\mathbf{w}_{t}^{0}\right), \mathbf{w}^{\mathcal{N}_t}_t-\mathbf{w}_t^0\right\rangle \\
    &\leq \left[1+\eta_{t}^{2}(2L_g + 2Z_g^2Z_w)\right] \mathbb{E}_{N_{t}} \mathbb{E}_{t}\left\|\mathbf{w}^{N_{t}}_{t}-\mathbf{w}_{t}^{0}\right\|^{2} \\
    &\quad +2 \eta_{t} \mathbb{E}_{N_{t}} \mathbb{E}_{t}\left\langle \nabla \mathcal{L}(\mathbf{w}^{\mathcal{N}_t}_t), \mathbf{w}^{\mathcal{N}_t}_t-\mathbf{w}_t^0 \right\rangle \\
    &\quad + 2 \eta_{t}^{2} \mathbb{E}_{N_{t}} \mathbb{E}_{t}\left\|\nabla \mathcal{L}(\mathbf{w}^{\mathcal{N}_t}_t)\right\|^{2} + 2 \eta_{t}^{2}\left\|\mu_{t}-\nabla \mathcal{L}\left(\mathbf{w}_{t}^{0}\right)\right\|^{2} \\
    &\quad - \mathbb{E}_{N_{t}} \mathbb{E}_{t}\left\|\mathbf{w}^{N_{t}+1}_{t}-\mathbf{w}_{t}^{0}\right\|^{2} \\
    &= \left[(2Z_g^2Z_w+2L_g)\eta_{t}^{2} -\frac{1}{\xi_t}\right] \mathbb{E}_{N_{t}} \mathbb{E}_{t}\left\|\mathbf{w}^{\mathcal{N}_t}_t-\mathbf{w}_t^0\right\|^{2} \\
    &\quad +2 \eta_{t} \mathbb{E}_{N_{t}} \mathbb{E}_{t}\left\langle \nabla \mathcal{L}(\mathbf{w}^{\mathcal{N}_t}_t), \mathbf{w}^{\mathcal{N}_t}_t-\mathbf{w}_t^0 \right\rangle \\
    &\quad + 2 \eta_{t}^{2} \mathbb{E}_{N_{t}} \mathbb{E}_{t}\left\|\nabla \mathcal{L}(\mathbf{w}^{\mathcal{N}_t}_t)\right\|^{2} + 2 \eta_{t}^{2}\left\|\mu_{t}-\nabla \mathcal{L}\left(\mathbf{w}_{t}^{0}\right)\right\|^{2}.
    \stepcounter{equation}\tag{\theequation}\label{lemma-neg-2eta-E-2}
\end{align*}
After replacing replace $\mathbf{w}^{\mathcal{N}_t}_t$ with $\tilde{\mathbf{w}}_t$, and $\mathbf{w}_t^0$ with $\tilde{\mathbf{w}}_{t-1}$, taking the expectation over the entire equation, we have
\begin{align}\label{lemma-neg-2eta-E}
    &-2\eta_t\mathbb{E} \langle \mu_{t}-\nabla \mathcal{L}\left(\mathbf{w}_{t}^{0}\right), \tilde{\mathbf{w}}_t - \tilde{\mathbf{w}}_{t-1} \rangle \notag\\
    &\leq \left[\eta_t^2(2L_g + 2Z_g^2Z_w) - \frac{1}{\xi_t}\right]\mathbb{E}\|\tilde{\mathbf{w}}_t - \tilde{\mathbf{w}}_{t-1}\|^2 \notag\\
    &\quad + 2\eta_t \mathbb{E} \langle \nabla \mathcal{L}(\tilde{\mathbf{w}}_t), \tilde{\mathbf{w}}_t - \tilde{\mathbf{w}}_{t-1} \rangle + 2\eta^2_t\mathbb{E}\|\nabla \mathcal{L}(\tilde{\mathbf{w}}_t)\|^2  \notag\\
    &\quad + 2\eta_t^2\mathbb{E}\|\mu_{t}-\nabla \mathcal{L}\left(\mathbf{w}_{t}^{0}\right)\|^2.
\end{align}
From the $L_{\ell}$-smoothness of the objective function $\mathcal{L}(\boldsymbol{\mathbf{w}})$, we have
\begin{align*}
    &\mathbb{E}_{\upsilon_t^n}[\mathcal{L}(\boldsymbol{\mathbf{w}}^{n+1}_t)] 
    \notag\\\geq& \mathbb{E}_{\upsilon_t^n}\left[\mathcal{L}(\boldsymbol{\mathbf{w}}_t^n) + \langle \nabla \mathcal{L}(\boldsymbol{\mathbf{w}}_t^n), \boldsymbol{\mathbf{w}}^{n+1}_t - \boldsymbol{\mathbf{w}}_t^n \rangle \right.  \left. - \frac{L_{\ell}}{2}\|\boldsymbol{\mathbf{w}}^{n+1}_t - \boldsymbol{\mathbf{w}}_t^n\|^2 \right] \\
    =& \mathcal{L}(\boldsymbol{\mathbf{w}}_t^n) + \eta_t \langle \mathbb{E}_{\upsilon_t^n}[\zeta_t^n], \nabla \mathcal{L}(\boldsymbol{\mathbf{w}}_t^n)\rangle - \frac{L_{\ell}\eta_t^2}{2} \mathbb{E}_{\upsilon_t^n}[\|\zeta_t^n\|^2] \\
    \geq &\mathcal{L}(\boldsymbol{\mathbf{w}}_t^n) + \eta_t \langle \nabla \mathcal{L}(\boldsymbol{\mathbf{w}}_t^n)-\nabla \mathcal{L}\left(\mathbf{w}_{t}^{0}\right)+\mu_{t}, \nabla \mathcal{L}(\boldsymbol{\mathbf{w}}_t^n)\rangle \\
    &\quad - \frac{L\eta_t^2}{2}[(2L_g + 2 Z_g^2 Z_w)\|\boldsymbol{\mathbf{w}}_t^n - \boldsymbol{\mathbf{w}}_t^0\|^2 \\
    &\quad + 2\|\nabla \mathcal{L}(\boldsymbol{\mathbf{w}}_t^n)\|^2 + 2\|\mu_{t}-\nabla \mathcal{L}\left(\mathbf{w}_{t}^{0}\right)\|^2] \stepcounter{equation}\tag{\theequation}\label{main-smooth-1} \\
    =& \mathcal{L}(\boldsymbol{\mathbf{w}}_t^n) + \eta_t(1-L_{\ell}\eta_t)\|\nabla \mathcal{L}(\boldsymbol{\mathbf{w}}_t^n)\|^2 \\
    &\quad + \eta_t \langle \mu_{t}-\nabla \mathcal{L}\left(\mathbf{w}_{t}^{0}\right), \nabla \mathcal{L}(\boldsymbol{\mathbf{w}}_t^n) \rangle \\
    &\quad - L\eta_t^2(L_g + Z_g^2Z_w)\|\boldsymbol{\mathbf{w}}_t^n - \boldsymbol{\mathbf{w}}_t^0\|^2 \\
    &\quad - L\eta_t^2\|\mu_{t}-\nabla \mathcal{L}\left(\mathbf{w}_{t}^{0}\right)\|^2.
\end{align*}
where \eqref{main-smooth-1} follows from Lemma \ref{lemma-E-v}. Use $\mathbb{E}_t$ to denote the expectation with respect to all trajectories $\{\upsilon_t^1, \upsilon_t^2, ...\}$, given $\mathcal{N}_t$. Since $\{\upsilon_t^1, \upsilon_t^2, ...\}$ are independent of $\mathcal{N}_t$, $\mathbb{E}_t$ is equivalently the expectation with respect to $\{\upsilon_t^1, \upsilon_t^2, ...\}$. The above inequality gives
\begin{align*}
    \mathbb{E}_{t}[\mathcal{L}(\boldsymbol{\mathbf{w}}^{n+1}_t)] 
    &\geq \mathbb{E}_{t}[\mathcal{L}(\boldsymbol{\mathbf{w}}_t^n)] \\
    &\quad + \eta_t(1-L_{\ell}\eta_t)\mathbb{E}_{t}\|\nabla \mathcal{L}(\boldsymbol{\mathbf{w}}_t^n)\|^2 \\
    &\quad + \eta_t \mathbb{E}_{t}\langle \mu_{t}-\nabla \mathcal{L}\left(\mathbf{w}_{t}^{0}\right), \nabla \mathcal{L}(\boldsymbol{\mathbf{w}}_t^n) \rangle \\
    &\quad - L\eta_t^2(L_g + Z_g^2Z_w)\mathbb{E}_{t}\|\boldsymbol{\mathbf{w}}_t^n - \boldsymbol{\mathbf{w}}_t^0\|^2 \\
    &\quad - L\eta_t^2\|\mu_{t}-\nabla \mathcal{L}\left(\mathbf{w}_{t}^{0}\right)\|^2.
\end{align*}
Taking $n=\mathcal{N}_t$ and using $\mathbb{E}^{\mathcal{N}_t}$ to denote the expectation w.r.t. $\mathcal{N}_t$, we have from the above:
\begin{align*}
    \mathbb{E}^{\mathcal{N}_t}\mathbb{E}_{t}[\mathcal{L}(\boldsymbol{\mathbf{w}}^{\mathcal{N}_t+1}_t)] 
    &\geq \mathbb{E}^{\mathcal{N}_t}\mathbb{E}_{t}[\mathcal{L}(\boldsymbol{\mathbf{w}}^{\mathcal{N}_t}_t)] \\
    &\quad + \eta_t(1-L_{\ell}\eta_t)\mathbb{E}^{\mathcal{N}_t}\mathbb{E}_{t}\|\nabla \mathcal{L}(\boldsymbol{\mathbf{w}}^{\mathcal{N}_t}_t)\|^2 \\
    &\quad + \eta_t \mathbb{E}^{\mathcal{N}_t}\mathbb{E}_{t}\langle \mu_{t}-\nabla \mathcal{L}\left(\mathbf{w}_{t}^{0}\right), \nabla \mathcal{L}(\boldsymbol{\mathbf{w}}^{\mathcal{N}_t}_t) \rangle \\
    &\quad - L\eta_t^2(L_g + Z_g^2Z_w)\mathbb{E}^{\mathcal{N}_t}\mathbb{E}_{t}\|\boldsymbol{\mathbf{w}}^{\mathcal{N}_t}_t - \boldsymbol{\mathbf{w}}_t^0\|^2 \\
    &\quad - L\eta_t^2\|\mu_{t}-\nabla \mathcal{L}\left(\mathbf{w}_{t}^{0}\right)\|^2.
\end{align*}
Rearrange the terms,
\begin{align*}
    &\eta_t(1-L_{\ell}\eta_t)\mathbb{E}^{\mathcal{N}_t}\mathbb{E}_{t}\|\nabla \mathcal{L}(\boldsymbol{\mathbf{w}}^{\mathcal{N}_t}_t)\|^2 \\
    \leq &\mathbb{E}^{\mathcal{N}_t}\mathbb{E}_{t}[\mathcal{L}(\boldsymbol{\mathbf{w}}^{\mathcal{N}_t+1}_t)] 
    +L_{\ell}\eta_t^2(L_g + Z_g^2Z_w) 
    \mathbb{E}^{\mathcal{N}_t}\mathbb{E}_{t}\|\boldsymbol{\mathbf{w}}^{\mathcal{N}_t}_t - \boldsymbol{\mathbf{w}}_t^0\|^2 \\
    &- \eta_t \mathbb{E}^{\mathcal{N}_t}\mathbb{E}_{t}\langle \mu_{t}-\nabla \mathcal{L}\left(\mathbf{w}_{t}^{0}\right), \nabla \mathcal{L}(\boldsymbol{\mathbf{w}}^{\mathcal{N}_t}_t) \rangle 
    \notag\\&+L_{\ell}\eta_t^2\|\mu_{t}-\nabla \mathcal{L}\left(\mathbf{w}_{t}^{0}\right)\|^2 
    - \mathbb{E}^{\mathcal{N}_t}\mathbb{E}_{t}[\mathcal{L}(\boldsymbol{\mathbf{w}}^{\mathcal{N}_t}_t)] \\
    =& \frac{1}{\xi_t}(\mathbb{E}_{t} \mathbb{E}_{N_{t}}[\mathcal{L}(\boldsymbol{\mathbf{w}}^{N_{t}}_{t})] - \mathcal{L}(\boldsymbol{\mathbf{w}}_{t}^{0})) 
    \notag\\&- \eta_{t} \mathbb{E}_{N_{t}} \mathbb{E}_{t}\langle \mu_{t}-\nabla \mathcal{L}\left(\mathbf{w}_{t}^{0}\right), \nabla \mathcal{L}(\boldsymbol{\mathbf{w}}^{N_{t}}_{t})\rangle \\
    & +L_{\ell}\eta_{t}^{2}(L_g + Z_{g}^{2}Z_w) 
    \mathbb{E}_{N_{t}} \mathbb{E}_{t}\left\|\boldsymbol{\mathbf{w}}^{N_{t}}_{t} - \boldsymbol{\mathbf{w}}_{t}^{0}\right\|^{2} 
    \notag\\&+L_{\ell}\eta_{t}^{2}\left\|\mu_{t}-\nabla \mathcal{L}\left(\mathbf{w}_{t}^{0}\right)\right\|^{2}.
\stepcounter{equation}\tag{\theequation}\label{main-smooth-2} 
\end{align*}
\eqref{main-smooth-2} follows $\mathbb{E}\left[D_N - D_{N+1}\right] = \left(1 - \frac{1}{\Gamma}\right)(\mathbb{E}[D_N] - D_0)$ with Fubini's theorem.
Note that $\tilde{\boldsymbol{\mathbf{w}}}_t = \boldsymbol{\mathbf{w}}^{\mathcal{N}_t}_t$ and $\tilde{\boldsymbol{\mathbf{w}}}_{t-1} = \boldsymbol{\mathbf{w}}_t^0$. If we take expectation over all the randomness and denote it by $\mathbb{E}$, we get
\begin{align*}
    &\eta_{t}(1-L_{\ell} \eta_{t}) \mathbb{E}\|\nabla \mathcal{L}(\tilde{\boldsymbol{\mathbf{w}}}_{t})\|^{2} \\
    &= \frac{1}{\xi_t} \mathbb{E}\left[\mathcal{L}(\tilde{\boldsymbol{\mathbf{w}}}_{t})-\mathcal{L}(\tilde{\boldsymbol{\mathbf{w}}}_{t-1})\right] 
    - \eta_{t} \mathbb{E}\left\langle \mu_{t}-\nabla \mathcal{L}\left(\mathbf{w}_{t}^{0}\right), \nabla \mathcal{L}(\tilde{\boldsymbol{\mathbf{w}}}_{t})\right\rangle \\
    &\quad +L_{\ell}\eta_t^2(L_g+Z_g^2Z_w) \mathbb{E}\|\tilde{\boldsymbol{\mathbf{w}}}_{t}-\tilde{\boldsymbol{\mathbf{w}}}_{t-1}\|^{2} \\
    &\quad +L_{\ell} \eta_{t}^{2} \mathbb{E}\|\mu_{t}-\nabla \mathcal{L}\left(\mathbf{w}_{t}^{0}\right)\|^{2} \\
    &= \frac{1}{\xi_t} \mathbb{E}\left[\mathcal{L}(\tilde{\boldsymbol{\mathbf{w}}}_{t})-\mathcal{L}(\tilde{\boldsymbol{\mathbf{w}}}_{t-1})\right] \\
    &\quad - \frac{1}{\xi_t} \mathbb{E}\left\langle \mu_{t}-\nabla \mathcal{L}\left(\mathbf{w}_{t}^{0}\right), \tilde{\boldsymbol{\mathbf{w}}}_{t}-\tilde{\boldsymbol{\mathbf{w}}}_{t-1}\right\rangle \\
    &\quad +L_{\ell}\eta_t^2(L_g+Z_g^2Z_w) \mathbb{E}\|\tilde{\boldsymbol{\mathbf{w}}}_{t}-\tilde{\boldsymbol{\mathbf{w}}}_{t-1}\|^{2} 
    \notag\\
    &\quad + \eta_{t}(1+L_{\ell} \eta_{t}) \mathbb{E}\left\|\mu_{t}-\nabla \mathcal{L}\left(\mathbf{w}_{t}^{0}\right)\right\|^{2} 
    \stepcounter{equation}\tag{\theequation}\label{main-lemma-eta-E}\\
    &\leq \frac{1}{\xi_t} \mathbb{E}\left[\mathcal{L}(\tilde{\boldsymbol{\mathbf{w}}}_{t})-\mathcal{L}(\tilde{\boldsymbol{\mathbf{w}}}_{t-1})\right] \\
    &\quad + \frac{1}{2 \eta_{t} \xi_t} \left[-\frac{1}{\xi_t}+\eta_t^2(2L_g+2Z_g^2Z_w)\right] \mathbb{E}\|\tilde{\boldsymbol{\mathbf{w}}}_{t}-\tilde{\boldsymbol{\mathbf{w}}}_{t-1}\|^{2} \\
    &\quad + \frac{1}{\xi_t} \mathbb{E}\left\langle\nabla \mathcal{L}(\tilde{\boldsymbol{\mathbf{w}}}_{t}), \tilde{\boldsymbol{\mathbf{w}}}_{t}-\tilde{\boldsymbol{\mathbf{w}}}_{t-1}\right\rangle 
    \notag\\&\quad + \frac{\eta_{t}}{\xi_t} \mathbb{E}\|\nabla \mathcal{L}(\tilde{\boldsymbol{\mathbf{w}}}_{t})\|^{2} 
    + \frac{\eta_{t}}{\xi_t} \mathbb{E}\left\|\mu_{t}-\nabla \mathcal{L}\left(\mathbf{w}_{t}^{0}\right)\right\|^{2} \notag\\
    & \quad +L_{\ell}\eta_t^2(L_g+Z_g^2Z_w) \mathbb{E}\|\tilde{\boldsymbol{\mathbf{w}}}_{t}-\tilde{\boldsymbol{\mathbf{w}}}_{t-1}\|^{2} 
    \notag\\ & \quad + \eta_{t}(1+L_{\ell} \eta_{t}) \mathbb{E}\left\|\mu_{t}-\nabla \mathcal{L}\left(\mathbf{w}_{t}^{0}\right)\right\|^{2}, 
    \stepcounter{equation}\tag{\theequation}\label{main-lemma-neg-2eta-E}
\end{align*}
where \eqref{main-lemma-eta-E} follows from Lemma \ref{lemma-eta-E} and \eqref{main-lemma-neg-2eta-E} follows from \eqref{lemma-neg-2eta-E}. By rearranging all the terms,
\begin{align*}
    &\frac{1 - 2\eta_t^2(L_g + Z_g^2 Z_w)\xi_t - 2L_{\ell}\eta_t^3(L_g + Z_g^2Z_w) \xi_t^2}{2\eta_t \xi_t^2} \mathbb{E}\|\tilde{\boldsymbol{\mathbf{w}}}_{t}-\tilde{\boldsymbol{\mathbf{w}}}_{t-1}\|^{2} \\
    &\quad + \eta_{t}(1 - L_{\ell} \eta_{t} - \frac{1}{\xi_t}) \mathbb{E}\|\nabla \mathcal{L}(\tilde{\boldsymbol{\mathbf{w}}}_{t})\|^{2} \\
    &\leq \frac{1}{\xi_t}\mathbb{E}\left[\mathcal{L}(\tilde{\boldsymbol{\mathbf{w}}}_{t})-\mathcal{L}(\tilde{\boldsymbol{\mathbf{w}}}_{t-1})\right] 
    + \frac{1}{\xi_t} \mathbb{E}\left\langle\nabla \mathcal{L}(\tilde{\boldsymbol{\mathbf{w}}}_{t}), \tilde{\boldsymbol{\mathbf{w}}}_{t}-\tilde{\boldsymbol{\mathbf{w}}}_{t-1}\right\rangle \\
    &\quad + \eta_t(1+L_{\ell}\eta_t + \frac{1}{\xi_t})\mathbb{E}\|\mu_{t}-\nabla \mathcal{L}\left(\mathbf{w}_{t}^{0}\right)\|^2 .
    \stepcounter{equation}\tag{\theequation}\label{lemma-proof:label-name-so-hard-1}
\end{align*}
Now we can apply Young's inequality, that 
For any real numbers $x$ and $y$, and for any $\rho > 0$, the following holds:
$xy \leq \frac{x^2}{2\rho} + \frac{\rho}{2} y^2$
, 
on $ \mathbb{E}\left\langle\nabla \mathcal{L}(\tilde{\boldsymbol{\mathbf{w}}}_{t}), \tilde{\boldsymbol{\mathbf{w}}}_{t}-\tilde{\boldsymbol{\mathbf{w}}}_{t-1}\right\rangle$ 
using $x=\tilde{\boldsymbol{\mathbf{w}}}_t - \tilde{\boldsymbol{\mathbf{w}}}_{t-1}$, $y=\nabla \mathcal{L}(\tilde{\boldsymbol{\mathbf{w}}}_t)$, and $\rho = \frac{1-2\eta_t^2(L_g + Z_g^2Z_w)\xi_t - 2L_{\ell}\eta_t^3(L_g + Z_g^2Z_w) \xi_t^2}{\eta_t \xi_t}$ 
to get:
\begin{align*}
    &\frac{1}{\xi_t} \mathbb{E}\left\langle y, x \right\rangle 
    \leq \frac{1}{2 \rho}
    \mathbb{E}\|x\|^2 
    + \frac{\rho}{2}
    \mathbb{E}\|y\|^2. 
    \stepcounter{equation}\tag{\theequation}\label{lemma-proof:label-name-so-hard-2}
\end{align*}
Combining \eqref{lemma-proof:label-name-so-hard-1} and \eqref{lemma-proof:label-name-so-hard-2} and rearranging, we have:
\begin{align*}
    &\eta_{t} \Bigg( 
 -  \frac{1}{2[1 - 2\eta_t^2(L_g + Z_g^2Z_w)\xi_t - 2L_{\ell}\eta_t^3(L_g + Z_g^2Z_w) \xi_t^2]} \notag\\ &+ 1 - L_{\ell} \eta_{t}  - \frac{1}{\xi_t} \Bigg) \mathbb{E}\|\nabla \mathcal{L}(\tilde{\boldsymbol{\mathbf{w}}}_{t})\|^{2} \\
    \leq& \frac{1}{\xi_t}\mathbb{E}\left[\mathcal{L}(\tilde{\boldsymbol{\mathbf{w}}}_{t}) - \mathcal{L}(\tilde{\boldsymbol{\mathbf{w}}}_{t-1})\right] 
   \notag\\& + \eta_t(1 +L_{\ell}\eta_t + \frac{1}{\xi_t})\mathbb{E}\|\mu_{t}-\nabla \mathcal{L}\left(\mathbf{w}_{t}^{0}\right)\|^2 \\
    \leq& \frac{1}{\xi_t}\mathbb{E}\left[\mathcal{L}(\tilde{\boldsymbol{\mathbf{w}}}_{t}) - \mathcal{L}(\tilde{\boldsymbol{\mathbf{w}}}_{t-1})\right] 
     \notag\\&+ \eta_t(1 +L_{\ell}\eta_t + \frac{1}{\xi_t}) \left[2\psi^2 + 64 \sigma^2 \frac{ V}{\xi_t}\right],
    \stepcounter{equation}\tag{\theequation}\label{main-final-1}
\end{align*}
where \eqref{main-final-1} follows from Lemma \ref{lemma-error-bound}.
We want to choose $\eta_t$ such that $1-2\eta_t^2(L_g + Z_g^2Z_w)\xi_t - 2L_{\ell}\eta_t^3(L_g + Z_g^2Z_w) \xi_t^2 > 0$. Denoting $\Phi = L_g + Z_g^2Z_w$, we have
\begin{align*}
    1>1-2\eta_t^2\Phi \xi_t - 2L_{\ell}\eta_t^3\Phi \xi_t^2 &> 0 \\
    \frac{1}{2(1-2\eta_t^2\Phi \xi_t - 2L_{\ell}\eta_t^3\Phi \xi_t^2)} &> \frac{1}{2}
\end{align*}
Thus, we can select $\eta_t$ such that:
\begin{align*}
    1- L_{\ell} \eta_{t}  - \frac{1}{\xi_t} -  \frac{1}{2[1-2\eta_t^2\Phi \xi_t - 2L_{\ell}\eta_t^3\Phi \xi_t^2]}  & \geq \frac{1}{4} \stepcounter{equation}\tag{\theequation}\label{eq:eta_1}
    \\
    L_{\ell} \eta_{t}  + \frac{1}{\xi_t} +  \frac{1}{2[1-2\eta_t^2\Phi \xi_t - 2L_{\ell}\eta_t^3\Phi \xi_t^2]} &\leq \frac{3}{4}
\end{align*}
Next, we aim to satisfy the condition: 
\begin{align*}
    \text{(i)} \frac{1}{2} <  \frac{1}{2[1-2\eta_t^2\Phi \xi_t - 2L_{\ell}\eta_t^3\Phi \xi_t^2]}  & \leq \frac{5}{8} \\ 
    \text{(ii)} L_{\ell}\eta_t \leq \frac{1}{16} \qquad\qquad\qquad\quad \text{(iii)} \frac{1}{\xi_t} & \leq \frac{1}{16}
\end{align*}
From condition (i), we deduce:
\begin{align*}
    \eta_t^2\Phi \xi_t +L_{\ell}\eta_t^3 \Phi \xi_t^2 &\leq \frac{1}{10}\\
    \rightarrow \eta_t^2\Phi \xi_t \leq \frac{1}{20} \qquad &\& \qquad L_{\ell}\eta_t^3\Phi \xi_t^2 \leq \frac{1}{20} \\
    \Rightarrow \eta_t \leq \frac{1}{20^{1/2}\Phi^{1/2}\xi_t^{1/2}} \quad &\& \quad \eta_t \leq \frac{1}{20^{1/3}\Phi^{1/3}\xi_t^{2/3}L_{\ell}^{1/3}}
\end{align*}
Using conditions (ii) and (iii), we get:
\begin{align*}
    \eta_t \leq \frac{1}{16L_{\ell}} \qquad \& \qquad \xi_t \geq 16 \\
\end{align*}
We can then choose $\eta_t \leq \frac{1}{2\tau \xi_t^{2/3}}$, where $\tau = (L_{\ell}\Phi)^{1/3} = (L_{\ell}(L_g + Z_g^2Z_w))^{1/3}$, s.t.
\begin{align*}
    &\frac{1}{2(L_{\ell}\Phi)^{1/3}\xi_t^{2/3}} \\
    &\quad \leq \min\left\{\frac{1}{16}, 
    \frac{1}{20^{1/2}\Phi^{1/2}\xi_t^{1/2}}, 
    \frac{1}{20^{1/3}(\Phi L_{\ell})^{1/3}\xi_t^{2/3}}\right\}.
\end{align*}
To ensure the conditions are satisfied, we choose $\eta_t = \frac{1}{2\tau \xi_t^{2/3}}$ when $\xi_t \geq 16$ which satisfies \eqref{eq:eta_1}.
We can obtain the following from \eqref{main-final-1} and \eqref{eq:eta_1}:
\begin{align*}
    \frac{1}{4}\eta_{t}\mathbb{E}\|\nabla \mathcal{L}(\tilde{\boldsymbol{\mathbf{w}}}_{t})\|^{2} 
    &\leq \frac{1}{\xi_t}\mathbb{E}\left[\mathcal{L}(\tilde{\boldsymbol{\mathbf{w}}}_{t}) - \mathcal{L}(\tilde{\boldsymbol{\mathbf{w}}}_{t-1})\right] \\
    &\quad + 2\eta_t \left[2\psi^2 + 64 \sigma^2 \frac{ V}{\xi_t}\right].
\end{align*}
Replacing $\eta_t = \frac{1}{2\tau \xi_t^{2/3}}$ and rearranging, we have
\begin{align*}
    \mathbb{E}\|\nabla \mathcal{L}(\tilde{\boldsymbol{\mathbf{w}}}_{t})\|^{2} 
    &\leq 4\Bigg[\frac{1}{\xi_t\eta_t} \mathbb{E}\left[\mathcal{L}(\tilde{\boldsymbol{\mathbf{w}}}_{t}) - \mathcal{L}(\tilde{\boldsymbol{\mathbf{w}}}_{t-1})\right] \\
    &\quad + 2 \left[2\psi^2 + 64 \sigma^2 \frac{ V}{\xi_t}\right]\Bigg]\\
    &\leq 4\Bigg[\frac{2 \tau \mathbb{E}\left[\mathcal{L}(\tilde{\boldsymbol{\mathbf{w}}}_{t}) - \mathcal{L}(\tilde{\boldsymbol{\mathbf{w}}}_{t-1})\right]}{\xi_t^{\frac{1}{3}}} \\
    &\quad + 4\psi^2 + 128 \sigma^2 \frac{ V}{\xi_t}\Bigg].
\end{align*}
Telescoping over $t = 1,2,...,T$ with a constant batch size $\xi_t$, we have for $\tilde{\boldsymbol{\mathbf{w}}}_{a}$ uniformly sampled from $\{\tilde{\boldsymbol{\mathbf{w}}}_t\}_{t=1}^{T}$:
\begin{align*}
    \mathbb{E}\|\nabla \mathcal{L}(\tilde{\boldsymbol{\mathbf{w}}}_{a})\|^{2} 
    &\leq 4\Bigg[\frac{2 \tau \mathbb{E}\left[\mathcal{L}(\tilde{\boldsymbol{\mathbf{w}}}_T) - \mathcal{L}(\tilde{\boldsymbol{\mathbf{w}}}_{0})\right]}{T\xi_t^{\frac{1}{3}}} \\
    &\quad + 4\psi^2 + 128 \sigma^2 \frac{ V}{\xi_t}\Bigg]\\
    &\leq \frac{8 \tau \left[\mathcal{L}(\tilde{\boldsymbol{\mathbf{w}}}^*) - \mathcal{L}(\tilde{\boldsymbol{\mathbf{w}}}_{0})\right]}{T\xi_t^{\frac{1}{3}}} \\
    &\quad + 16\psi^2 + 512 \sigma^2 \frac{ V}{\xi_t},
\end{align*}
which completes the proof.
\end{proof}

\begin{table}[!htbp]
    \centering
    \small
    \renewcommand{\arraystretch}{1.2} 
    \setlength{\aboverulesep}{0pt}
    \setlength{\belowrulesep}{0pt}
    \setlength{\abovetopsep}{0pt}
    \setlength{\belowbottomsep}{0pt}
    \caption{Experiment parameters.}
    \begin{tabular}{l|l}
        \toprule
        \textbf{Parameters} & \textbf{Values} \\
        \midrule
        Total number of agents $K$ & 10 \\
        Fraction of malicious agents $\alpha$ & 20\% \\ 
        Training round $T$ & 2000 \\
        Steps & 1500 \\
        Vehicle length & 3 meters \\
        Number of hidden layers & 3 \\
        Batch size $\xi_t$ & 512 \\
        Discount factor $\gamma$ & 0.9995 \\
        Step size $\eta_t$ & 0.001 \\
        Mini-batch size $b_t$ & 32 \\
        Scaling factor $\lambda$ & 10 \\
        Threshold $\psi$ & 1 \\
        
        \bottomrule
    \end{tabular}
    \label{tab: parameter}
\end{table}

\begin{table}[h!]
\centering
\small
\caption{Neural network architecture.}
\label{tab:nn_architecture}
\begin{tabular}{|l|c|}
\hline
\textbf{Layer}             & \textbf{Shape} \\ \hline
Input                      & 8 \\ \hline
Fully Connected + Sigmoid  & 8 × 256 \\ \hline
Fully Connected + Sigmoid  & 256 × 256     \\ \hline
Fully Connected + Sigmoid  & 256 × 256     \\ \hline
Fully Connected + Sigmoid  & 256 × 1 \\ \hline
Output                     & 1 \\ \hline
\end{tabular}
\end{table}



\section{Details of Poisoning Attacks}
\label{sec:attacks}

\myparatight{Trim attack~\cite{fang2020local}}
%
This attack is specifically designed for the Trimmed-mean~\cite{Yin18} and Median~\cite{Yin18} aggregation rules.
In the Trim attack, the attacker meticulously designs the local gradients on the malicious agents so that the aggregated gradient after the attack significantly diverges from the one before the attack.

\myparatight{Random attack~\cite{cao2022mpaf}}
%
Each malicious agent samples a Gaussian vector from a Gaussian distribution with a mean of 0 and a variance of 10000, then transmits this vector to the server.

\myparatight{History attack~\cite{cao2022mpaf}}
%
In this attack, the attacker replaces the local gradients on the malicious agents with scaled versions of their previous values, thereby skewing the aggregated gradient to disrupt the aggregation process.

\myparatight{MPAF attack~\cite{cao2022mpaf}}
%
%
Each malicious agent calculates and amplifies the difference between a selected base gradient and the current aggregated gradient using a large scaling factor.

\myparatight{FTI attack~\cite{zhang2024poisoning}}
%
%
The FTI attack operates similarly to the MPAF attack. Each malicious agent calculates the difference between the attacker-selected base gradient and the current aggregated gradient but applies different weights to the base gradient and current aggregated gradient to avoid detection.


\myparatight{MinMax attack~\cite{shejwalkar2021manipulating}}
%
%
This is an attack that is agnostic to the aggregation rule. In this attack, the attacker crafts the malicious local gradients so that the maximum distance between the malicious and benign local gradients is no greater than the maximum distance between any two benign local gradients.

\myparatight{MinSum attack~\cite{shejwalkar2021manipulating}}
%
%
Similar to the MinMax attack, the MinSum attack is also aggregation rule-independent. The attacker precisely designs the malicious local gradients so that the total distance between the malicious and benign local gradients does not exceed the maximum sum of distances between any two benign local gradients.

\myparatight{Adaptive attack~\cite{shejwalkar2021manipulating}}
%
%
%
The adaptive attack represents the worst-case scenario, where it is assumed that the attacker is aware of the aggregation rule employed by the server (i.e., \alg in our paper) and has knowledge of the local gradients on all agents. In this attack, the attacker crafts the local gradients on the malicious agents to maximize the distance between aggregated gradient with and without the attacks.

\section{Details of Comparison Aggregation Rules}
\label{sec:comparison}



\myparatight{FedAvg~\cite{mcmahan2017communication}}
%
%
The server calculates the aggregated gradient by averaging the local gradients from all agents.

\myparatight{Median~\cite{Yin18}}
%
%
The Median rule computes the median of gradients across each dimension coordinate-wise.

\myparatight{Trimmed-mean (Trim)~\cite{Yin18}}
%
%
Trim is another coordinate-wise robust aggregation rule. For each dimension, the server removes the largest $b$ and smallest $b$ values, then averages the remaining elements, where $b$ is the trim parameter.

\myparatight{Krum~\cite{blanchard2017machine}}
%
%
%
%
In the Krum method, the server selects the local gradient with the smallest sum of distances to its nearest neighbors.

\myparatight{FoolsGold~\cite{fung2018mitigating}}
%
%
%
In FoolsGold, the server computes a similarity matrix using cosine similarity among all agents' gradients and adjusts the weights according to these similarities. The weighted local gradients are then aggregated to create the aggregated gradient.

\myparatight{FABA~\cite{xia2019faba}}
%
%
%
In FABA, the server employs an iterative method to eliminate outlier local gradients from the system. Specifically, in each iteration, the server removes the local gradients with the largest distance from the current average, repeating this process until $m$ gradients are deleted, where $m$ represents the number of malicious gradients, assumed to be known by the server.

\myparatight{FLTrust~\cite{cao2020fltrust}} 
%
%
%
In this aggregation rule, it is assumed that the server can compute a gradient. The server will accept local gradients from agents that are positively correlated with this reference gradient.

\myparatight{FLAIR~\cite{sharma2023flair}}
%
%
In FLAIR, the server assigns reputation scores to agents based on their historical behavior. These scores are then used to adjust the weight of each agent's contribution to the aggregated gradient.

\myparatight{FedPG-BR (FedPG)~\cite{fan2021fault}}
The FedPG uses the mean of median technique to choose a gradient that is in the center of the selected set. After performing a comparison between the mean of median gradient, the local gradients that fall within a threshold are considered benign.



\myparatight{FLAME~\cite{nguyen2022flame}}
%
%
%
%
The FLAME rule uses a clustering algorithm to identify and remove highly impactful gradients. Additionally, it applies a dynamic weight-clipping approach and noise injection technique to further minimize the influence of malicious gradients.

\myparatight{DeepSight~\cite{rieger2022deepsight}}
%
%
%
%
%
DeepSight utilizes a combined approach of deep model analysis and clipping in its filtering technique to detect possible poisoning attacks.

\section{Extensive Experiments and Ablation Study}
\label{more_exp}

In this section, we provide further evaluations on our \algns.

\myparatight{Different variants of \algns}
We explore several variants of \alg used in the AD systems, as summarized in Table~\ref{tab:variant}.
\begin{list}{\labelitemi}{\leftmargin=1em \itemindent=-0.08em \itemsep=.2em}
    \item \textbf{Variant I (Var. I):} In Line \ref{server_mom} of Algorithm~\ref{alg:our}, instead of calculating $\mu_{t}^{\mathcal{S}}$ as the model closest to the mean of the set $\mathcal{S}$, we compute it using the median~\cite{Yin18} of $\mathcal{S}$. That is, $\mu_{t}^{\mathcal{S}} = \operatorname{argmin}_{\mu_t^{(\tilde{k})}}\|\mu_t^{(\tilde{k})} - \operatorname{median}(\mathcal{S})\|$, where $\tilde{k} \in \mathcal{S}$.

    \item \textbf{Variant II (Var. II):} In this variant, we remove Lines \ref{server_mom} and \ref{server_model_filter} from Algorithm~\ref{alg:our} entirely. Instead, we use the set $\mathcal{S}$ computed in Line \ref{server_mom_filter} directly as the benign model set $\mathcal{W}_t$. By doing so, we ensure that the number of selected models remains greater than half of the total participating agents, without further adjustments.

    \item \textbf{Variant III (Var. III):} Here, the server uses the model calculated in Line \ref{server_mom}, {$\mu_{t}^{\mathcal{S}} = \operatorname{argmin}_{\mu_t^{(\tilde{k})}}\|\mu_t^{(\tilde{k})} - \operatorname{mean}(\mathcal{S})\|$}, as the aggregated gradient. This version removes Line \ref{server_model_filter} and uses the computed $\mu_{t}^{\mathcal{S}}$ for the SVRG process.

    \item \textbf{Variant IV (Var. IV):} In Line \ref{server_mom_filter}, gradients are not filtered; instead, all gradients are incorporated into the set $\mathcal{S}$.
\end{list}

\begin{table}[!h]
\renewcommand{\arraystretch}{1.2}
\addtolength{\tabcolsep}{-2.88pt}
\centering
\caption{Performance on different variants of \algns.}
\scriptsize
\begin{tabularx}{\columnwidth}{l|*{5}{>{\centering\arraybackslash}X}}
\toprule
Attack & Var. I & Var. II & Var. III & Var. IV & \textbf{\alg} \\
\midrule
No attack     & 100.0\% & 100.0\% & 100.0\% & 100.0\% & 100.0\% \\
\hline
Trim attack     & 100.0\% & 23.17\% & 100.0\% & 40.15\% & 100.0\% \\
\hline
Random attack   & 100.0\% & 100.0\% & 100.0\% & 100.0\% & 100.0\% \\
\hline
History attack  & 100.0\% & 29.38\% & 100.0\% & 100.0\% & 100.0\% \\
\hline
MPAF attack     & 100.0\% & 29.77\% & 22.29\% & 100.0\% & 100.0\% \\
\hline
FTI attack      & 100.0\% & 28.69\% & 28.03\% & 100.0\% & 100.0\% \\
\hline
MinMax attack   & 100.0\% & 26.53\% & 19.72\% & 100.0\% & 100.0\% \\
\hline
MinSum attack   & 92.45\% & 15.10\% & 17.61\% & 23.19\% & 100.0\% \\
\hline
Adaptive attack & 90.81\% & 20.96\% & 14.63\% & 31.93\% & 100.0\% \\
\bottomrule
\end{tabularx}
\label{tab:variant}
\end{table}

Each of the four variants exhibits varying \textit{levels} of resilience to attacks. As expected, \alg clearly outperforms the others, maintaining the highest no-collision rates and demonstrating strong resistance to all attacks. Interestingly, Var. I performs similarly to the original \alg but remains vulnerable in specific cases like the MinSum and Adaptive attacks, where its performance drops to around 90\%. Var. II and Var. III fail across most attacks, showing a sharp decline in performance with rates under 30\%, leaving them unable to defend against adversarial interference. Var. IV performs better than Var. II and Var. III but still fails in some cases, such as achieving 31.93\% under the Adaptive attack.
From this ablation study, we confirm that Lines \ref{server_mom_filter}, \ref{server_mom}, and \ref{server_model_filter} in Algorithm~\ref{alg:our} are crucial to the robustness of \alg against poisoning attacks. Line \ref{server_mom_filter} ensures the reliability of the set $\mathcal{S}$ by filtering out malicious gradients. Line \ref{server_mom} establishes a reliable benchmark by selecting a center gradient from the benign set $\mathcal{S}$, where the choice of $\operatorname{mean}(\mathcal{S})$ leads to slightly better performance than Var. I with $\operatorname{median}(\mathcal{S})$. Line \ref{server_model_filter} is the most critical one, as it applies strict filtering to each local model based on $\mu_{t}^{\mathcal{S}}$ and $\mu_{t-1}$. Missing any of these design philosophies in \alg results in a significant vulnerability to attacks.

\myparatight{Impact of total number of agents}
Fig.~\ref{fig:count} provides a comparative analysis of various aggregation methods in defending against poisoning attacks as the number of participating RL agents increases, the fraction of malicious agents is set to 20\%. 
As depicted in Fig.~\ref{fig:count}a, traditional aggregation methods such as Median and Trim fail to defend the system across the entire range of agent counts, with no-collision rates remaining around or below 20\% as the number of agents scales from 10 to 40, indicating frequent failures in collision avoidance under attack conditions. In contrast, our \alg demonstrates exceptional robustness, maintaining a 100\% no-collision rate regardless of the number of agents used in the system. This trend of superior performance is consistent across all other attack scenarios, as illustrated in Figs.~\ref{fig:count}b-h. Such desired resilience of \alg can be attributed to its advanced filtering mechanisms, which effectively identify and neutralize the influence of malicious gradients, ensuring the reliability of the aggregated gradient even in large-scale agent systems. 

\begin{figure*}[!htbp]
    \centering
    \includegraphics[width=0.9\textwidth]{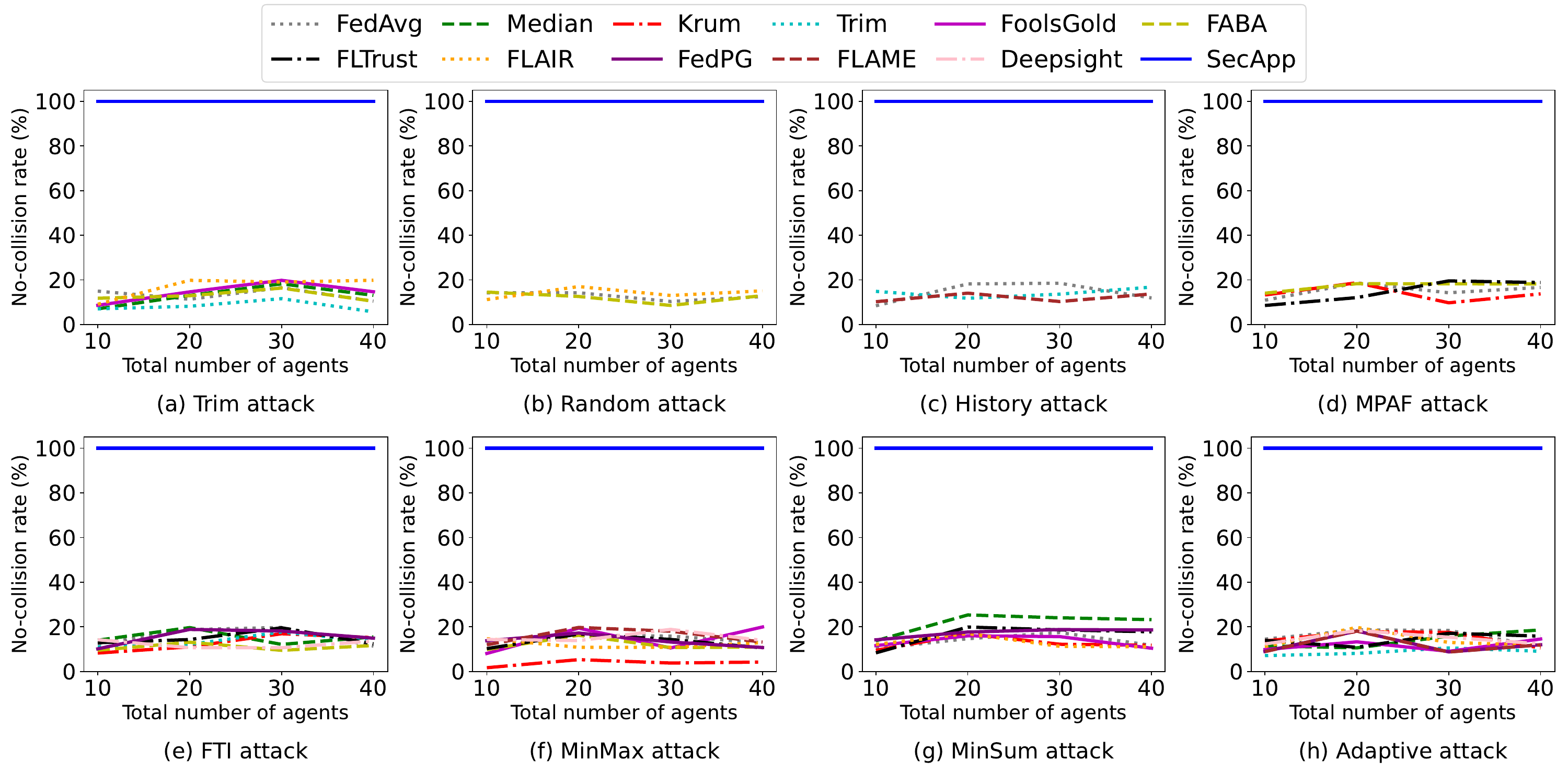}
    \caption{Impact of total number of agents.}
    \label{fig:count}
\end{figure*}

\begin{table}[h]
\renewcommand{\arraystretch}{1.2}
\centering
\caption{Performance on different threshold $\psi$.}
\footnotesize 
\begin{tabularx}{\columnwidth}{l|*{4}{>{\centering\arraybackslash}X}}
\toprule
Attack & 0.1 & 1 & 5 & 10 \\
\midrule
No attack & 82.14\% & 100.0\% & 100.0\% & 100.0\% \\
\hline
Trim attack & 24.31\% & 100.0\% & 100.0\% & 46.97\% \\
\hline
Random attack & 23.62\% & 100.0\% & 100.0\% & 100.0\% \\
\hline
History attack & 43.85\% & 100.0\% & 100.0\% & 100.0\% \\
\hline
MPAF attack & 34.95\% & 100.0\% & 100.0\% & 100.0\% \\
\hline
FTI attack & 46.30\% & 100.0\% & 100.0\% & 100.0\% \\
\hline
MinMax attack & 43.01\% & 100.0\% & 100.0\% & 100.0\% \\
\hline
MinSum attack & 32.13\% & 100.0\% & 34.60\% & 28.92\% \\
\hline
Adaptive attack & 36.34\% & 100.0\% & 100.0\% & 35.88\% \\
\bottomrule
\end{tabularx}
\label{tab:psi}
\end{table}

\myparatight{Impact of different threshold $\psi$}
Table~\ref{tab:psi} evaluates the performance of \alg under various attack scenarios with different values of the filtering threshold $\psi$ in line \ref{server_mom_filter} of Algorithm~\ref{alg:our}. A lower $\psi$ enforces more stringent filtering, which can degrade performance by excluding too many benign gradients, as seen with attacks like Trim and Random attacks. For instance, when $\psi = 0.1$, \alg achieves only 82.14\% in attack-free scenarios and fails under all attack conditions. With $\psi = 1$, there is a marked improvement, with \alg reaching 100.0\% performance in most cases. This demonstrates that a moderate value of $\psi$ strikes an optimal balance between filtering and retaining sufficient benign gradients for an effective aggregation. At $\psi = 5$, the algorithm still performs well overall, but specific attacks, such as MinSum Attack, cause performance drops. This suggests that overly lenient filtering allows malicious gradients to infiltrate the gradient aggregation process. While at $\psi = 10$, performance declines significantly in scenarios such as 46.97\% and 35.88\% under Trim and Adaptive attacks, respectively, highlighting that excessive leniency compromises the defense mechanism. Therefore, selecting an appropriate filtering threshold $\psi$ is another critical factor in ensuring that \alg maintains robustness against various attacks.

\myparatight{Performance against backdoor attacks}
In this section, we demonstrate that our proposed \alg effectively protects the AD system against targeted attacks, including backdoor attacks.
The three backdoor attacks, BACKDOORL~\cite{wang2021backdoorl}, BadRL~\cite{cui2024badrl}, and MARNet~\cite{chen2022marnet}, are particularly effective in RL-based scenarios due to their adaptive strategies that exploit the decentralized nature of learning. 
Additional implementation details on these three backdoor attacks can be found in Appendix~\ref{sec:Backdoorattacks}.
%
Note that the three backdoor attacks were initially designed for traditional RL. To apply them to AD tasks, each malicious agent trains its local model using deceptive states and rewards, enabling it to execute malicious actions when specific triggers are detected.


Table~\ref{tab:backdoor} shows the effectiveness of various defense mechanisms against three distinct backdoor attacks. These attacks are engineered to subtly manipulate the model’s performance on specific tasks while remaining undetected on others, making them particularly challenging to counteract. 
The values in Table~\ref{tab:backdoor} are shown in the form of ``no-collision rate / attack success rate''.
A larger no-collision rate and a smaller attack success rate indicate a more effective defense.
Table~\ref{tab:backdoor} shows that during the BACKDOORL attack, the no-collision rate of defenses such as FLAIR significantly decreases.
In stark contrast, \alg maintains a flawless 100.0\% no-collision rate across all evaluated metrics, demonstrating its exceptional proficiency in detecting and neutralizing these stealthy threats. Similarly, under the BadRL attack, many defense mechanisms suffer a significant decline in no-collision rate, with some, such as FLAME and FABA, plummeting to 12.42\% and 14.68\%, respectively. The attack success rate remains below 5\% for \alg against all backdoor attacks. 
This indicates that \alg not only mitigates the effects of backdoor attacks but also maintains the model's performance on its main control tasks.

\begin{table}[!t]
\renewcommand{\arraystretch}{1.5}
\centering
\caption{Performance of \alg across various backdoor attacks. The results are in the form of ``no-collision rate / attack success rate''.}
\scriptsize
\begin{tabularx}{\columnwidth}{l|*{3}{>{\centering\arraybackslash}X}}
\toprule
Defense & BACKDOORL attack & BadRL attack & MARNet attack \\
\midrule
FedAvg            & 23.68\% \textbf{/} 89.45\% & 27.42\% \textbf{/} 81.67\% & 13.50\% \textbf{/} 90.12\% \\
\hline
Median          & 100.0\% \textbf{/} 5.62\% & 100.0\% \textbf{/} 9.89 \% & 100.0\% \textbf{/} 6.17\% \\
\hline
Trim            & 15.67\% \textbf{/} 56.21\% & 100.0\% \textbf{/} 8.47\% & 100.0\% \textbf{/} 9.28\% \\
\hline
Krum            & 28.45\% \textbf{/} 77.84 \% & 100.0\% \textbf{/} 3.11\% & 21.57\% \textbf{/} 62.95\% \\
\hline
FoolsGold       & 100.0\% \textbf{/} 4.23\% & 100.0\% \textbf{/} 5.56\% & 100.0\% \textbf{/} 7.34\% \\
\hline
FABA            & 100.0\% \textbf{/} 8.56\% & 14.68\% \textbf{/} 73.29\% & 100.0\% \textbf{/} 5.89\% \\
\hline
FLTrust         & 100.0\% \textbf{/} 9.78\% & 100.0\% \textbf{/} 6.43 \% & 100.0\% \textbf{/} 7.16\% \\
\hline
FLAIR           & 22.89\% \textbf{/} 66.54\% & 100.0\% \textbf{/} 8.99\% & 100.0\% \textbf{/} 6.78\% \\
\hline
FedPG           & 100.0\% \textbf{/} 7.12\% & 100.0\% \textbf{/} 8.76\% & 100.0\% \textbf{/} 5.33\% \\
\hline
FLAME           & 100.0\% \textbf{/} 6.89\% & 12.42\% \textbf{/} 69.84\% & 100.0\% \textbf{/} 9.47\% \\
\hline
Deepsight       & 100.0\% \textbf{/} 8.23\% & 100.0\% \textbf{/} 9.05\% & 100.0\% \textbf{/} 6.78\% \\
\hline
\textbf{\alg}            & \textbf{100.0\% / 3.45\%} & \textbf{100.0\% / 3.12\%} & \textbf{100.0\% / 2.89\%} \\
\bottomrule
\end{tabularx}
\label{tab:backdoor}
\end{table}

\myparatight{Adaptive computation of scaling factor $\lambda$}
%
In our experiments, we default to using a fixed scaling factor \(\lambda\) as shown in Eq.~(\ref{fix_scaling}). Considering the mobile vehicular environment, here we investigate a moving average approach to dynamically compute \(\lambda\), allowing it to adjust to the variability in the gradients.
Specifically, let \(\lambda_t\) denote the scaling factor at training round \(t\). Then \(\lambda_t\) can be calculated as 
$
\lambda_t = \alpha \cdot \lambda_{t-1} \cdot \left\| \mu_{t}^{\mathcal{S}} - \mu_{t-1} \right\| + (1 - \alpha) \cdot \lambda_{0},
$
where the initial scaling factor $\lambda_{0}=10$ and weighing parameter $\alpha=0.2$. 
The historical component $\lambda_{t-1}$ ensures that the adjustment of scaling factor is not too abrupt.
The no-collision rates of \alg under both Trim and Adaptive attacks are both 100.0\% when the adaptive computation of the scaling factor is applied.

\myparatight{Impact on decision-making sequences}
Fig.~\ref{fig:example} illustrates the impact of poisoning attacks on vehicle behavior in a five-car scenario within the FRL framework, where the FedAvg aggregation rule is used. 
``Delta Distance'' refers to the distance between two vehicles.
When all agents are benign, as shown in Fig.~\ref{fig:sub1}, the FRL system operates without any collision incidents, the vehicles can maintain smooth trajectories with safe inter-vehicle distances, reflecting the effectiveness of the secure aggregation method in ensuring safety. 
Specifically, FedAvg rule ensures that the distance between consecutive vehicles remains within a safe range of 4.8 to 20 meters throughout the evaluation, with vehicle velocities decreasing consistently as they approach their destinations, ultimately coming to a safe halt.
However, when the FRL model is compromised by the Adaptive attack, as demonstrated in Fig.~\ref{fig:sub2}, the integrity of the system is significantly degraded. When vehicle 5, the leading vehicle, begins its abrupt deceleration, the compromised model causes the RL-controlled vehicle 4 to continue accelerating instead of braking, leading to a collision. The distance between vehicle 4 and vehicle 5 rapidly diminishes, even crossing below zero, indicating an overlap and hence a collision. The failure of vehicle 4 to decelerate, in contrast to the leading vehicle's abrupt stop, underscores the disruption caused by the malicious gradients. This scenario validation highlights the catastrophic impact of poisoning attacks, where the compromised model fails to maintain effective coordination among vehicles, leading to unsafe behaviors and potential accidents.

\begin{figure*}[!htbp]
    \centering
    \begin{subfigure}[b]{0.9\textwidth}
        \centering
        \includegraphics[width=\linewidth]{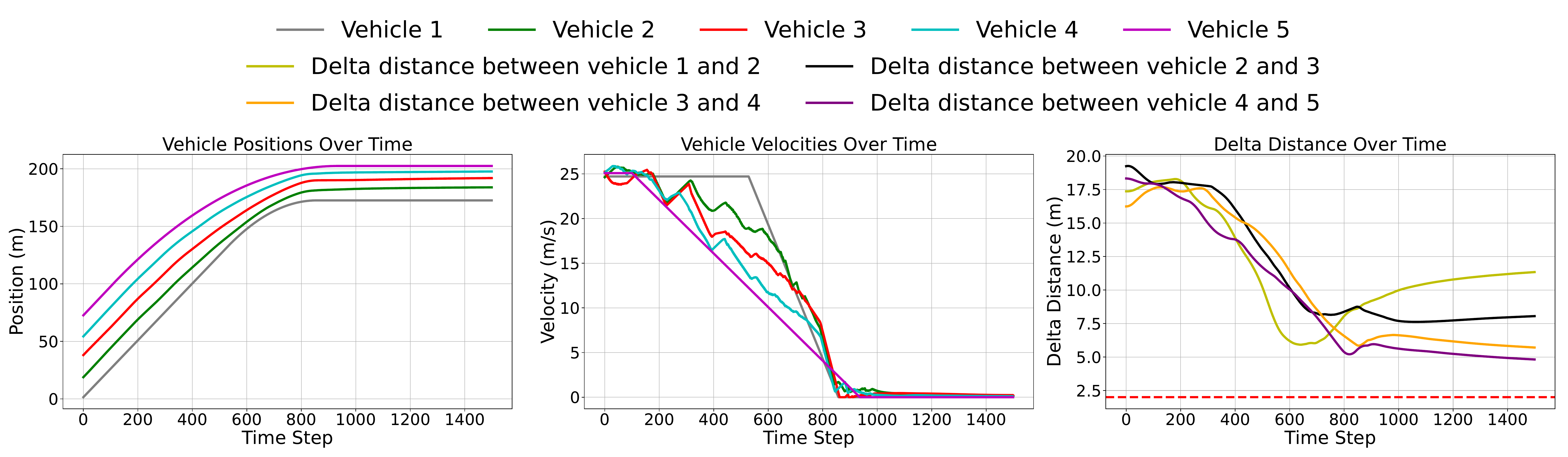}
        \caption{Vehicle behaviors when the FedAvg aggregation rule is used, under no attack setting.}
        \label{fig:sub1}
    \end{subfigure}
    \hfill
    \begin{subfigure}[b]{0.9\textwidth}
        \centering
        \includegraphics[width=\linewidth]{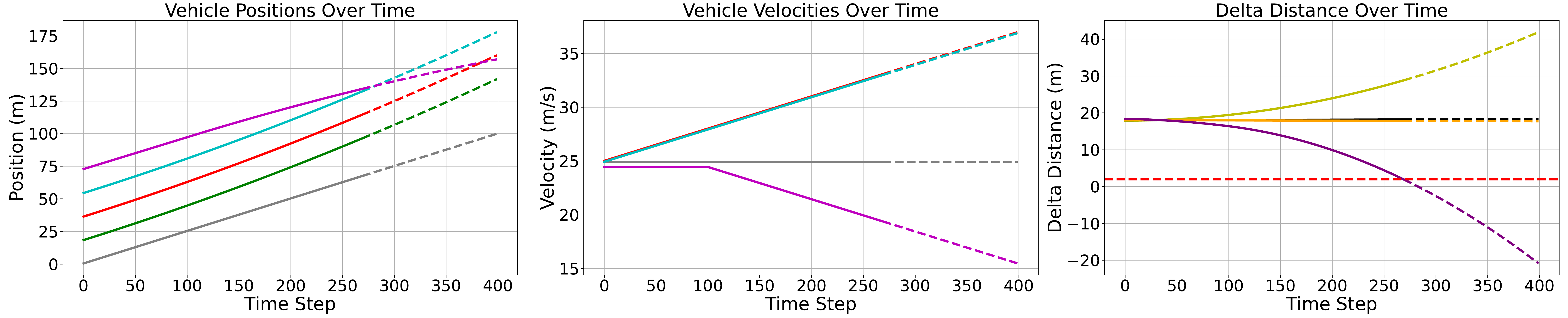}
         \caption{Vehicle behaviors when the FedAvg aggregation rule is used, under Adaptive attack setting.}
        \label{fig:sub2}
    \end{subfigure}
    \caption{Attacked RL algorithm implementation and verification.}
    \label{fig:example}
\end{figure*}

\begin{table}
\centering
\caption{Performance on different model sizes.}
\footnotesize 
\begin{tabularx}{\columnwidth}{l|*{3}{>{\centering\arraybackslash}X}}
\toprule
Attack & 3 layers & 4 layers & 5 layers \\
\midrule
No attack     & 100.0\% & 100.0\% & 100.0\% \\
\hline
Trim attack     & 100.0\% & 100.0\% & 100.0\% \\
\hline
Random attack   & 100.0\% & 100.0\% & 100.0\% \\
\hline
History attack  & 100.0\% & 100.0\% & 100.0\% \\
\hline
MPAF attack     & 100.0\% & 100.0\% & 100.0\% \\
\hline
FTI attack      & 100.0\% & 100.0\% & 100.0\% \\
\hline
MinMax attack   & 100.0\% & 100.0\% & 100.0\% \\
\hline
MinSum attack   & 100.0\% & 100.0\% & 100.0\% \\
\hline
Adaptive attack & 100.0\% & 100.0\% & 100.0\% \\
\bottomrule
\end{tabularx}
\label{tab:layer}
\end{table}

\myparatight{Impact on numbers of neural network layers}
Table~\ref{tab:layer} examines the impact of varying the number of neural network hidden layers on the performance of \algns. Interestingly, the performance remains remarkably consistent, with \alg achieving 100.0\% accuracy across models with 3, 4, and 5 hidden layers, regardless of the attack type. This indicates that \alg is highly scalable and can effectively handle models of different complexities without a loss in defensive capability. The consistent performance across different model sizes suggests that \algns's defense mechanisms are robust and not overly sensitive to the underlying model architecture, making it a versatile solution for a wide range of FRL applications.


\begin{table*}[t]
\renewcommand{\arraystretch}{1.5}
\centering
\caption{Cache hit rate in edge caching scenario.}
\footnotesize 
\begin{tabularx}{\textwidth}{l|*{12}{>{\centering\arraybackslash}X}}
\toprule
Attack & FedAvg & Median & Trim & Krum & FoolsGold & FABA & FLTrust & FLAIR & FedPG & FLAME & Deepsight & \textbf{\alg} \\
\midrule
No attack       & 77.29\% & 77.68\% & 77.89\% & 77.17\% & 76.55\% & 76.94\% & 77.83\% & 77.64\% & 77.41\% & 77.12\% & 76.92\% & \textbf{77.84\%} \\
\hline
Trim attack     & 53.98\% & 76.99\% & 77.78\% & 76.13\% & 44.69\% & 58.74\% & 77.39\% & 50.29\% & 76.66\% & 77.23\% & 76.81\% & \textbf{77.17\%} \\
\hline
Random attack   & 48.37\% & 76.92\% & 77.11\% & 52.47\% & 74.86\% & 56.92\% & 77.54\% & 54.11\% & 78.79\% & 78.93\% & 79.36\% & \textbf{77.66\%} \\
\hline
History attack  & 47.12\% & 77.47\% & 58.42\% & 77.39\% & 77.16\% & 46.88\% & 77.02\% & 76.95\% & 77.48\% & 54.62\% & 77.89\% & \textbf{76.98\%} \\
\hline
MPAF attack     & 44.67\% & 59.18\% & 77.29\% & 76.82\% & 76.84\% & 57.13\% & 50.41\% & 77.36\% & 76.97\% & 77.18\% & 77.08\% & \textbf{77.63\%} \\
\hline
FTI attack      & 49.12\% & 54.78\% & 49.89\% & 58.92\% & 77.11\% & 54.62\% & 48.79\% & 76.92\% & 59.21\% & 77.39\% & 55.44\% & \textbf{77.14\%} \\
\hline
MinMax attack   & 51.02\% & 77.58\% & 77.33\% & 47.98\% & 44.52\% & 49.69\% & 55.17\% & 56.12\% & 58.23\% & 54.21\% & 55.78\% & \textbf{77.44\%} \\
\hline
MinSum attack   & 49.62\% & 56.11\% & 77.62\% & 48.49\% & 53.77\% & 76.89\% & 45.23\% & 52.34\% & 56.77\% & 77.48\% & 76.98\% & \textbf{77.16\%} \\
\hline
Adaptive attack & 58.42\% & 50.17\% & 76.77\% & 55.12\% & 48.32\% & 77.38\% & 51.93\% & 58.27\% & 46.71\% & 77.14\% & 49.13\% & \textbf{77.29\%} \\
\bottomrule
\end{tabularx}
\label{tab:cache}
\end{table*}

\section{SecApp Used in Edge Caching Scenarios}
\label{caching_setting}

In the context of emerging next-generation cellular networks, the dense deployment of small-cell base stations (BSs) connected to data centers via low-bandwidth, high-latency backhaul links is anticipated. Each BS is equipped with a cache unit that stores frequently requested content, with the aim of enhancing network performance by minimizing latency and reducing communication costs. The primary objective in such an environment is to maximize the cache hit rate, ensuring that content is readily available in the BS cache, and to minimize peak traffic loads on the backhaul links, which occur when content must be fetched from the data center. Cache management strategies are designed to dynamically replace less frequently accessed content with data that is more likely to be requested in the future, based on predictive models such as the Zipf distribution, which describes the frequency of content requests. We leverage this digital twin environment to assess our proposed defense strategy. The cache hit rate is the percentage of requests that are successfully served from the cache, reducing the need to retrieve data from a slower, more distant source. 
It is used as a key performance metric, where a higher hit rate signifies better defense performance; a random cache strategy typically achieves a hit rate of around 40-50\%.

The edge caching strategy in cellular networks can be modeled as a MDP, defined by the tuple \( \{ \mathcal{S}, \mathcal{A}, \mathcal{R}, \mathcal{T}, s_0 \} \).

\begin{itemize}
    \item \textbf{State Space (\( \mathcal{S} \))}: Each state \( s \in \mathcal{S} \) is represented as a tuple \( (\ell, t_\ell, f_\ell) \), where \( \ell \) denotes the base station (BS) index, \( t_\ell \) is the time elapsed since the last cache update, and \( f_\ell \) indicates the frequency of content requests.

    \item \textbf{Action Space (\( \mathcal{A} \))}: The action space \( \mathcal{A} \) defines the possible caching operations. The server can either skip a request (\( a = 0 \)) or allocate a specific cache slot for the requested content \( d_t \). 

    \item \textbf{Reward Function (\( \mathcal{R}(s, a) \))}: The immediate reward reflects the improvement in cache performance, quantified by the increase in cache hits between consecutive requests.

    \item \textbf{State Transition Probabilities (\( \mathcal{T} \))}: The state transition probabilities \( \mathcal{T} \) define the likelihood of moving between states based on the selected actions.

    \item \textbf{Initial State (\( s_0 \))}: The initial state \( s_0 \) represents the system’s configuration at the start of the process, where all cache units are assumed to be empty.
\end{itemize}


We consider a network consisting of 10 BSs, each equipped with a local cache capable of storing 150 items. Note that in the edge caching setting, a single RL agent controls one BS. Each BS serves 8 edge devices, and the service areas of edge devices may overlap, allowing an edge device to connect to up to 2 BSs simultaneously. Cache decisions are made dynamically, taking into account the load of the BSs and the current cache status of the involved BSs. The request frequencies \( f_\ell \) follow a Zipf distribution with parameter \( p = 0.8 \), which reflects real-world request patterns where a small subset of content is requested much more frequently than the rest.
In our experiments, we set the values of $\lambda$ and $\psi$ to 10 and 5, respectively.
20\% of the agents are compromised.




In edge caching scenarios, our proposed \alg exhibits superior performance across various types of attacks, consistently maintaining high cache hit rate that outperforms other defense mechanisms. This is consistent with its performance in autonomous vehicle scenarios.
For instance, under the Random attack, \alg achieves a cache hit rate of 77.66\%, outperforming FoolsGold, which only reaches 74.86\%. This suggests that \alg is particularly effective at countering the effects of randomized adversarial behavior, likely due to its robust aggregation mechanism that filters out malicious gradients. Similarly, when confronted with the Adaptive attack—which dynamically adjusts its strategy to circumvent defenses—\alg still maintains a high cache hit rate of 77.29\%, outperforming Krum and Trim, which drop to 55.12\% and 76.77\%, respectively. A similar trend is observed in the autonomous driving scenario, where \alg consistently shows better performance across comparable attack types, underscoring its generalizability and effectiveness in diverse FRL contexts. 

\section{Details of Backdoor Attacks}
\label{sec:Backdoorattacks}

\myparatight{BACKDOORL~\cite{wang2021backdoorl}}BACKDOORL attack embeds hidden policies within the local models of individual agents, which are activated under specific conditions to trigger undesirable behaviors. Despite these triggers, the overall model maintains strong performance on non-targeted tasks, making the attack difficult to detect during the aggregation process.

\myparatight{BadRL~\cite{cui2024badrl}}BadRL attack targets key decision-making pathways by selectively poisoning critical model components. This targeted approach ensures that only a few gradients are compromised, allowing the attack to evade detection while still achieving its malicious goals by subtly influencing the collective decision-making process. 

\myparatight{MARNet~\cite{chen2022marnet}}This attack exploits the cooperative dynamics of multi-agent RL by coordinating malicious behaviors among agents. Although each agent may appear benign individually, their coordinated actions significantly degrade overall system performance, making the attack particularly challenging to identify and mitigate.

\section{Limitation and Future Works}
\label{sec:limitation}
One limitation of our \alg is that its robustness has only been demonstrated in two representative safety-critical systems: autonomous driving and edge caching. The performance of \alg in other critical systems, such as unmanned aerial vehicle (UAV) assisted public safety~\cite{kuang2020real,tisdale2009autonomous}, remains unknown. Additionally, \alg is currently constrained to server-assisted FRL, where a central server helps manage the global training process. In many real-world serverless safety-critical systems, such as employing fully-decentralized RL~\cite{busoniu2008comprehensive,tan1993multi}, there is no central server for coordination, and each agent exchanges information with its neighbors. Therefore, the robustness of \alg in such fully-decentralized scenarios is yet to be determined.
Addressing these limitations opens several promising research directions. Future extensions of \alg could incorporate graph-based policy alignment and local-area aggregation to enable reliable learning without centralized coordination. In addition, integrating environment-aware abstractions such as twin–derived surrogate models may help transfer robustness guarantees across domains.